\documentclass[12pt]{article}
\usepackage{cite}
\usepackage{graphicx}
\usepackage{caption}
\usepackage{refstyle}
\usepackage{amsmath}
\usepackage{amsthm}
\usepackage{amsmath}
\usepackage{amsfonts}
\usepackage{amssymb}
\usepackage{bbm}
\usepackage{amsmath,array}
\usepackage{mathtools}
\usepackage{slashbox}
\usepackage[all]{xy}
\usepackage{mathtools}
\usepackage[utf8]{inputenc}
\usepackage[english]{babel}
\usepackage[T3,T1]{fontenc}
\usepackage{fancyhdr}
\DeclareSymbolFont{tipa}{T3}{cmr}{m}{n}
\DeclareMathAccent{\invbreve}{\mathalpha}{tipa}{16}

\usepackage[left=1.1cm,right=1.1cm,top=2.9cm,bottom=2.9cm]{geometry}
\usepackage{setspace}
\newtheorem{theorem}{Theorem}

\title{Capacity Region of Two-users Weak Gaussian Interference Channel}

\usepackage{authblk}
\author[1]{Amir K. Khandani\thanks{E\&CE Department, University of Waterloo, Waterloo, Ontario, Canada, khandani@uwaterloo.ca.}}

\begin{document}

\vspace{-5cm}
\maketitle

\vspace{-1cm}
\begin{abstract}
Capacity region of two-users (2-users) weak Gaussian Interference Channel (GIC) has been solved only in some special  cases. The problem is complex since  knowledge of input distributions is needed in order to express the underlying mutual information terms in closed forms, which in turn should be optimized over selection of input distributions and the associated power and bandwidth allocation. In addition, the optimum solution may require dividing the available resources (time, bandwidth and power) among several 2-users GIC (called ``constituent 2-users GIC'', or simply ``constituent regions'', hereafter) and apply time-sharing among them  to form their upper concave envelope,  and thereby enlarge the capacity region.   The current article, in most parts, focuses  on a single constituent 2-users GIC, meaning that the constraints on resources are all satisfied with equality.   This does not result in any loss of generality because a similar solution technique is applicable to any constituent 2-users GICs.   
This article shows that, by relying on a different,  intuitively straightforward,  interpretation of the underlying optimization problem, one can determine the encoding/decoding strategies in the process of computing the optimum solution (for a single constituent 2-users GIC).  This is based on gradually moving along the boundary of the capacity region in infinitesimal steps, where the solution for the end point in each step is constructed and optimized relying on the solution at the step's starting point. This approach enables proving Gaussian distribution is optimum over the entire boundary, and also allows finding simple closed form solutions describing different parts of the capacity region.   The solution for each constituent 2-users GIC coincides with the optimum solution to the well known Han–Kobayashi (HK) system of constraints with i.i.d.  (scalar) Gaussian inputs.
Although the article is focused on {\em 2-users weak} Gaussian interference channel, the proof for optimality of Gaussian distribution is independent of the values of cross gains, and thereby is universally applicable to  strong, mixed and Z interference channels, as well as to GIC with more than two users. In addition, the procedure for the construction of boundary is applicable for arbitrary cross gain values, by re-deriving various conditions that have been established in this article assuming $a<1$ and $b<1$.

\end{abstract}
{\bf Note:} The current arXiv submission differs from the earlier one (submitted on November 25th, 2020) in: (i) An Introduction section and a  brief literature survey are added. (ii) The proof for the upper bound (presented in Appendix~\ref{APP2}) is made more clear. (iii) Section 3.1.4, which provided an alternative proof for part of what is presented in Appendix~\ref{APP2}, is removed. Sections 1.4/1.5 are revised. 

\section{Introduction} \label{sec1}
The problem of 2-users Gaussian Interference Channel (GIC) shown in Fig.~\ref{GICfig} is solved in some  special cases. In particular, the capacity region for the class of weak 2-users GIC (corresponding to $a<1$ and $b<1$ in Fig.~\ref{GICfig}) has been an open problem. The primary question is that of finding the capacity achieving input distribution (which is known to be Gaussian in all cases with a known solution). The complexity of the problem has motivated researchers to study the problem relying on Gaussian input distributions. Even under the assumption of Gaussian inputs, the resulting problem is complex as it involves optimization over many parameters, where the form of the functions to be optimized depends on the encoding and decoding strategies deployed by transmitters and receivers, respectively. Consequently, the form of the function to be optimized depends on the optimum solution itself. This entangles two problems: (i) the problem of finding the closed form mutual information terms which fundamentally change depending on  encoding and decoding strategies, and (ii) finding the optimum power allocation for optimum encoding and decoding strategies.  A brute force approach would require solving the latter problem over all possible encoding and decoding strategies, then compare the results to select the optimum solution. This is intractable due to the complexity of the closed form expressions, and sheer number of parameters to be optimized.  

This article introduces an alternative view based on the concept of nested optimality which serves as an engine to construct the optimum boundary in infinitesimal steps. In each step, the optimum encoding/decoding strategy and its associated power allocation are found for the end point on the step based on the solution at its  starting point. This is achieved  by expressing nested optimality in the form of differential equations.   
This viewpoint enables: (i) proving optimality of Gaussian distribution, using a layered code-book, at each transmitter, (ii)  proving optimality of HK rate region, and (iii) finding simple closed form solutions describing various shapes (depending on channel parameters and available power levels) of the boundary of the capacity region.   

In this article, it is assumed the two transmitters use the entire bandwidth at all times, and power available to each transmitter is entirely used.  A complete solution would require computing the upper concave envelope by time sharing among a number of solutions of the type presented in this article, through optimizing the resource allocation (time/bandwidth and power). Using the results presented in this article, deriving the upper concave envelope would be rather straightforward, although lengthy, and is left for future work. 

\subsection{Literature Survey}

The problem of Gaussian interference channel has been the subject of numerous outstanding prior works, paving the way to the current point and moving beyond. A subset of these works, reported in \cite{GIC1} to \cite{GIC38}, are briefly discussed in this section.  A more complete and detailed literature survey will be provided in subsequent revisions of this article

Reference~\cite{GIC1} discusses degraded Gaussian interference channel (degraded means one of the two receivers is a degraded version of the other one) and presents multiple bounds and achievable rate regions. Reference~\cite{GIC2} studies the capacity of 2-users GIC for the class of strong interference and shows the capacity region is at the intersection of two MAC regions, consistent with the current article.  Reference~\cite{GIC3}  establishes optimality for two extreme points in the achievable region of the general 2-users GIC. \cite{GIC3} also proves that the class of degraded Gaussian interference channels is equivalent to the class of Z (one-sided) interference channels. 

References~\cite{GIC4} to \cite{GIC6} present achievable rate regions for interference channel. In particular, \cite{GIC4} presents the well-known Han-Kobayashi (HK) achievable rate region.  HK rate region coincides with all results derived previously (for Gaussian 2-users GIC), and is shown to be optimum for the class of weak 2-users GIC in the current article. References \cite{GIC7}\cite{New1} have further studied the HK rate region.  \cite{New1} shows that HK achievable rate region is strictly sub-optimum for a class of discrete interference channels. 

References~\cite{GIC9} to \cite{GIC15} have studied the problem of outer bounds for the interference channel.  Among these, \cite{GIC11}\cite{GIC12}\cite{GIC13} have also provided optimality results in some special cases of weak 2-users GIC.  

References~\cite{GIC16}\cite{GIC17} have studied the problem of interference channel with common information. References~\cite{GIC18} to \cite{GIC20} have studied the problem of interference channel with cooperation between transmitters and/or between receivers. 
References~\cite{GIC21}\cite{GIC22} have studied the problem of interference channel with side information. Reference~\cite{GIC23} has studied the problem of interference channel assuming cognition, and  reference~\cite{GIC24} has studied the  problem assuming cognition, with or without secret messages. 

Reference \cite{GIC25} has found the capacity regions of vector Gaussian interference channels for classes of very strong and aligned strong interference. \cite{GIC25} has also generalized some known results for 
sum-rate of scalar Z interference, noisy interference, and mixed interference to the case of vector channels. Reference \cite{GIC26} has addressed the sum-rate  of the parallel Gaussian
interference channel. Sufficient conditions are derived
in terms of problem parameters (power budgets and channel coefficients) such
that the sum-rate can be realized by independent transmission across sub-channels while treating interference as noise, and corresponding optimum power allocations are computed. 
Reference \cite{GIC27} studies a Gaussian interference network where each message is encoded by a single transmitter and is aimed at a single receiver. Subject to feeding back the output from receivers to their corresponding transmitter, efficient strategies
are developed based on the discrete Fourier transform signaling. 

 Reference \cite{GIC28} computes the capacity of interference channel within one bit. 
References \cite{GIC29}\cite{GIC30} study the impact of interference in GIC.  \cite{GIC30} shows that treating interference as noise in 2-users GIC achieves the closure of the capacity
region to within a constant gap, or within a gap that scales as O(log(log(.)) with signal to noise ratio. Reference \cite{GIC31} relies on game theory  to define
the notion of a Nash equilibrium region of the interference
channel, and characterizes the Nash equilibrium region for: (i) 2-users
linear deterministic interference channel in exact form, and (ii) 2-users GIC within 1 bit/s/Hz in an approximate form. 

Reference \cite{GIC32} studies the problem of 2-users GIC based on a sliding window superposition coding scheme. 

References \cite{GIC33} and  \cite{GIC34}, independently, introduce the new concept of non-unique decoding as an intermediate alternative to ``treating interference as noise'', or ``canceling interference''.  Reference~\cite{GIC35} further studies the concept on non-unique decoding and proves that (in all reported cases) it can be replaced by a special joint unique decoding without penalty.

Reference~\cite{GIC36} studies the degrees of freedom of the K-user
Gaussian interference channel, and, subject to a mild sufficient condition on the channel
gains, presents an expression  for the degrees of freedom of the scalar interference
channel as a function of the channel matrix.

Reference~\cite{GIC37} studies the problem of state-dependent Gaussian interference channel, where two receivers are affected by scaled versions of the same state. The state
sequence is  (non-causally) known at both transmitters, but not at receivers. 
Capacity results are established (under certain conditions on channel parameters) in the very strong, strong, and weak interference regimes. For the weak regime, the sum-rate is computed. Reference~\cite{GIC38} studies the problem of state-dependent Gaussian interference channel under the assumption of 
correlated states, and characterizes  (either fully or partially) the capacity region  or the sum-rate under various channel parameters.

%

\subsection{Notations and Definitions}
$A\leftrightarrow B$ means  $A$ and $B$ are interchanged, i.e., $A$ is replaced with $B$, and at the same time, $B$ is replaced with $A$. 
$A \nearrow B$ means $A<B$ increases, approaching $B$. $A\searrow B$ means $A>B$ decreases, approaching $B$. 
Notation $X \triangleleft	Y$, equivalently $Y \triangleright	X$, means $X$ is a degraded version of $Y$. Notation $\alpha \lnapprox \beta$ means $\alpha$ is smaller than $\beta$, but $\alpha$ and $\beta$ are almost equal. 
The breve sign $\breve .$ is used to provide a generic representation of variables in order to avoid confusion with the same variable names used in  deriving article's main results. 
 $AW\!G\!N(\breve{P},\breve{\sigma}^2)$ refers to an Additive White Gaussian Noise (AWGN) channel of power budget $\breve{P}$ and noise power $\breve{\sigma}^2$. 
 $M\!AC(\breve{P}_1,\breve{P}_2,\breve{\sigma}^2)$ refers to a Multiple Access Channel (MAC) with two users of power budgets $\breve{P}_1, \breve{P}_2$, and AWGN of power  $\breve{\sigma}^2$. 

Two-users GIC is composed of two MACs at $Y_1$ and $Y_2$.  The difference with an ordinary MAC is that $V_1$ is do-not-care (considered as noise) at $Y_2$ and $V_2$ is do-not-care (considered as noise) at $Y_1$. In other words, the rate tuples considered relevant at $Y_1$ and $Y_2$ are 
$(R_{U_1},R_{U_2},R_{V_1})$ and 
$(R_{U_1},R_{U_2},R_{V_2})$, respectively. This means, in the context of GIC,  the multiple-access channel formed at $Y_1$ is projected on 
$(U_1,U_2,V_1)$ and the multiple-access channel formed at $Y_2$ is projected on $(U_1,U_2,V_2)$.
To simplify notations, in general, same terminologies are used to refer to an actual MAC region, and to its projection on a sub-space. To emphasize the difference, in Appendix~\ref{APP2},  notations $\overline{M\!AC_1}$ and $\overline{M\!AC_2}$ are used to refer to the actual MAC regions with four rate values, i.e., $(U_1,U_2,V_1,V_2)$. 

$G\!I\!C(\breve{P}_1,\breve{P}_2,a,b,\breve{\sigma}_1^2,\breve{\sigma}_2^2)$ refers to a 2-users weak Gaussian Interference Channel (GIC) with cross gains $a<1$, $b<1$, power budgets $\breve{P}_1,\breve{P}_2$ and noise powers  $\breve{\sigma}_1^2,\breve{\sigma}_2^2$, respectively (see Fig.~\ref{GICfig})\footnote{This article follows a nested viewpoint regarding the optimality. For this reason, Fig.~\ref{GICfig} is shown for the more general case that power of AWGN at receiver 1 and receiver 2 are equal to $\breve{\sigma}_1^2$ and $\breve{\sigma}_2^2$, respectively. However, this article is concerned with the solution to the weak 2-users GIC where the power of the AWGN at  both receivers is equal to 1.}. 
In some cases, this notation is simplified to $G\!I\!C(\breve{P}_1,\breve{P}_2)$. 

Notation $R(.)$ is used to refer to rate as a function of power. When applicable, subscripts 1 and 2 are used to specify user 1 and user 2, respectively.  $R_{ws}$ is used to refer to weighted sum-rate,  $R^{(public)}$ and $R^{(private)}$ are used to specify the public and private rates, respectively, $R^{(merged)}$ is used to specify the rate of two layers (code-books) upon being merged into a single layer (code-book). 
$R_{sum}(\breve{\mbox{MAC}})$ and $R_{partial-sum}(\breve{\mbox{MAC}})$ are used to refer to the sum-rate and partial sum-rate of the multiple access channel $\breve{\mbox{MAC}}$. 

A Gaussian layer of power $\breve{P}$ refers to a Gaussian random code-book of power $\breve{P}$. 
A Gaussian layer of power $\breve{P}$ is called continuous if it can be considered as the superposition of independent Gaussian layers of powers $\forall \breve{P}_1$, $\forall \breve{P}_2$ satisfying  $\breve{P}=\breve{P}_1+\breve{P}_2$.
 This means, two independent, continuous  Gaussian layers of powers $\breve{P}_1$ and $\breve{P}_2$, if superimposed, can be merged into a single continuous Gaussian layer of power $\breve{P}_1+\breve{P}_2$ if the rate of the merged layer is equal to $R^{(merged)}(\breve{P}_1+\breve{P}_2)=R(\breve{P}_1)+R(\breve{P}_2)$. This means, a merged continuous Gaussian layer, by relying on a single Gaussian code-book, can achieve the same rate that the two superimposed continuous Gaussian layers can achieve using  superposition coding and successive decoding. 

As will be discussed in Section~\ref{assumptions}, two types of Gaussian layers (code-books) are required in achieving the boundary of the capacity region; a private layer for each user (which is decoded only by its intended receiver) and one or more public layer(s) for each user (which is/are decoded by both receivers). 

The boundary of the capacity region, formed by optimizing $R_{ws}=R_1+\mu R_2$, is studied in two parts. The lower part corresponds to $\mu\leq 1$ and the upper part corresponds to $\mu\geq 1$. This article focuses on the lower part of the boundary, and the corresponding results and statements can be applied to the upper part by relying on exchanges in expression~\ref{changes} below: 
   \begin{eqnarray} \label{changes}
 P_1 & \leftrightarrow & P_2\\ \nonumber  
 \hat{P}_1& \leftrightarrow & \hat{P}_2 \\ \nonumber  
 T_1& \leftrightarrow & T_2 \\ \nonumber  
 a & \leftrightarrow & b \\ \nonumber  
 \geq & \leftrightarrow & \leq \\ \nonumber 
 > & \leftrightarrow & < .
  \end{eqnarray}

Notations $\rho$ and $\theta$ are used to show how the power is divided between private code-book and public code-book(s) for each of the two users. For user 1, the total power of public code-book(s) is equal to $\rho P_1$ and, consequently, the power of private code-book is $(1-\rho)P_1$. For user 2, the total power of public code-book(s) is equal to $\theta P_2$ and, consequently, the power of private code-book is $(1-\theta)P_2$. To simplify expressions, we use notations $\hat{P}_1$ and $\hat{P}_2$ to refer to the powers of private code-books for user 1 and user 2, respectively. 

To construct the boundary, an infinitesimal amount of power, denoted by $\delta P_1$ and $\delta P_2$ (for user 1 and user 2, respectively) are moved from the public to private part, or vice versa. The corresponding changes in $R_1$, $R_2$ and $R_{ws}$ are denoted as $\delta R_1$, $\delta R_2$ and $\delta R_{ws}$, respectively. Indeed, it will be shown that in computing $\delta R_1$, $\delta R_2$, $\delta R_{ws}$, the bulk of public message(s) can be ignored. This means if a $\delta$-layer is moved from public to private, the leftover public part can be ignored, and if it is moved from private to public, the original public (prior to the move) can be ignored. This means changes in rate values are functions of amount of power allocated to private layers. To emphasis this point, notations $\delta \hat{R}_1$ and $\delta \hat{R}_2$ are used to represent the changes in $R_1$ and $R_2$ due to moving of a $\delta$-layer. 

This article focuses on weak 2-users GIC where the power of the AWGN at either receiver is equal to 1. In this case, the two-tuple of public code-books, $(U_1,U_2)$, at $Y_1$ is subject to a Gaussian noise of power 
$\sigma_1^2=(1-\rho)P_1+b(1-\theta)P_2+1$. Likewise, the two-tuple of public code-books, $(U_1,U_2)$,  at $Y_2$ is subject a Gaussian noise of power  $\sigma_2^2=a(1-\rho)P_1+(1-\theta)P_2+1$. In computing the capacity region of $G\!I\!C(P_1,P_2,a,b,1,1)$, the rate of public layer(s) will be at the intersection of two MAC channels, 
a $M\!AC(\rho P_1,a\theta P_2,\sigma_1^2)$ formed at $Y_1$ and  
$M\!AC(b\rho P_1,\theta P_2,\sigma_2^2)$ formed at $Y_2$. Joint decoding is needed to realize particular points at the intersection of  MAC$_1$ and  MAC$_2$ (see Fig~\ref{GICfig}) which fall on the sum-rate front of one of the two MACs.
As an alternative to joint decoding, layering\footnote{Layering means using more than a single continuous Gaussian code-book.} of public message(s) can be used to achieve points on the sum-rate front of either of the two MACs.  On the other hand, the private code-book of each user is a continuous layer, unless it is trivially divided into multiple continuous sub-layers.

  \begin{figure}[htp]
   \centering
   \includegraphics[width=0.6\textwidth]{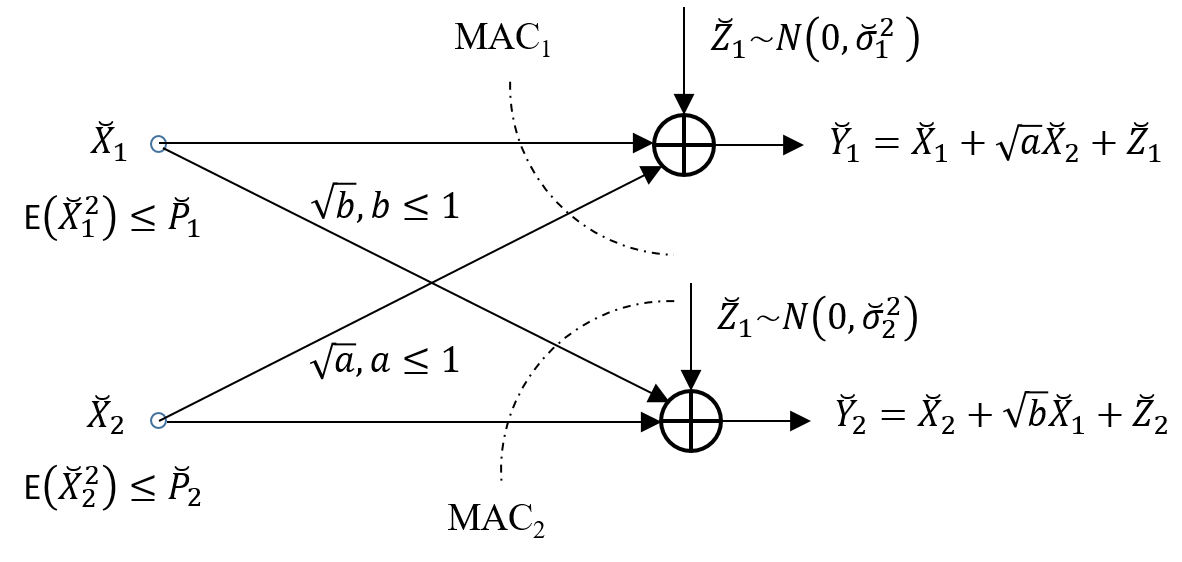}
   \caption{Channel model for two-users Gaussian interference channel.}
   \label{GICfig}
 \end{figure}

\subsection{Article Flow/Structure} 
To describe the flow/structure of the article, first, some definitions are needed. 

In a 2-users GIC, resources (time, bandwidth, power) may be divided among multiple instances of 2-users GIC, where, through optimizing the resource allocation and (optimized) time-sharing among corresponding solutions, the capacity region can be potentially enlarged. The enlarged capacity region is known as the {\em upper concave envelope}.  Hereafter, each region obtained for fixed allocation of resources is referred to as a {\em constituent region}. In cases that some of time/frequency dimensions are allocated to one of the two users alone, the solution to 2-users GIC will be the same as a simple AWGN channel. For this reason, we focus on a constituent region where the entire bandwidth is used simultaneously by both users. In this case, the two power budgets determine the shape of the capacity region and its associated encoding/decoding strategies. For a given quota of resource values (power, time, bandwidth),  we refer to a solution for which constraints on resources are all satisfied with equality as a {\em strategy}. The part of the strategy that is under the control of each of the two transmitters is referred to as the {\em state} of the corresponding transmitter. 
 
Section~\ref{sec1} presents some preliminary discussions that pave the way for subsequent sections containing the main results. Section~\ref{sec2} describes the capacity region of 2-users weak GIC assuming scalar Gaussian inputs and a fixed strategy, i.e., a single constituent region for which power constraints for the two users are satisfied with equality, bandwidth is entirely  shared by the two users, and both users simultaneously transmit at all times. It will be later shown in Section~\ref{sec3.3} that in such a case using a vector Gaussian as input would result in i.i.d. components for the vector, meaning that an i.i.d. Gaussian code-book is optimum.  

Section~\ref{sec3} proves the optimality of Gaussian inputs in achieving capacity.  The proof is presented for a single constituent region. The proof regarding Gaussian optimality  is first expressed in terms of scalar inputs  in Section~\ref{sec3.2}, and is then generalized to vector inputs in Section~\ref{sec3.3}.  In Section~\ref{sec3.3}, it is shown that for a single strategy, using a vector input is not required. In other words, a vector input, once optimized, will have i.i.d. Gaussian components.

Finally, Section~\ref{sec4} discusses the relationship to Han Kobayashi rate region, and
 presents simplified forms for the underlying (active) HK constraints. It is established that the HK rate region overlaps with the capacity region of 2-users weak GIC for a single strategy (fixed power budgets and fixed allocation of time/bandwidth).  This is achieved by pursuing the same decoding orders as suggested by capacity achieving arguments to determine redundant constraints among the HK set.  It is then shown that the region captured by simplified HK expressions is the same as the corresponding 2-users GIC capacity region.   

Article focuses on the lower part of the boundary, i.e., $\mu\leq 1$ in $R_{ws}=R_1+\mu R_2$. Various conditions (expressed in terms of parameters $a$,$b$,$P_1$ and $P_2$) are selected such that the lower part of the boundary posses its richest possible feature set. This means, the lower part of the boundary passes through all points $A$,$B$,$C$,$D_1$,$D_2$,$D_3$ in Fig.~\ref{move-D} (to be defined later in details) and finally becomes tangent to the sum-rate front at point $S$. Derived conditions can be used to determine scenarios (in terms of conditions on parameters $a$,$b$,$P_1$ and $P_2$) that the lower part of the boundary intersects with the sum-rate front prior to going through all the points mentioned above.  

\subsection{An Alternative Viewpoint} \label{aav}
This article presents a bounding region for the 2-users GIC based on the intersection of MAC1 and MAC2 formed at $Y_1$ and $Y_2$ (see Fig.~\ref{GICfig}), respectively (see Appendix~\ref{APP2}). This article shows that one can remove the gap between the 2-users GIC capacity region and its bounding region, while at the same time, (optimally) enlarge the bounding region. {\em This means, one can view this article as aiming to realize  the  bounding region, and to optimally enlarge it, while removing any gap between the capacity region of the underlying 2-users weak GIC and its associated bounding region.} This alternative viewpoint helps in verifying/following different arguments presented in this article. See Appendix~\ref{APP2} for relevant discussions.
 
 \subsection{Four Simplifying Properties and their Justifications} \label{assumptions}
 
  This article relies on the following four properties. Note that the proof for optimality of Gaussian distribution does not depend on any of these properties.  The proof only relies on assuming that moving along the boundary is achieved by increasing the rate of user 2 at the cost of reducing the rate of user 1 (for the lower part of the boundary), or vice versa (for the upper part of the boundary).   

 \begin{enumerate} 
 \item There are two types of messages, private messages that are decoded only by their respective (intended) receiver, and public messages that are decoded by both receivers. In other words, the power available to each of the two users is entirely allocated to the construction of its corresponding public/private message(s). 

 \item Moving along the boundary is achieved by changing the allocation of power between the private code-book and the public code-book of one of the two users, or of both users, in normalized infinitesimal steps. By dividing the boundary into segments with continuous slope, the steps that involve modifying the power allocation of both users (called composite steps) are expressed as two consecutive simple steps, where each simple step involves modifying the power allocation of a single user. 
In each simple step, a fixed, infinitesimal amount of power is moved from the quota of power for public  message(s) of the user to the quota of power for private message of the same user, or vice versa.

\item Joint decoding, if needed, will be for decoding of public messages. In other words, private message intended for receiver 1 will be decoded successively, upon removal of public message(s), while considering interference from transmitter 2 as noise. Likewise, private message intended for receiver 2 will be decoded successively (considering interference from transmitter 1 as noise) once the public message(s) is/are decoded/removed at receiver 2. Note that Appendix~\ref{APP2} shows a joint encoding in GIC has to follow the trivial case of (a joint encoding) based on a layered code-book using $X_1=U_1+V_1$ and $X_2=U_2+V_2$ with $U_1$, $V_1$, $U_2$, $V_2$ being independent Gaussian. 

\item No explicit converse to the proof  is presented, because: (i) proof for the optimality of Gaussian inputs is independent of any assumptions regarding the structure  of the capacity region, and (ii) it is shown (see Appendix~\ref{APP2}) that for Gaussian inputs,  the intersections of the two MACs formed at $Y_1$ and $Y_2$ contains the capacity region, and this article shows this bound is achievable, and coincides with the capacity region of 2-users GIC, as well as with the (achievable) HK region. 
  \end{enumerate}  

Next, these properties are justified\footnote{Note that the first three properties are known features of  the upper bounding region (corresponding to the intersection of MACs) discussed  in Appendix~\ref{APP2}, and can be justified relying on the alternative perspective presented  in Sections~\ref{aav}.} 

{\bf First property:} 
Consider the optimum signaling scheme.  Without loss of generality, let us assume  $X_1$, and thereby $Y_1$ and $Y_2$, are dependent on  $\mathsf{X}$, which is an input  to the  message encoder at transmitter 1. However, no information is assigned to the selection of $\mathsf{X}$. By assigning a rate equal to $\min \left( I(\mathsf{X};Y_1), (\mathsf{X};Y_2)\right)$ to $\mathsf{X}$,  $R_{ws}$ will be increased by $\min \left( I(\mathsf{X};Y_1), (\mathsf{X};Y_2)\right)>0$ contradicting the assumed optimality. 

In summary,  what is being sent by transmitter 1 or transmitter 2 (which is at the cost of using some of the available power) should be entirely decode-able in at least of one of the two receivers. This means, some messages are decoded at both receivers and some messages are decoded at one of the two receivers, only. This is possible only if there are two types of messages, called public and private.    $\square$

{\bf Second  property:} Upon establishing that the code-books are Gaussian, and noting the discussions in Section~\ref{aav} and Appendix~\ref{APP2}, one can rely on  a layered structure (for the superposition of independent random code-books representing the public and the private messages for each of the two users), which justifies the second property. $\square$

{\bf Third property:} As mentioned earlier, the intersection of MAC1 and MAC2 forms an upper bound to the capacity region, which will be shown in this article to be achievable, coinciding with the capacity region. Section~\ref{aav} expresses the overall objective in terms of realizing and enlarging the intersection of the two MAC regions, instead of achieving the capacity of 2-users GIC.  As a result, the constraints of the two MACs govern the optimum rate tuple.  Joint decoding in MAC is needed when, given a fixed sum-rate, one needs to provide a trade-off between corresponding rate components. This is used for achieving a point, other than the corner points, on the sum-rate front of the underlying MAC by increasing one of the rate values at the cost of decreasing the other one. In 2-users GIC,  the rate tuple defining MAC1 at $Y_1$ includes $(R_{U_1},R_{U_2},R_{V_1})$, likewise, the rate tuple defining MAC2 at $Y_2$ includes  $(R_{U_1},R_{U_2},R_{V_2})$. As a result, any possible intersection between the capacity regions of the two MACs is limited to the shared sub-space, i.e., sub-space spanned by $(R_{U_1},R_{U_2})$. This means, if joint decoding is needed, it will involve a trade-off between $R_{U_1}$ and $R_{U_2}$.
A question remains if there may be any benefits in joint decoding of public and private messages sent to one of the receivers. Without loss of generality, let us focus on user 1. There is no need to trade off the rate of public message vs. the rate of the private message for user 1, as only their sum contributes to $R_1$, which appears as a standalone component in $R_{ws}=R_1+\mu R_2$.   On the other hand, the rate of the private message of user 1 does not impact the MAC formed at $Y_2$ (only its power matters, because it will not be decoded). This entails, the point of interest on MAC1  may only involve joint decoding of the two public messages, and the private message $V_1$ will be decoded once the public messages are decoded and their effects  are removed at $Y_1$.  Note that, unlike private message, the rates of public messages affect $R_{ws}$ as a weighted sum, i.e., $R_{ws}^{(public)}=R_1^{(public)}+\mu R_2^{(public)}=R_{U_1}+\mu R_{U_2}$, which is another indication that (potentially) joint decoding of $U_1$ and $U_2$ may be required. Similar arguments apply to MAC2. The above arguments are expressed more accurately in mathematical terms in Appendix~\ref{APP2}. Section~\ref{sec-inter} studies cases that joint decoding cannot be avoided.  The alternative viewpoint presented in Section \ref{aav}, which is based on expressing/formulating  the overall objective in terms of achieving/optimizing the intersection of MAC1 and MAC2, provides a more natural framework to justify the above arguments by relying on well known properties of MAC. 
 $\square$
 
{\bf Fourth property:} Details are provided in Appendix~\ref{APP2}.  $\square$

In the following, in Sections~\ref{sec1.4}, \ref{sec1.5p} and \ref{sec1.5}, starting from some background materials related to AWGN and MAC channels, a framework is established  to simplify the derivation of optimality conditions for 2-users weak GIC. These conditions are then used to construct the capacity region by successively taking infinitesimal steps along its boundary. These discussions pave the way to derive the boundary of the capacity region in Section~\ref{sec2}. The boundary is described assuming: (i) scalar Gaussian inputs, and (ii) both users occupy the entire bandwidth and transmit all the time. Then, in Section~\ref{sec3}, it is proved that Gaussian input is optimum. The proof is first expressed in terms of scalar Gaussian inputs, and then it is generalized to vector inputs. Relying on these results, it is concluded that the  construction of the capacity region presented in Section~\ref{sec2} is indeed optimum, i.e., the optimum Gaussian vector (as a code alphabet) is composed of i.i.d. components.       
 
\subsection{Superposition Coding and Successive Decoding (SC-SD)}  \label{sec1.4}
Section~\ref{sec1.4A} of the Appendix provides some background material regarding the application of SC-SD in a point-to-point AWGN channel.  This is a trivial example for the application of SC-SD technique, because all the resulting code configurations achieve the same rate of $\breve{R}=0.5\log_2\left(1+\frac{\breve{P}}{\breve{\sigma}^2}\right)$ for the AWGN channel, and successive continuous layers can be merged together, forming a smaller number (and eventually a single) continuous Gaussian layer(s). However, it is well known that in multi-user Information Theory, application of SC-SD technique can have non-trivial outcomes in achieving points on the boundary of capacity region. This is based on the need for Gaussian layers that cannot be merged together. 

As an example for multi-user Information Theory, it is known that SC-SD technique can facilitate achieving  rate tuples on the sum-rate front of a MAC without time sharing. The case of MAC will be further discussed in later parts of this article, as a preamble to explaining the capacity region of the GIC. Subsequently, this article shows that, using a private Gaussian code-book for each of the two users (decoded only at its intended receiver), and a single (or possibly multiple) public code-book(s) for each user, one can achieve the capacity region of the two-user $G\!I\!C(\breve{P}_1,\breve{P}_2,a,b,\breve{\sigma}_1^2,\breve{\sigma}_2^2)$ for a given strategy (fixed allocation of resources, i.e., power/time/frequency).

To move along the capacity region, this article relies on moving layers of infinitesimal power from the private part(s) to public part(s), and/or vice versa. This is in effect equivalent to adjusting the rate of such infinitesimal layers, and accordingly adjusting their order of decoding, and determining if a given layer is decoded at both receivers (public layer), or only at its intended receiver (private layer). Note that, if the infinitesimal layer, refereed to as $\delta$-layer, is moved from the private part to the public part, the orders of decoding at both receivers will be relevant, and if it is moved from public part to the private part, only the order of decoding at its  intended receiver will matter.  The powers of $\delta$-layers for users 1 and 2 are equal to $\delta \breve{P}_1$ and $\delta \breve{P}_2$, respectively. 
Note that  $\delta \breve{P}_1$ (and  $\delta \breve{P}_2$) can be viewed  the infinitesimal value used in integration over (sub-)ranges $[0,\breve{P}_1]$ (and $[0,\breve{P}_2]$) to be carried out in computing the rate of different {\em continuous layers} (continuum of infinitesimal layers involved in SC-SD code construction). The first viewpoint entails the total power of user 1 and user 2 are divided into an infinitely large number ($\breve{P}_1/\delta=\breve{P}_2/\delta$) of  infinitesimal Gaussian layers, and integration is used as the limit of  summing up the rates of consecutive infinitesimal layers forming a continuous (sub-)layer.

This article presents the general layering structure achieving a point on the capacity region of a GIC shown in Fig.~\ref{GICfig}, refereed to as 
$G\!I\!C(\breve{P}_1,\breve{P}_2,a,b,\breve{\sigma}_1^2,\breve{\sigma}_2^2)$.  We first need to discuss the concept of {\em nested optimality}, a definition presented to unify some known results for the point-to-point AWGN and Gaussian MAC channels, which will be then extended to the case of GIC.

 \subsection{Moving of $\delta$-layers to Cover the Boundary of a 2-Users GIC} \label{sec1.5}

This article relies on a different interpretation in forming the Gaussian code-books achieving points on the boundary of the capacity region of 2-users weak GIC. Each step along the boundary corresponds to moving a single $\delta$-layer of power from one of the two 
Gaussian code-books (public/private) of one of the two users to the other  Gaussian code-book (private/public) of the same user. For example,  a single $\delta$-layer is moved from the  public code-book of user 2 to the private code-book of user 2. Analogous to water filling, $\delta$-layers forming a continuous Gaussian code-book posses a {\em fluidity property}, as explained next.   

{\bf Definition 1 -} {\em Fluidity property of $\delta$-layers in forming a continuous Gaussian code-book} (see Figs.~\ref{MAC-pointA-R2} and~\ref{MAC-pointA-R1}): Analogous to water filling, continuous Gaussian code-books forming the public and private messages can be each viewed as a container of water, and moving of a $\delta$-layer is equivalent to moving a small amount of water from a first container (public/private) to a second container (private/public). The moving of $\delta$-layer is performed by adjusting its rate such that the $\delta$-layer can be decoded within its new category of public/private. Under these conditions, the moved $\delta$-layer becomes a fluid layer in the second container.   {\em Fluidity Property} means the issues of: (i) which one of the layers in the fist container (order of the layer in SC-SD) is being moved, and (ii) where  it is moved to (order of the layer in the second container),  do not affect the net outcome due to the move (in terms of change in rate). Note that by the word {\em order} it is meant the order of decoding, which governs the effective noise observed by a given $\delta$-layer, and consequently, determines its rate. $\square$
  
To verify the {\em Fluidity Property}, referring to Fig.~\ref{MAC-pointA-R2}, let us label the $\delta$-layers in the public part of user 2 by $1,\cdots,N_2$ from top to bottom. It follows that the changes in the rate of the public message of user 2, caused by moving  a single $\delta$-layer from public to private, does not depend on the index of the layer to be moved. Regardless of which layer is moved, the rate of the public message  of user 2, after the move, will be equal to the rate of layers indexed by $1,\cdots,N_2-1$ prior to the move. The reason is that the effective noise observed by any of layers $\ell\in\{1,\cdots N_2-1\}$ after the move is the same as the effective noise observed by layers with the same indices $\ell\in\{1,\cdots N_2-1\}$ prior to the move. In addition, the change in the rates of the private messages  of user 1 and user 2, due to moving the $\delta$-layer, does not depend on where the new layer is added among the $\delta$-layers forming the private message of user 2. The reason is that, different positions for the moved layer will become equivalent if one applies a relabeling of layers, an operation that does not affect the effective noise power, and consequently, the rate. Similar arguments apply to moving a $\delta$-layer for user 1 (also see Fig.~\ref{MAC-pointA-R1}).

     \begin{figure}[htp]
   \centering
   \includegraphics[width=0.55\textwidth]{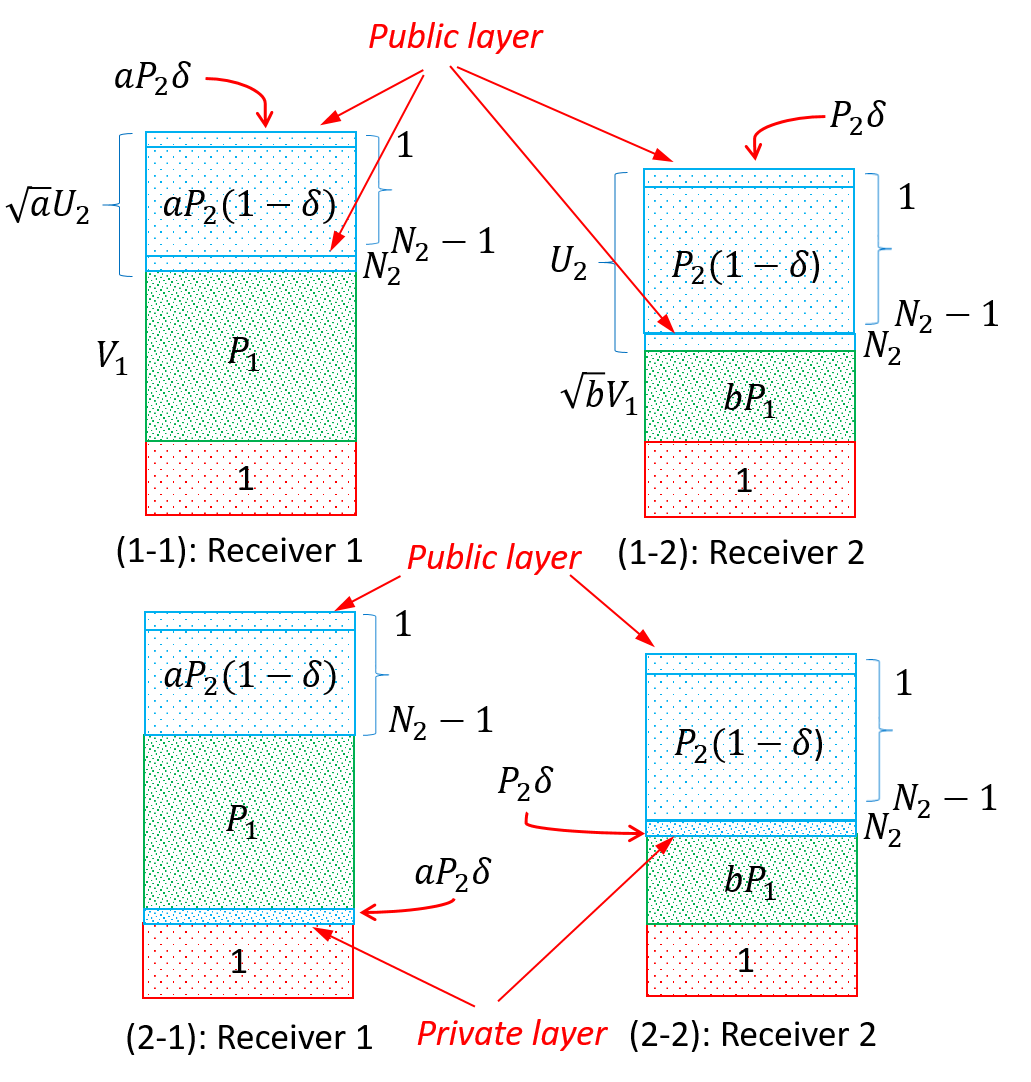}
   \caption{Moving a $\delta$-layer from the public part (1-1)/(1-2) to the private part (2-1)/(2-2) for user 2. These figures show that neither the index of the layer   to be moved, i.e., order of the $\delta$-layer among $\delta$-layers forming the public code-book, nor the position where the $\delta$-layer is moved to in the private code-book, affect the net change in $R_{ws}$.}
   \label{MAC-pointA-R2}
 \end{figure}

   \begin{figure}[htp]
   \centering
   \includegraphics[width=0.45\textwidth]{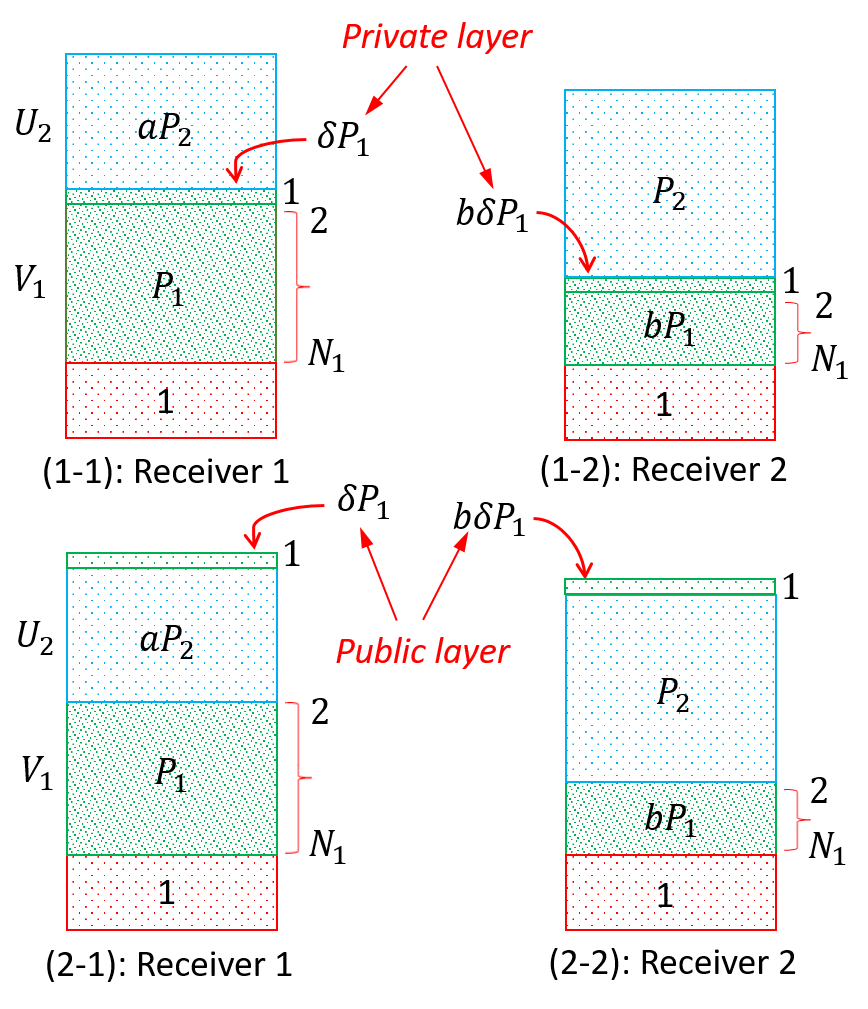}
   \caption{Moving a $\delta$-layer from the private part (1-1)/(1-2) to the public part: (2-1)/(2-2) for user 1. These figures show that neither the index of the layer   to be moved, i.e., order of the $\delta$-layer among $\delta$-layers forming the private code-book, nor the position where the $\delta$-layer is moved to in the public code-book, affect the net change in $R_{ws}$.}
   \label{MAC-pointA-R1}
 \end{figure}
 
 \subsubsection{Advantages of Relying on $\delta$-layers to Cover the Boundary}

Traditional approaches to computing the boundary of the capacity region are based on: (i) finding  the optimum layering structure and the corresponding decoding orders at the two receivers, and (ii)   optimizing the corresponding rate/power values through constrained optimization, where the rate is optimized by optimally allocating the power subject to constraints on power budgets $P_1$ and $P_2$.  This work relies on a different interpretation of the underlying optimization problems based on moving of (normalized) $\delta$-layers of power.  Moving normalized $\delta$-layers provides a framework to gradually compute the changes in the optimum layering structure as one moves along the boundary. This approach is simpler and more intuitive in comparison to first proving the optimality of certain layering structures (for each of the two users, including decoding orders at the two receivers) and then optimizing corresponding power/rate values.  In simple words, in traditional approaches, multiple complex parametric optimization problems need to be solved (each corresponding to a possible layering structure and its associated decoding order) and then the results should be compared to find the optimum solution. The alternative approach followed in this article constructs the optimum solution as part of the  procedure used to find the optimum layering and its associated power allocation and decoding order.

In the following, the article studies the capacity region of two-users weak GIC subject to using Gaussian code-books, and subsequently, Section~\ref{sec3} proves the optimality of using Gaussian code-books.

\section{Boundary of the Capacity Region for Scalar  Gaussian Inputs}\label{sec2}

First, the concept of nested optimality will be reviewed for point-to-point AWGN and MAC in Section~\ref{sec141} and Section~\ref{sec142}, respectively, which will be then extended to the case of 2-users  GIC in Section~\ref{nested-section}. Nested optimality will be used to derive some optimality conditions which will in turn determine the shape of the capacity region. 

\subsection{Nested Optimality}  \label{sec1.5p}

  \subsubsection{Nested Optimality in Point-to-Point AWGN} \label{sec141}
  Figure~\ref{nested-awgn}\,(1) shows the Gaussian code-book of power $\breve{P}$ used to achieve capacity of an AWGN channel using a single continuous layer.  Let us use the notation  $AW\!G\!N(\breve{P},\breve{\sigma}^2)$  to refer to such a channel. Figure~\ref{nested-awgn}\,(2) shows the case that the capacity of the same channel is achieved using superposition of two continuous Gaussian layers of powers $\gamma \breve{P}$ and  $(1-\gamma) \breve{P}$. Nested optimality refers to the property that the Gaussian layer of power  $\gamma \breve{P}$ is optimum for   $AW\!G\!N(\gamma \breve{P}, (1-\gamma) \breve{P}+\breve{\sigma}^2)$ and the Gaussian layer of power  $(1-\gamma) \breve{P}$ is optimum for     $AW\!G\!N((1-\gamma) \breve{P},\breve{\sigma}^2)$.

   \begin{figure}[htp]
   \centering
   \includegraphics[width=0.35\textwidth]{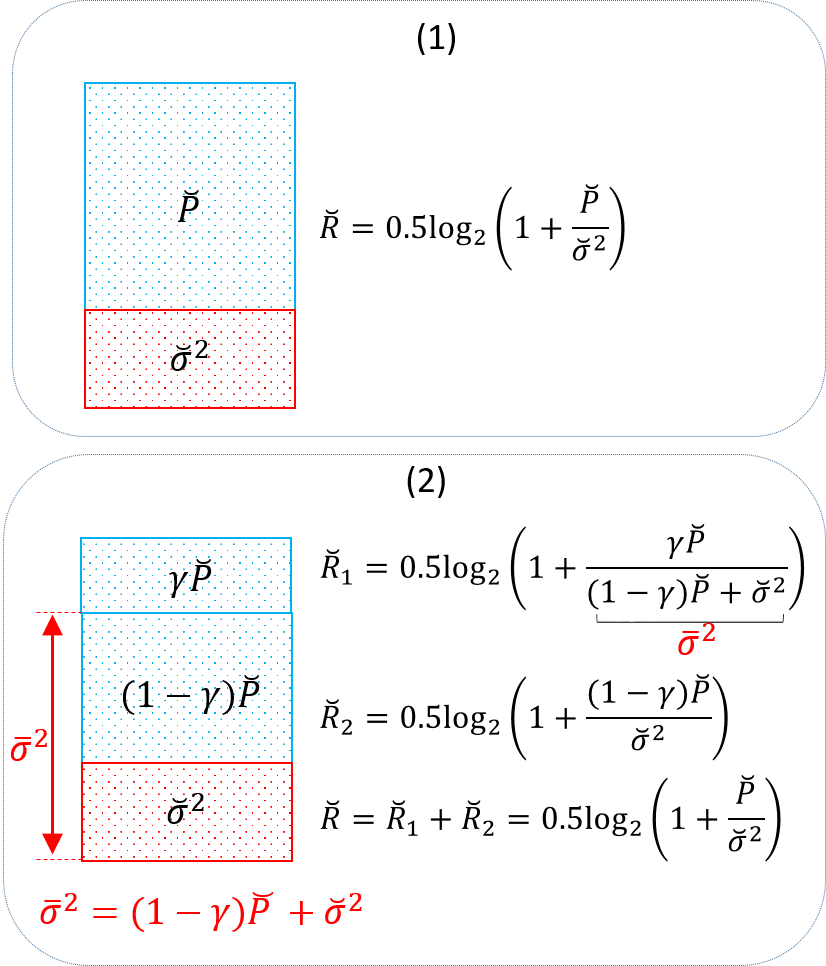}
   \caption{Nested optimality for achieving capacity of a point-to-point AWGN channel.}
   \label{nested-awgn}
 \end{figure}
 
   \subsubsection{Nested Optimality in Gaussian Multiple Access Channel}\label{sec142}
   Figure~\ref{nested-mac}\,(1) shows the application of two continuous Gaussian layers (code-books) of power $\breve{P}_1$ and $\breve{P}_2$ to achieve a corner point on the capacity region of a Gaussian MAC with power constraints $\breve{P}_1$ and $\breve{P}_2$  and AWGN of power $\breve{\sigma}^2$, refereed to as $M\!AC(\breve{P}_1,\breve{P}_2,\breve{\sigma}^2)$. This is a point on the boundary of the capacity region maximizing $\breve{R}_1+m \breve{R}_2$ for $m=1$ (indeed for $0\leq m\leq 1$) referred to as {\em sum-rate front}. Figure~\ref{nested-mac}\,(2) shows a different layering structure for achieving a different point on the boundary of the capacity region maximizing sum-rate. Nested optimality refers to the property that, in Fig.~\ref{nested-mac}\,(2),  Gaussian layers (code-books) of power $\beta \breve{P}_1$ and $\alpha \breve{P}_2$ are optimum for maximizing sum-rate in 
   $M\!AC(\beta \breve{P}_1,\alpha \breve{P}_2,(1-\beta) \breve{P}_1+(1-\alpha) \breve{P}_2+\breve{\sigma}^2)\equiv M\!AC(\beta \breve{P}_1,\alpha \breve{P}_2,\bar{\sigma}^2)$ where $\bar{\sigma}^2=(1-\beta) \breve{P}_1+(1-\alpha) \breve{P}_2+\breve{\sigma}^2$ and Gaussian layers (code-books) of power $(1-\beta) \breve{P}_1$ and $(1-\alpha) \breve{P}_2$ are optimum for maximizing sum-rate in $M\!AC((1-\beta) \breve{P}_1,(1-\alpha) \breve{P}_2,\breve{\sigma}^2)$. As a second example,    Fig.~\ref{nested-mac}\,(3) shows a different layering structure based on nested optimality.  In Fig.~\ref{nested-mac}\,(3),  Gaussian layers (code-books) of power $\beta \breve{P}_1$ and $\breve{P}_2$ are optimum for maximizing sum-rate in 
   $M\!AC(\beta \breve{P}_1,\breve{P}_2,(1-\beta) \breve{P}_1+\breve{\sigma}^2)\equiv M\!AC(\beta \breve{P}_1,\breve{P}_2,\tilde{\sigma}^2)$ where $\tilde{\sigma}^2=(1-\beta) \breve{P}_1+\breve{\sigma}^2$ and Gaussian layer (code-book) of power $(1-\beta) \breve{P}_1$ is optimum for maximizing the rate in $M\!AC((1-\beta) \breve{P}_1,0,\breve{\sigma}^2)\equiv AW\!G\!N((1-\beta) \breve{P},\breve{\sigma}^2)$.

       \begin{figure}[htp]
   \centering
   \includegraphics[width=0.6\textwidth]{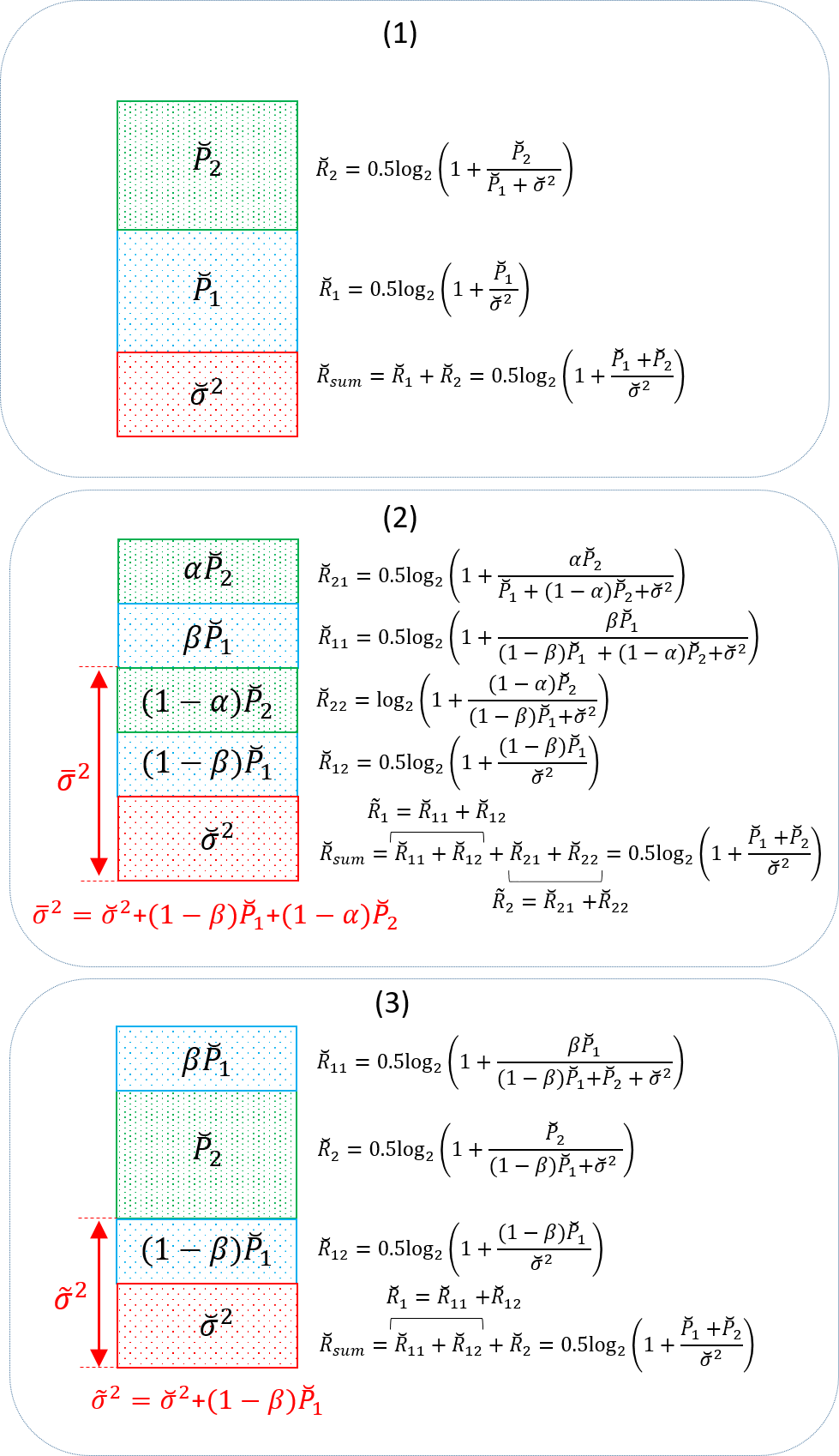}
   \caption{Nested optimality for achieving capacity region of a Gaussian Multiple Access Channel (Gaussian MAC).}
   \label{nested-mac}
 \end{figure}

 \subsubsection{Nested Optimality in 2-users Weak GIC} \label{nested-section}
 
In the following, with some misuse of notation, the set of arguments $\alpha$, $\beta$, $\gamma$, etc.,  used in the definition of $G\!I\!C(\alpha, \beta,\gamma,...)$  are selected according to the context and requirements of each discussion. 
The following theorems extend the concept of {\em nested optimality} to two-users GIC. The corresponding results will be helpful in computing the power allocated to the private messages, and will also serve in establishing  optimality conditions for 2-users weak GIC in the rest of this article. 

\begin{theorem} \label{theorem1} Consider $G\!I\!C(P_1,P_2,\mu)$, for a given strategy, i.e., given power budgets $P_1$ and $P_2$. The optimum power allocation $(P^{(public)}_1, P^{(private)}_1)$, $(P^{(public)}_2, P^{(private)}_2)$ satisfy $P^{(public)}_1+P^{(private)}_1=P_1$  and $P^{(public)}_2+P^{(private)}_2=P_2$. This means the power budgets are completely utilized. \end{theorem}
 
\begin{proof}  Let us consider a solution ${\mathcal S}$ where the power budget $P_1$ and/or $P_2$ is/are not entirely utilized. Any leftover power can be used to add a new layer of public message(s) on top of the solution ${\mathcal S}$, with rates that permit detection subject to the public message(s) and private messages(s) in  ${\mathcal S}$ acting as additive noise. Rate(s) of these additional layer(s) increases $R_{ws}$ contradicting the optimality of   solution ${\mathcal S}$. 
\end{proof}

\begin{theorem} \label{theorem2} Consider a point on the boundary of  $G\!I\!C(P_1,P_2,a,b,1,1)$, maximizing $R_{ws}=R_1+\breve{\mu} R_2$, where powers allocated to private messages are equal to $P^{(private)}_1$ and  $P^{(private)}_2$, respectively. The optimum solution, maximizing $R_{ws}$ in  $G\!I\!C(P^{(private)}_1,P^{(private)}_2,a,b,1,1)$ for the same weight factor $\breve{\mu}$  has only private messages of powers $P^{(private)}_1$ and  $P^{(private)}_2$, respectively. 
\end{theorem}

\begin{proof}  Let us use notations $R^{(public)}_{ws}$ and $R^{(private)}_{ws}$ to refer to the contributions of the public and private messages to $R_{ws}$, i.e.,
 $R_{ws}=R^{(public)}_{ws}+R^{(private)}_{ws}$, and use notations $R_1^{(private)}$ and  $R_2^{(private)}$ to refer to the rates of the two private messages. 

Let us assume the optimum solution for $G\!I\!C(P_1,P_2,a,b,1,1)$, maximizing $R_{ws}=R_1+\breve{\mu} R_2$, is given in Fig.~\ref{theorem3pp}(a). 
We need to show that, upon decoding of the public messages, for given $P^{(private)}_1$ and  $P^{(private)}_2$, $R^{(private)}_{ws}=R^{(private)}_1+\breve{\mu} R^{(private)}_2$, is the optimum weighted sum-rate for $G\!I\!C(P^{(private)}_1,P^{(private)}_2,a,b,1,1)$ as shown in Fig.~\ref{theorem3pp}(b). If this is not the case, then the optimum $R_{ws}$ for  $G\!I\!C(P^{(private)}_1,P^{(private)}_2,a,b,1,1)$ could be increased by including some non-zero public message for either or both of the users, at the cost of reducing power(s) of corresponding private messages  as shown in Fig.~\ref{theorem3pp}(c). This means $R_{ws}$ in Fig.~\ref{theorem3pp}(c) is higher than $R_{ws}$ in Fig.~\ref{theorem3pp}(b).
To distinguish between public messages, we refer to public messages in Fig.~\ref{theorem3pp}(a) as ``primary public messages" and to public messages extracted from private messages in Fig.~\ref{theorem3pp}(c) as ``auxiliary public messages". Auxiliary public messages are formed to maximize  $R_{ws}$ with weight factor 
$\breve{\mu}$ in Fig.~\ref{theorem3pp}(c). 
 
Now let us assume primary public messages in Fig.~\ref{theorem3pp}(a) are added once again as additional (second layer) of public 
messages to Fig.~\ref{theorem3pp}(c), 
resulting in Fig.~\ref{theorem3pp}(d).  Let us assume the two types of public messages, i.e., primary and auxiliary, are sequentially decoded.  
The rate due to primary public messages is determined by the intersection of the two MAC channels  for which the capacity regions are affected by powers of primary public messages and powers of noise terms. In decoding of primary public messages in Fig.~\ref{theorem3pp}(d), noise terms are composed of private messages plus auxiliary public messages. This means the ``powers of primary public messages'' and the ``power of noise terms'' in Fig.~\ref{theorem3pp}(d)  are the same as their corresponding values in Fig.~\ref{theorem3pp}(a). Consequently, the weighted sum-rate due to primary public messages in Fig.~\ref{theorem3pp}(d) remains the same as in Fig.~\ref{theorem3pp}(a), which, upon successive decoding, will be added to the weighted sum-rate of GIC in Fig.~\ref{theorem3pp}(c) to form the weighted sum-rate of GIC in Fig.~\ref{theorem3pp}(d). The conclusion is that if $R_{ws}$ in GIC of Fig.~\ref{theorem3pp}(c) is larger than the $R_{ws}$ in GIC of Fig.~\ref{theorem3pp}(b), then the $R_{ws}$ in GIC of Fig.~\ref{theorem3pp}(d) will be also larger than the  $R_{ws}$ in GIC of Fig.~\ref{theorem3pp}(a). This contradicts the assumption of optimality for GIC of Fig.~\ref{theorem3pp}(a).  This verifies the configuration in Fig.~\ref{theorem3pp}(b) is optimum without allocating any power to any public messages, and it completes the proof. 

\end{proof}

      \begin{figure}[htp]
   \centering
   \includegraphics[width=1.0\textwidth]{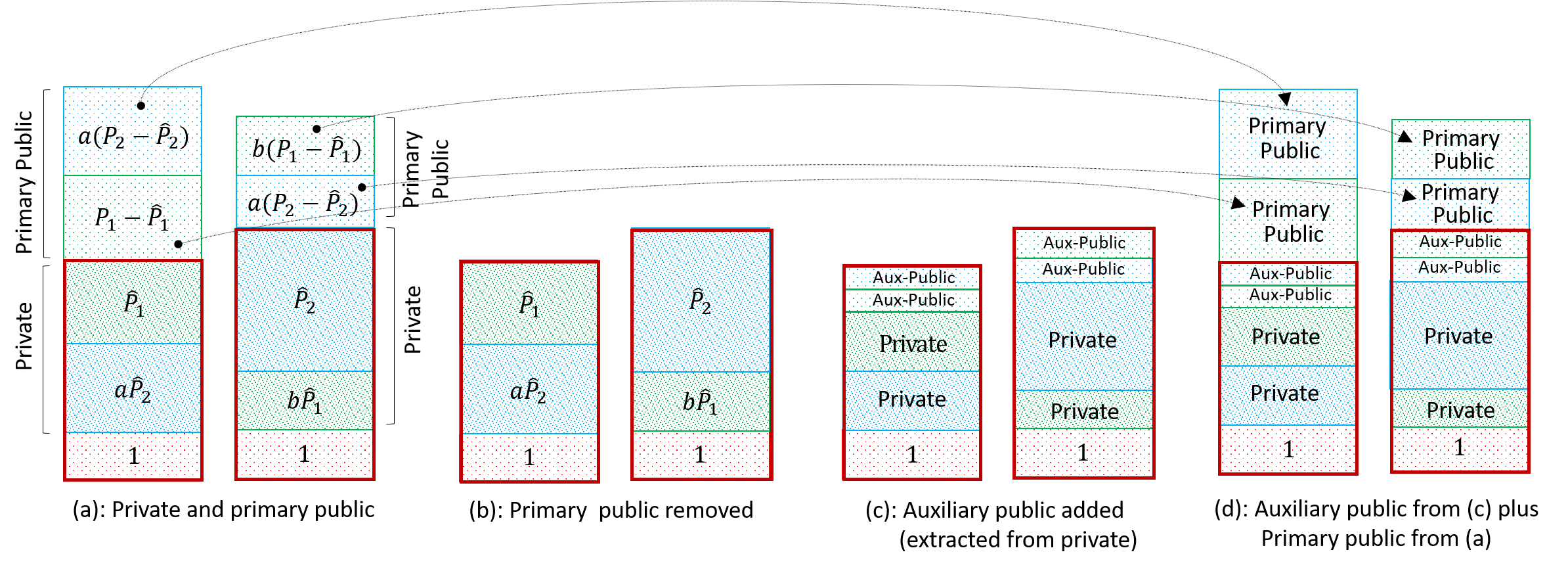}
   \caption{Cases in  the proof of nested optimality for $G\!I\!C(P_1,P_2,a,b,1,1)$.}
   \label{theorem3pp}
 \end{figure}

      \begin{figure}[htp]
   \centering
   \includegraphics[width=0.85\textwidth]{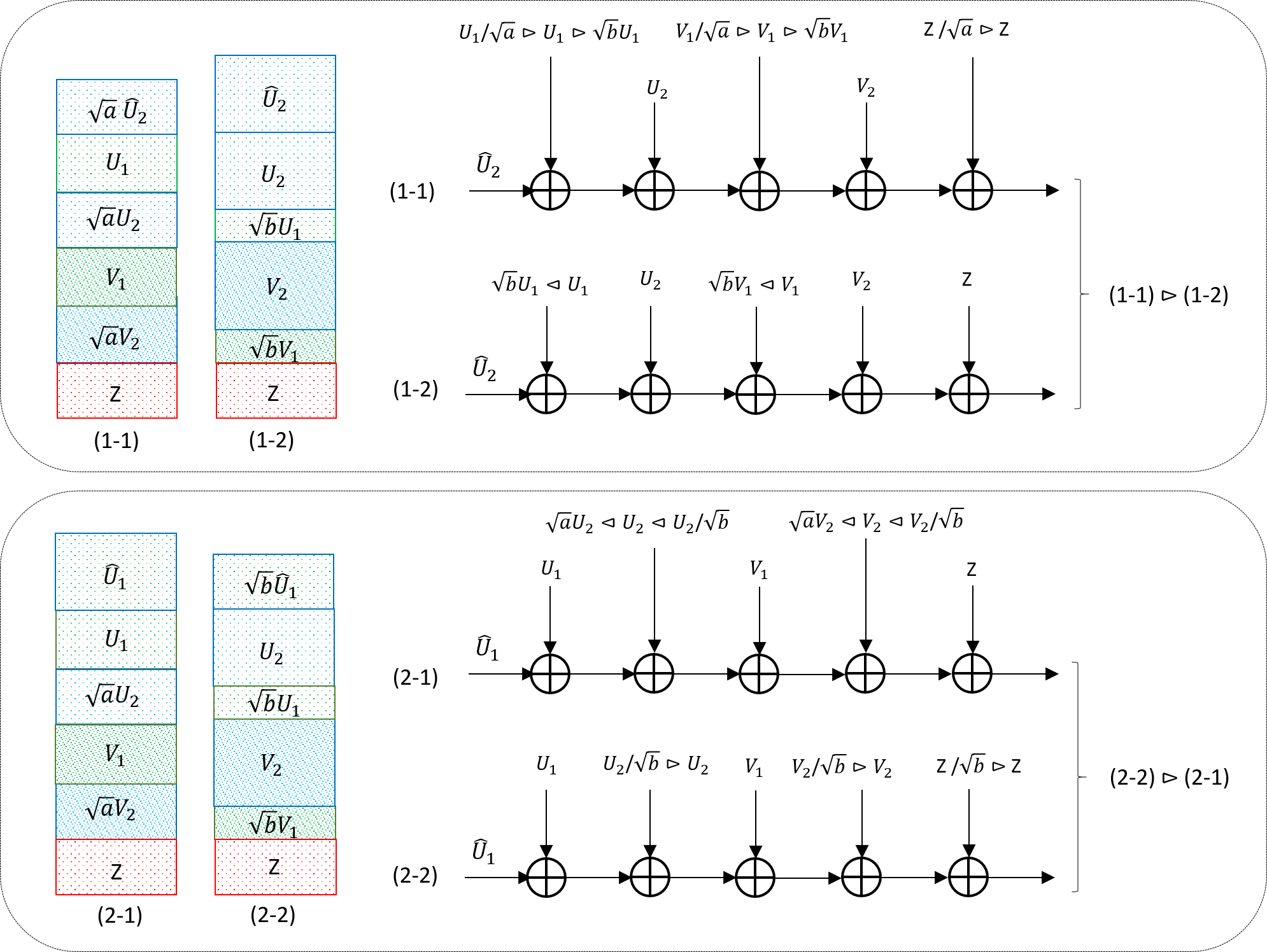}
   \caption{Figure related to Theorem~\ref{theorem3-TH}, note that $X \triangleleft	Y$, equivalently $Y \triangleright	X$, means $X$ is a degraded version of $Y$.}
   \label{theorem3}
 \end{figure}
  
 Section~\ref{sec-inter} discusses different cases for the intersection MAC1 and MAC2. Notations $R_1^-(1)$ and $R_1^-(2)$ are used to specify the rate of $U_1$ at $Y_1$ and $Y_2$, respectively, when all other terms are considered as noise. Likewise, notations $R_2^-(1)$ and $R_2^-(2)$  specify the rate of $U_2$ at $Y_1$ and $Y_2$, respectively, when all other terms are considered as noise. In such a case, the constraint determining decodability (rate) of $U_1$ is governed by restrictions imposed at $Y_2$, and vice versa, i.e., 
\begin{align} \label{R12}
R_1^-(2) & ~<~  R_1^-(1) \\ \label{R21}
R_2^-(1) & ~<~  R_2^-(2). 
\end{align} 
Following theorem proves \ref{R12} and \ref{R21} in a slightly more general setting. 
 
\begin{theorem}  \label{theorem3-TH} Consider Fig.~\ref{theorem3} where public layer $\hat{U}_1$ is decoded while considering all other terms as noise, and likewise, $\hat{U}_2$ is decoded while considering all other terms as noise. Rate of public layer $\hat{U}_2$ in Fig.~\ref{theorem3}(1) is determined by constraint on detectability at $Y_1$, and   rate of public layer $\hat{U}_1$ in Fig.~\ref{theorem3}(2) is determined by constraint on detectability at $Y_2$. 
\end{theorem}

\begin{proof}  Proof follows by noting that the public layer $\hat{U}_2$ in Fig.~\ref{theorem3}\,(1-1) is a degraded version of $\hat{U}_2$ in Fig.~\ref{theorem3}\,(1-2), likewise,  public layer $\hat{U}_1$ in Fig.~\ref{theorem3}\,(2-1) is a degraded version of $\hat{U}_1$ in Fig.~\ref{theorem3}\,(2-2). 
\end{proof}

Note that Theorem~\ref{theorem3-TH} is valid regardless of the distribution of additive noise $Z$, and distributions of other layers of public messages and distributions of the private messages.

Next we summarize the changes in the allocation of power to private and public parts and the decoding strategies as one moves along the capacity region of the 2-users weak GIC. We study segment of the boundary corresponding to $\mu\leq 1$, referred to as ``the lower part of the capacity region''. It is easy to see that the case of 
$1\leq \mu$ (upper part of the boundary) can be obtained from the case of 
$\mu\leq 1 $ (lower part of the boundary) by applying exchanges in expressions~\ref{changes}. It should be emphasized that by ``capacity region", here it is meant ``a single constituent capacity region". This is the capacity region corresponding to a single strategy. The actual capacity region is obtained by optimum dividing of resources (time/bandwidth/power) among multiple constituent capacity regions and optimum time sharing among them. This article shows that the optimum random coding probability density function for each constituent region is i.i.d. Gaussian.  

  Given power budgets $P_1$ and $P_2$, the first step in solving the optimization problem for maximizing $R_{ws}=R_1+\mu R_2$ is to compute the fraction of power allocated to private messages. Using traditional optimization techniques for this purpose is complicated because the closed form expression for $R_{ws}$ changes from one segment of the boundary  to another. 
This article relies on a simpler and more intuitive approach to solve the underlying optimization problems. Starting from point $A$ (see Figs.~\ref{Fig2} and \ref{Fig3p}), taking a step along the boundary is formulated in terms of moving a $\delta$-layer of power from the public part to the private part, or vice versa. 
By dividing the boundary into segments over which $\mu$ is continuous, for a given step within a given segment, it is required to move a $\delta$-layer for only one of the two users. In other words, none of the steps requires adjusting the allocation of power for  both users at the same time (within the same $\delta$-step). This is helpful in proving the optimality of Gaussian random code-books in Section~\ref{sec3}.

  \subsection{Moving Counterclockwise Along the Lower part of the Boundary}
  
    To compute the capacity region, we start from the corner point with maximum $R_1$ (point $A$ in Figs.~\ref{Fig2} and \ref{Fig3p}) and move counterclockwise along the boundary in normalized steps (each step increases $R_2$ at the cost of a reduction in $R_1$).  To form such normalized steps, we divide each of the two ranges $[0,P_1]$ and $[0,P_2]$ into $L=1/\delta$ infinitesimal portion of power $\delta P_1$ and $\delta P_2$, respectively.  By using $L=1/\delta$ steps, one would be able to sweep through the entire ranges of the two power values in an equal number of normalized steps. The starting point, refereed to as Point $A$, is the point maximizing $R_1$. Noting the power budget of user 1, namely $P_1$, the rate from transmitter 1 to receiver 1 is limited by the capacity of the $AW\!G\!N(P_1,1)=0.5\log_2(1+P_1)$, which should rely on using a Gaussian random code-book of power $P_1$ for $X_1$. This maximum rate can be achieved if the rate of user 2 is low enough such that $X_2$ can be decoded at $Y_1$, while considering $X_1$ as (additional) noise (interference). Upon removing $X_2$ at $Y_1$, $X_1$ sees a Gaussian noise of power 1; consequently, to maximize $R_1$, one should use a Gaussian code-book for $X_1$. Keep in mind that, to maximize $R_1$, $X_1$ must be Gaussian regardless of $X_2$ and $R_2$, consequently, $X_2$ sees a Gaussian noise of power $P_1+1$ at $Y_1$ and a Gaussian noise of power $aP_1+1$ at $Y_2$. The rate of $X_2$ in both of these configurations is maximized using a Gaussian random code-book of power $P_2$ for $X_2$.  The rate of $R_2$ is governed by decodability constraint at $Y_1$. This is consistent with 
Theorem~\ref{theorem3-TH}.  Figure~\ref{Fig2} depicts the details of achieving  Point $A$ and its corresponding rate values. Figure~\ref{Fig3p} depicts the first $\delta$-step along the boundary starting from point $A$. 

  \begin{figure}[htp]
   \centering
   \includegraphics[width=0.9\textwidth]{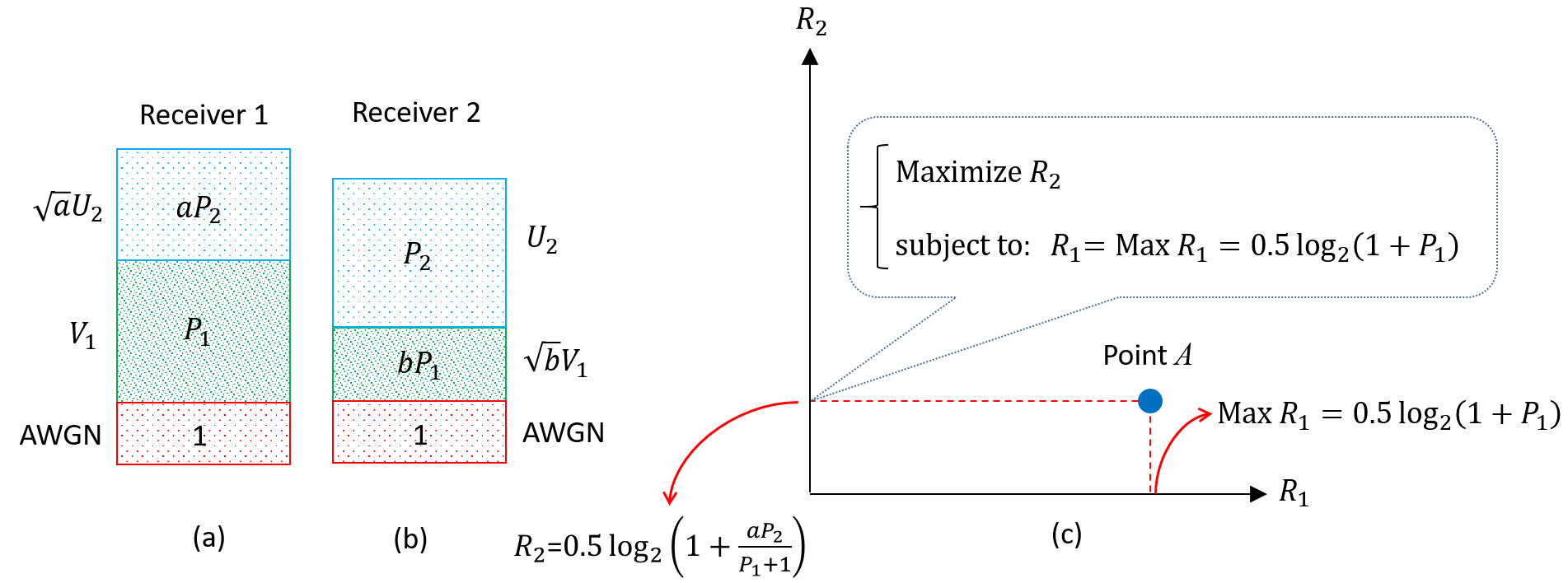}
   \caption{Corner point maximizing $R_1$ (Point $A$).}
   \label{Fig2}
 \end{figure}
  
   \begin{figure}[htp]
   \centering
   \includegraphics[width=0.6\textwidth]{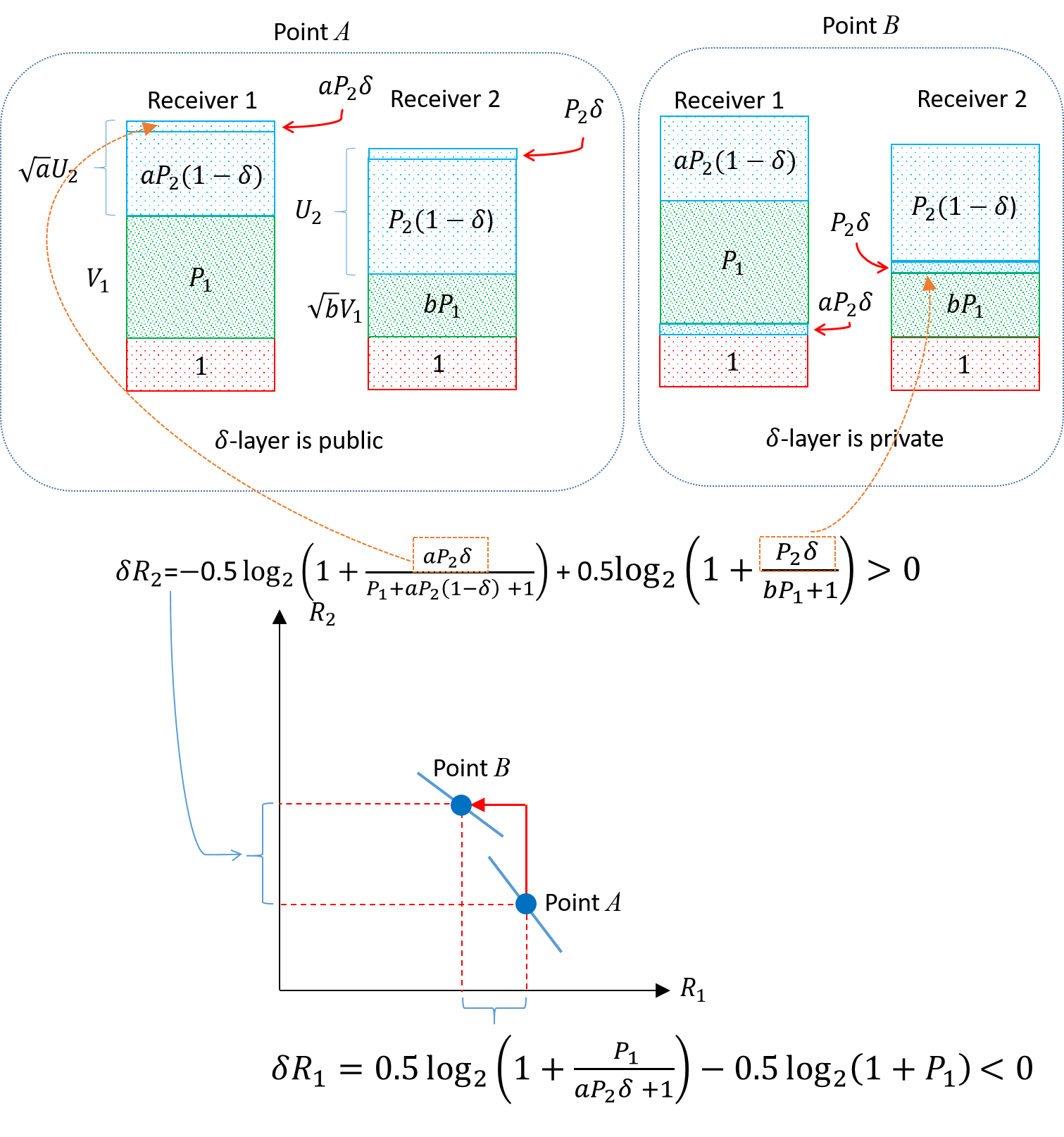}
   \caption{Moving a $\delta$-step from the corner point maximizing $R_1$ (Point $A$) along the boundary of the capacity region to a point $B$ by relocating a $\delta$-layer of power for user 2. }
   \label{Fig3p}
 \end{figure}

Figure \ref{move1} and Fig.~\ref{move-D} are provided to shed some light on the behavior of $\hat{P}_1$ and $\hat{P}_2$.
Figure~\ref{move1} depicts how the parameters $\rho$ and $\theta$ change as one moves counterclockwise along the boundary. 
 At the corner point $A$,  we have $\rho=0$ ($X_1$ is entirely private), and $\theta=1$ ($X_2$ is entirely public).

  \begin{figure}[htp]
   \centering
   \includegraphics[width=0.7\textwidth]{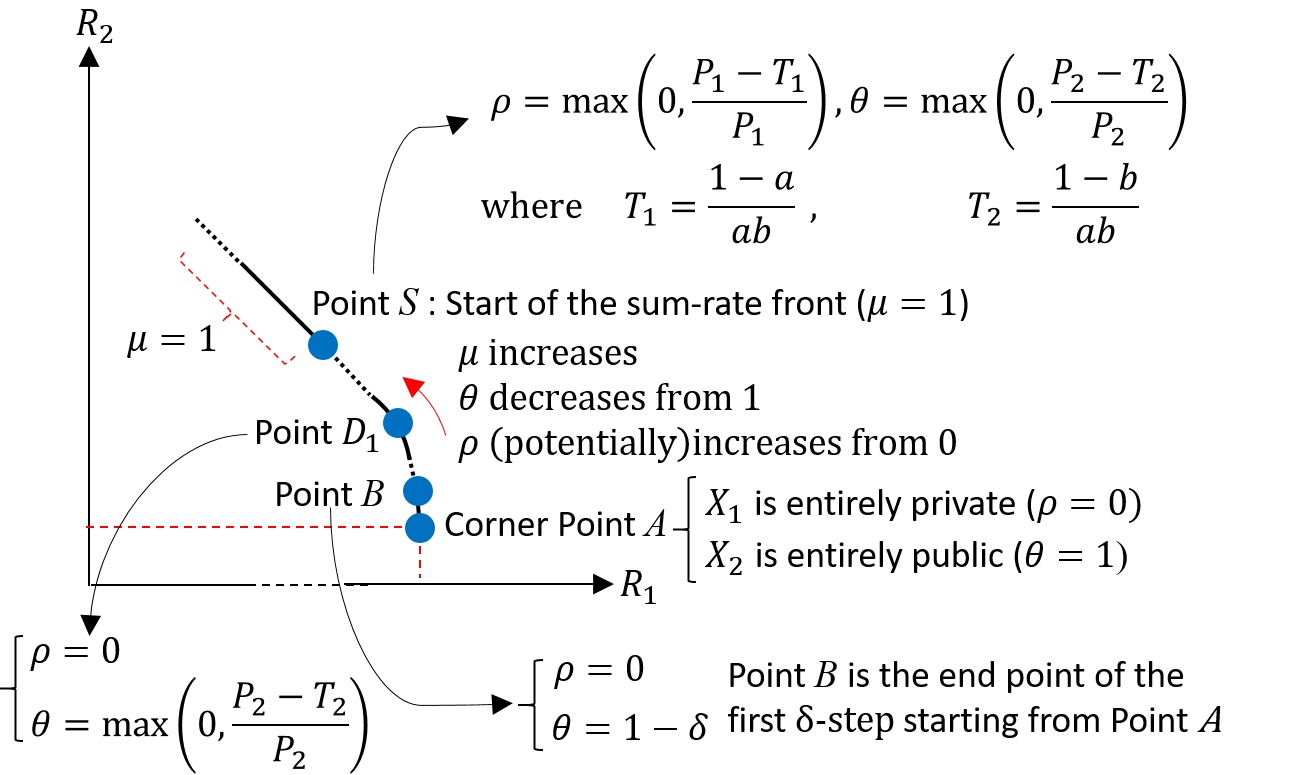}
   \caption{Moving counterclockwise along the lower part, i.e., $\mu\leq 1$, of the capacity region boundary for two-users GIC. } 
   \label{move1}
 \end{figure}
 
   \begin{figure}[htp]
   \centering
   \includegraphics[width=0.75\textwidth]{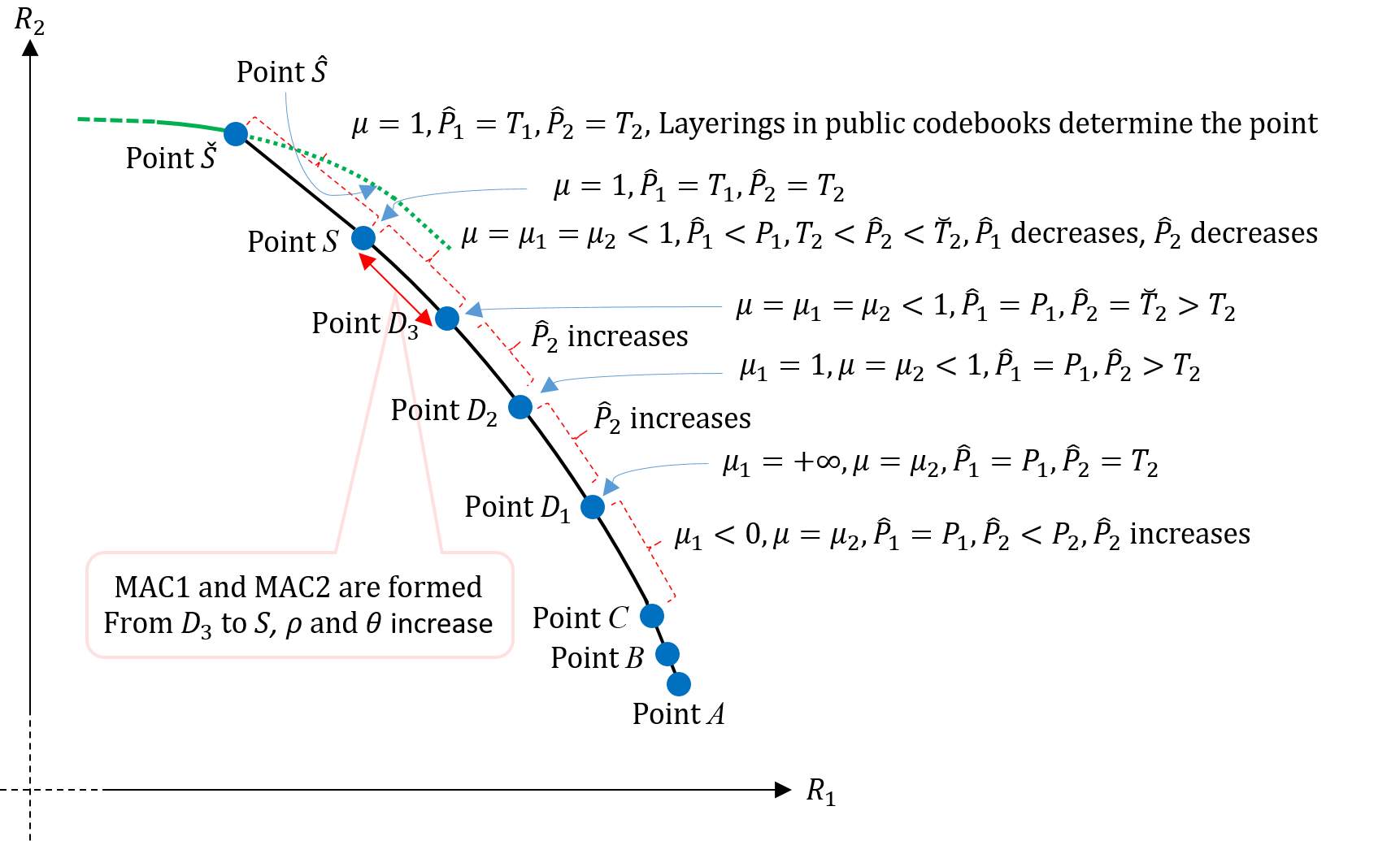}
   \caption{Moving counterclockwise along the lower part of the capacity region from point $A$ to point $S$, subject to the condition $P_1(1-b)<P_2(1-a)$, $P_1>T_1$, $P_2>\breve{T}_2>T_2$ and $ab<\frac{1}{2}$  (these conditions are needed to realize the richest structure for the lower part of the boundary as shown in this figure). }
   \label{move-D}
 \end{figure}
      
Next, we consider that scenario that starting from point $A$, and for $0< \mu \leq 1$, one moves counterclockwise along boundary of the capacity region up to the line (or point) maximizing the sum-rate, i.e., $\mu =1$.  This is refereed to as the ``lower part'' of the boundary. The ``upper part'' of the boundary, i.e.,  $\mu \geq 1$ can be computed following a similar set of arguments as used for deriving the lower part, and applying exchanges in expression~\ref{changes}. To construct the lower part of the boundary, relying on nested optimality, the conditions for optimality of the solution for $G\!I\!C(P_1,P_2,a,b,1,1)$ is described in terms of optimality condition(s) for   
$G\!I\!C(P_1^{(private)},P_2^{(private)},a,b,1,1)$. The stationary condition for optimality necessitates that moving a $\delta$-layer from the private part of either users to its initially zero public part does not change $R_{ws}$ for the underlying $G\!I\!C(P_1^{(private)},P_2^{(private)},a,b,1,1)$, to be discussed next, but first we need some clarifications concerning the procedure governing movement of $\delta$-layers. 

\subsubsection{Why (Leftover) Public Layer(s) are Ignored in Computing $\delta R_{ws}$ due to Moving a  $\delta$-layer}
In Section~\ref{optimality-condition}, this article refers to moving a $\delta$-layer of power between public and private messages of user 1 and/or user 2. In some cases, this is simply refereed to as {\em allocating a $\delta$-layer of power to public or private messages of users 1 and/or to public or private messages of user 2}. The following question may arise: {\em where does this (apparently) extra layer of power come from?} This viewpoint is explained/justified next. 

The general structure under study includes public messages for user 1 and user 2 that form two multiple access channels, MAC1 and MAC2, with power of private messages acting as interference. Without loss of generality, to explain moving a $\delta$-layer from public to private, let us focus on user 1.  As far as computation of rate is concerned, we can assume public message of user 1 and public message of user 2 are jointly decoded, even if successive decoding would be possible. Under these conditions, let us divide the power of public message of user 1, namely $P_1^{public}$ to two parts, 
a first part of power $P_1^{public}-\delta P_1$ and a second part of power $\delta P_1$. Now let us consider an encoding scheme wherein the first part of the public message for user 1 is jointly encoded with the public message of user 2, while considering the second part of  the public message for user 1 (as well as the two private messages) as noise. Let us refer to the resulting multiple access channels as Reduced-MAC1 and Reduced-MAC2. From properties of multiple access channel, it is concluded that such a layering of public message of user 1 does not change the total rate due to public messages. One can decode the public messages in Reduced-MAC1 and Reduced-MAC2 relying on joint decoding, and then successively decode the second part of the public message for user 1. 
Now let us assume the second part of the public message for user 1 is added as a layer to the private message of user 1. It is easy to see that Reduced-MAC1 and Reduced-MAC2 see the same total noise, and consequently their contribution to $R_{ws}$ will not not be changed. Under these conditions, the changes in $R_{ws}$ due to the second part of the public message for user 1 being public or private can be computed without considering the effect of  public messages in Reduced-MAC1 and Reduced-MAC2. Upon removing public messages forming Reduced-MAC1 and Reduced-MAC2, the remaining configuration can be viewed as if an additional $\delta$-layer of power is available to either ``be allocated to the private part of user 1'', or to ``act as a single public layer for user 1''. Similar arguments apply to the case of user 2, or in moving a $\delta$-layer of power from private to public.   

 \subsection{Optimality Conditions: Nested Optimality of Private Messages}  \label{optimality-condition}

Let us assume at a given optimum point  (corresponding to a given $\mu$), the energies allocated to private messages are equal to $P_1^{(private)}=\hat{P}_1$ and $P_2^{(private)}=\hat{P}_2$, respectively. The optimality condition for user 1 is derived by computing  the changes in $R_{ws}$ of $G\!I\!C(\hat{P}_1,\hat{P}_2,a,b,1,1)$ due to relocating a Gaussian layer 
 $\delta P_1$ from the private part of user 1 to the public part of user 1.  Likewise, the optimality condition for user 2 is derived by computing  the changes in $R_{ws}$ of $G\!I\!C(\hat{P}_1,\hat{P}_2,a,b,1,1)$ due to relocating a Gaussian layer 
 $\delta P_2$ from the private part of user  2 to the public part of user 2. 
Relocating a $\delta$-layer for user 1 plays the same role as computing the derivative of $R_{ws}$ with respect to $\rho$, where $P_1^{(public)}=\rho P_1$ and $P_1^{(private)}=\hat{P}_1=(1-\rho) P_1$, subject to the constraint that the total power of user 1 is equal to $P_1$.  Likewise, relocating of a $\delta$-layer for user 2 plays the same role as computing the derivative of $R_{ws}$ with respect to $\theta$, where $P_2^{(public)}=\theta P_2$ and $P_2^{(private)}=\hat{P}_2=(1-\theta) P_2$, subject to the constraint that the total power of user 2 is equal to $P_2$.  The effect of moving a $\delta$-layer for each of the two users is  considered separately. This is equivalent to relying on sum of partial derivatives (with respect to $\rho$ and $\theta$) to compute the total change in $R_{ws}$ with respect to infinitesimal changes in $\rho$ and $\theta$.  Recalling basic properties of  partial derivatives,  this is justified if the corresponding segment of the boundary has a continuous slope, i.e., $\mu$ is continuous. The assumption is that the boundary is divided into segments, each with a continuous slope, and the discussions presented here concern moving within one such segment with a continuous slope. Section~\ref{optimality-condition2} proves that the slope is indeed continuous.

          \begin{figure}[htp]
   \centering
   \includegraphics[width=0.6\textwidth]{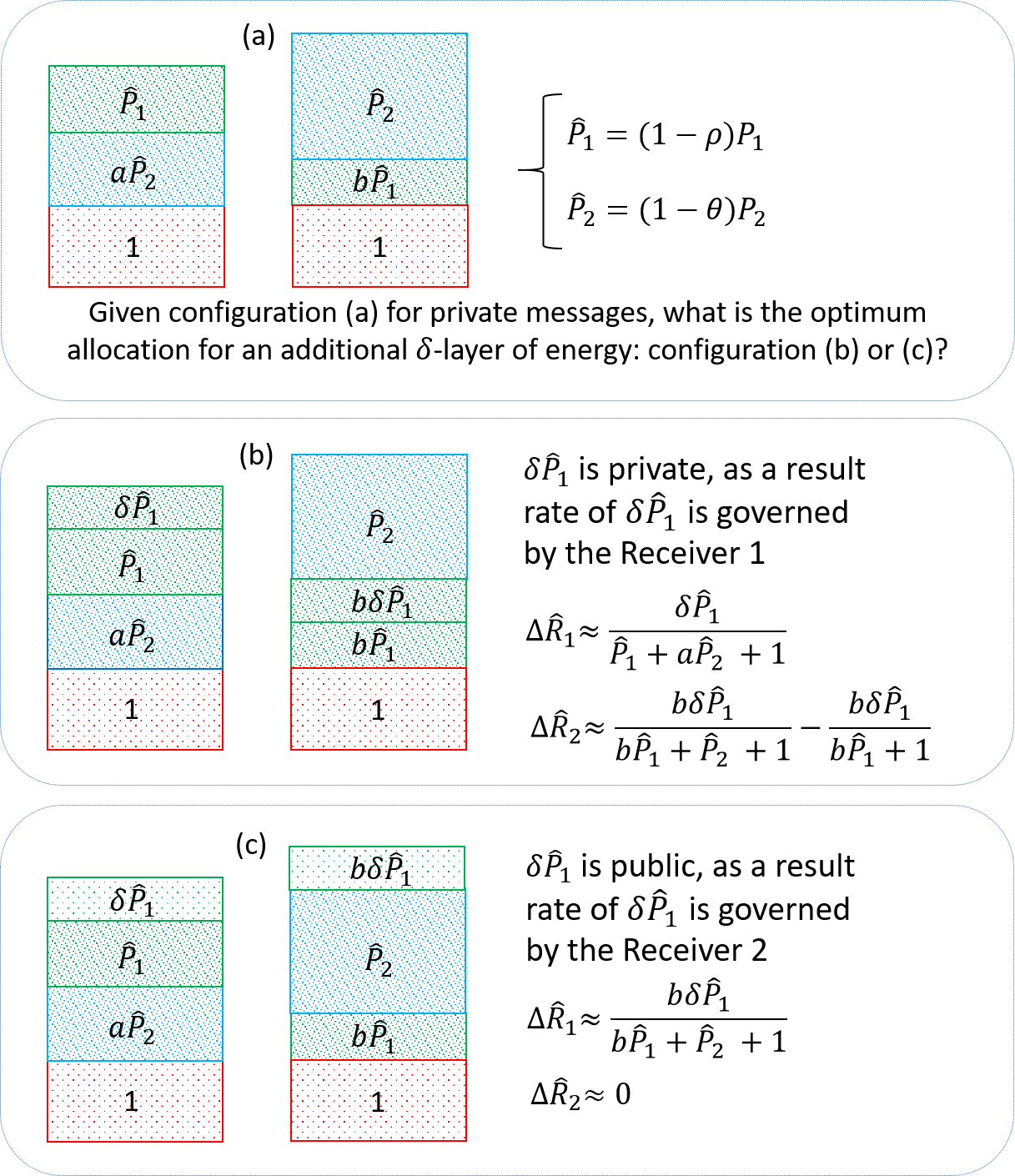}
   \caption{Optimality conditions for private layers (focusing on expressing the conditions for user 1).  The configuration in~\ref{derivative3}(a) is optimum if an additional $\delta$-layer of power for user 1 should be allocated to its public part, i.e.,  the configuration in~Fig.~\ref{derivative3}(c) results in a higher $R_{ws}$ as compared to the configuration in~Fig.~\ref{derivative3}(b). }
   \label{derivative3}
 \end{figure}
 
From Theorem~\ref{theorem3-TH}, in moving a 
 $\delta$-layer from private to public for user 1,   its rate will be governed by detectability constraint at user 2. Likewise, in moving a 
 $\delta$-layer from private to public for user 2,  its rate will be governed by detectability constraint at user 1. Referring to Fig.~\ref{derivative3}, the changes in the rate of user 1 if an additional $\delta$-layer is allocated to its private part is equal to, 
\begin{align} \nonumber 
 \mbox{Fig.}~\ref{derivative3}\mbox{(b):} & \\ \label{optimality-user1}   
  \delta \hat{R}_1= & 
\frac{1}{2}\log_2\left(\frac{1+a \hat{P}_2+\hat{P}_1+\delta \hat{P}_1}{1+a \hat{P}_2}\right)-
\frac{1}{2}\log_2\left(\frac{1+a \hat{P}_2+\hat{P}_1}{1+a \hat{P}_2}\right)
 \approx  \left(\frac{\delta \hat{P}_1}{\hat{P}_1+ a\hat{P}_2+1} \right)\mathsf{S} \\ \nonumber   
  \mbox{Fig.}~\ref{derivative3}\mbox{(b)}  & \\ \label{optimality-user1p} 
  \delta \hat{R}_2  =  &
\frac{1}{2}\log_2\left(\frac{b\hat{P}_1+b\delta\hat{P}_1+1+\hat{P}_2}{b\hat{P}_1+b\delta\hat{P}_1+1}\right)-
\frac{1}{2}\log_2\left(\frac{b\hat{P}_1+1+\hat{P}_2}{b\hat{P}_1+1}\right)
\approx  \left(\frac{b\delta \hat{P}_1}{b\hat{P}_1+\hat{P}_2+ 1}-\frac{b\delta \hat{P}_1}{b\hat{P}_1+1} \right)\mathsf{S}.  
  \end{align}
where $  \mathsf{S}=(1/2\ln(2))$ is the scale factor due to the change from $\log_2$ to $\log_e$ in capacity/rate expressions.  As the scale factor $\mathsf{S}$ will be canceled out in expressions of interest in this article,  it will be removed from expressions hereafter. 
Likewise, the changes in the rate of user 1 if an additional $\delta$-layer is allocated to its public part can be computed, resulting in 
\begin{align}    \label{optimality-user1z}
 \mbox{Fig.} &~\ref{derivative3}\mbox{(c)}  \Longrightarrow   \delta \hat{R}_1 =
\log_2\left(\frac{b\hat{P}_1+\hat{P}_2+1+b\delta\hat{P}_1}{b\hat{P}_1+\hat{P}_2+1}\right)-0  \approx  \left(\frac{b\delta \hat{P}_1}{b\hat{P}_1+ \hat{P}_2+1}
\right)\mathsf{S}  \\ \label{optimality-user1p-1}
  \mbox{Fig.}&~\ref{derivative3}\mbox{(c)}  \Longrightarrow   \delta \hat{R}_2  \approx 0.  
\end{align}
Again, the scale factor $\mathsf{S}$ will be removed, hereafter. 

 Optimum solution may be located at: (i) the boundary of the region governed by the corresponding power budgets, or (ii) at a stationary point. Case (i) corresponds to a situation in which the derivative of $R_{ws}$ with respect to $\rho$ (recall that the fraction of power allocated to the public message of user 1 is equal to $\rho P_1$) does not change sign. In other words, the equation obtained by setting the derivative (of $R_{ws}$ with respect to $\rho$) equal to zero does not have any valid solution in the admissible range of $0\leq\rho\leq1$. Case (ii) corresponds to a situation in which the derivative results in a valid solution. In the following, these arguments are expressed in the language of changes in $R_{ws}$ due to the allocation of an additional $\delta$-layer of power corresponding to user 1, i.e., $\delta P_1$, to the public vs. private message of user 1.  To differentiate between the two users, notations $\mu_1$ and $\mu_2$ are used (instead of $\mu$) in deriving the stationary conditions for users 1 and 2, respectively.  
 
From expressions \ref{optimality-user1} and \ref{optimality-user1p},  $\delta$-layer is added to the private part of user 1 if, 
 \begin{equation}
 \frac{1}{\hat{P}_1+ a\hat{P}_2+1}+\mu_1\left( \frac{b}{b\hat{P}_1+\hat{P}_2+ 1}-\frac{b}{b\hat{P}_1+1}\right) > \frac{b}{b\hat{P}_1+ \hat{P}_2+1}
\label{eq4}
\end{equation}
and it is added to the public part of user 1 if, 
 \begin{equation}
 \frac{1}{\hat{P}_1+ a\hat{P}_2+1}+\mu_1\left( \frac{b}{b\hat{P}_1+\hat{P}_2+ 1}-\frac{b}{b\hat{P}_1+1}\right) < \frac{b}{b\hat{P}_1+ \hat{P}_2+1}.\\
 \label{eq5}
\end{equation}
The condition for having a stationary solution for user 1, with some reordering of the terms, is expressed as,
 \begin{equation}
 \mu_1\left(\frac{b}{b\hat{P}_1+1}- \frac{b}{b\hat{P}_1+\hat{P}_2+ 1}\right) \stackrel{?}{=} \frac{1}{\hat{P}_1+ a\hat{P}_2+1}-\frac{b}{b\hat{P}_1+ \hat{P}_2+1}.\\
 \label{eq6}
\end{equation}

A similar set of expressions can be obtained for user 2 in terms of $\theta$, by applying the exchanges in expression~\ref{changes} to \ref{eq4}, \ref{eq5}, \ref{eq6}, summarized next. 
An additional $\delta$-layer for user 2 is added to the private part of user 2 if, 
 \begin{equation}
 \frac{1}{\hat{P}_2+ b\hat{P}_1+1}+\mu_2\left( \frac{a}{a\hat{P}_2+\hat{P}_1+ 1}-\frac{a}{a\hat{P}_2+1}\right) > \frac{a}{a\hat{P}_2+ \hat{P}_1+1}\\
\label{eq7}
\end{equation}
and it is added to the public part of user 2 if,
 \begin{equation}
 \frac{1}{\hat{P}_2+ b\hat{P}_1+1}+\mu_2\left( \frac{a}{a\hat{P}_2+\hat{P}_1+ 1}-\frac{a}{a\hat{P}_2+1}\right) < \frac{a}{a\hat{P}_2+ \hat{P}_1+1}.\\
\label{eq8}
\end{equation}
The condition for having a stationary solution for user 2 is expressed as,
 \begin{equation}
 \mu_2\left(\frac{a}{a\hat{P}_2+1}-\frac{a}{a\hat{P}_2+\hat{P}_1+ 1}\right) \stackrel{?}{=} \frac{1}{\hat{P}_2+ b\hat{P}_1+1}-\frac{a}{a\hat{P}_2+ \hat{P}_1+1}.
\label{eq9}
\end{equation}
The optimality condition for the lower part of the capacity boundary, i.e., $0\leq\mu\leq1$, requires,
    \begin{eqnarray} 
0\leq \mu=\min(\mu_1,\mu_2)\leq 1,~&~\mbox{if}~~ \mu_1\in[0,1]~~\mbox{and}~~\mu_2\in[0,1] 
\\
0\leq \mu=\mu_1\leq 1,~&~\mbox{if}~~\mu_1\in[0,1]~~\mbox{and}~~\mu_2\notin[0,1] \\
0\leq \mu=\mu_2\leq 1,~&~\,\mbox{if}~~\mu_1\notin[0,1]~~\mbox{and}~~\mu_2\in[0,1]. 
\label{op-mu}
    \end{eqnarray}

Figure~\ref{move-D} depicts the details of moving from point $C$ to point $S$ by dividing point $D$ into three points, $D_1$, $D_2$ and $D_3$. Following discussions provide a sketch of the proof for the optimality of Fig.~\ref{move-D}, and its associated optimality conditions.

Using~\ref{eq6}, $\mu_1$  can be calculated as follows:
\begin{equation}
\mu_1 =\frac{\left({b \hat{P}_1 +1}\right) \left({\hat{P}_2-b-ab \hat{P}_2+1}\right)}{b \hat{P}_2 \left({\hat{P}_1+a \hat{P}_2+1}\right)}.
\label{DD1}
\end{equation}

From~\ref{DD1}, it is concluded that
\begin{align}
\frac{\partial \mu_1}{\partial \hat{P}_2}&=-\frac{\left({b \hat{P}_1 +1}\right) \left({\hat{P}_1-b+2a \hat{P}_2 -b \hat{P}_1+a \hat{P}^2_2-a^2b \hat{P}^2_2-2ab \hat{P}_2+1}\right)}{b \hat{P}_2 \left({\hat{P}_1+a\hat{P}_2+1}\right)^2}&\nonumber \\
&=\frac{\left({b\hat{P}_1+1}\right)\left({\left({a^2b-a}\right)\hat{P}^2_2+2a(b-1)\hat{P}_2+(b-1)(1+\hat{P}_1)}\right)}{b\hat{P}_2^2 \left({\hat{P}_1+a \hat{P}_2+1}\right)^2}. &
\label{DD3}
\end{align}
As the terms $(b\hat{P}_1+1)$  and $b\hat{P}_2  \left({\hat{P}_1+a \hat{P}_2+1}\right)^2$ in \ref{DD3} are positive, we focus on the second order term 
\begin{equation}
\left({a^2b-a}\right)\hat{P}^2_2+2a (b-1) \hat{P}_2+(b-1)(1+\hat{P}_1).
\label{DD4}
\end{equation}
Since $a,b <1$, it follows that $(a^2 b -a)<0$ . Expression~\ref{DD4} has two  roots with a sum of $\frac{2(1-b)}{(ab-1)}<0$ and a product of $\frac{(b-1)(1+\hat{P}_1)}{a(ab-1)}>0$, consequently, the two roots of~\ref{DD4} are negative. This means 
\begin{equation}
\left({a^2b-a}\right)\hat{P}^2_2+2a (b-1) \hat{P}_2+(b-1)(1+\hat{P}_1)<0 ~\Longrightarrow~\frac{\partial \mu_1}{\partial \hat{P}_2}<0,~~\forall \hat{P}_2>0
\label{DD5}
\end{equation}
resulting in, 
\begin{equation}
\frac{\partial \mu_1}{\partial \hat{P}_2}<0~\mbox{(always)}~
\xRightarrow{\mbox{Applying~exchanges~in~\ref{changes}}} 
\frac{\partial \mu_2}{\partial \hat{P}_1}<0~\mbox{(always)}.
\label{DD5p}
\end{equation}
From~\ref{DD1}, for $\frac{\partial \mu_1}{\partial \hat{P}_1}$, we have  
\begin{align}
\frac{\partial \mu_1}{\partial \hat{P}_1}&=-\frac{a^2 b^2 \hat{P}^2_2-ab \hat{P}^2_2+2ab^2 \hat{P}_2-2ab \hat{P}_2-b \hat{P}_2+\hat{P}_2+b^2-2b+1}{b \hat{P}_2 \left({\hat{P}_1+a\hat{P}_2+1}\right)^2}&\nonumber \\
&=-\frac{\left({a^2 b^2 - ab}\right)\hat{P}^2_2+\left({2ab^2-2ab-b+1}\right) \hat{P}_2+b^2-2b+1}{b \hat{P}_2 \left({\hat{P}_1+a\hat{P}_2+1}\right)^2}.&
\label{DD6}
\end{align}
In~\ref{DD6}, the term $b \hat{P}_2 \left({\hat{P}_1+a\hat{P}_2+1}\right)^2$  in the denominator is positive. Roots of the term 
\begin{equation}
\left({a^2 b^2 - ab}\right)P^2_2+\left({2ab^2-2ab-b+1}\right) \hat{P}_2+b^2-2b+1
\end{equation} 
in~\ref{DD6} can be computed as follows:
\begin{equation}
\hat{P}_2=\frac{-(2ab^2-2ab-b+1) \pm \sqrt{(2ab^2-2ab-b+1)^2-4(b^2-2b+1)(a^2 b^2 - ab)}}{2(a^2 b^2-ab)}.
\label{DD7}
\end{equation}
Expression~\ref{DD7} has one positive and one negative root, since the multiplication of its roots has negative value of $\frac{-(b-1)^2}{(ab-a^2b^2)}$. Computing the roots of~\ref{DD7}, it follows that $\frac{\partial \mu_1}{\partial \hat{P}_1}>0$ if 
\begin{align}
\hat{P}_2 &< \frac{b-1}{ab-1}<0\\
\hat{P}_2 &> \frac{(1-b)}{ab}>0.
\end{align}
As a result,
\begin{equation} 
\frac{\partial \mu_1}{\partial \hat{P}_1}>0~~\mbox{for}~~\hat{P}_2 > \frac{1-b}{ab}
 ~\xRightarrow{\mbox{Applying~exchanges~in~\ref{changes}}}
\frac{\partial \mu_2}{\partial \hat{P}_2}>0~~\mbox{for}~~\hat{P}_1 > \frac{1-a}{ab}.
\label{DD8p}
\end{equation}

As mentioned earlier, the focus of discussions is on the structure of the lower part of the boundary. Arguments for the upper part will be very similar and can be obtained by applying exchanges in \ref{changes} to what is concluded/computed for the lower part. Focusing on the lower part, in the following, it is assumed $\hat{P}_2\geq T_2=\frac{1-b}{ab}$ (otherwise, the lower part of the boundary will trivially end prior to point $D_1$ in Fig.~\ref{move-D})\footnote{Referring to Fig.~\ref{move-D}, at points $D_1$ and $S$, we have $\hat{P}_2=T_2$, while $\hat{P}_2>T_2$ all other points as one moves along the lower part from point $D_1$ to point $S$. }. It is also assumed $\hat{P}_1>T_1=\frac{1-a}{ab}$ (otherwise, the lower part of the boundary will be trivially composed of only a private message of power $P_1$ for user 1).  This means,  from $D_1$ and $S$, we have 

\begin{equation}
\frac{\partial \mu_1}{\partial \hat{P}_1}>0~\mbox{(always)}~
\xRightarrow{\mbox{Applying~exchanges~in~\ref{changes}}} 
\frac{\partial \mu_2}{\partial \hat{P}_2}>0~\mbox{(always)}.
\label{DD5p2}
\end{equation}

To provide some insight into the behavior of the boundary, Table \ref{table1} summarizes the changes in $\hat{P}_1$ and $\hat{P}_2$ following points specified on Fig.~\ref{move-D}. Notations $\mu(D_2)$ and $\mu(D_3)$ are used to show the slope at points $D_2$ and $D_3$, respectively. 
\begin{table}[htbp]
\begin{center} 
\begin{tabular}{||c|c|c|c|c||}   \hline
$A \rightarrow  D_2$ & $\hat{P}_1=P_1$ 
& $\hat{P}_2\nearrow$   & 
$\mu_1$ not valid; $\mu_2\nearrow \mu(D_2)$  & 
$X_1$ is entirely private  \\ \hline \hline
$D_2 \rightarrow  D_3$ & $\hat{P}_1=P_1$ 
& $\hat{P}_2 \nearrow \breve{T}_2$  & 
$\mu_1\searrow \mu(D_3)$~\mbox{and}~$\mu_2\nearrow \mu(D_3)$ & 
$X_1$ is entirely private  \\ \hline \hline
$D_3\rightarrow S$ & $\hat{P}_1\searrow T_1$  & $\hat{P}_2\searrow T_2$  &
$\mu_1\nearrow 1$~\mbox{and}~$\mu_2\nearrow 1$ & MAC1, MAC2 formed\\ 
\hline
\end{tabular}
\end{center}
\caption{Changes in $\hat{P}_1$ and $\hat{P}_2$ following points specified on Fig.~\ref{move-D}. In moving from point $A$ towards point $D_3$, $\hat{P}_2$ increases, reaching its maximum of $\breve{T}_2$ at $D_3$, and $\hat{P}_1$ remains at its maximum possible value of $P_1$.  
As $\mu=\min(\mu_1,\mu_2)$ increases beyond point $D_3$, $\hat{P}_1$ starts reducing from its maximum (boundary) value of $P_1$, ultimately reaching to $T_1$ at point $S$,  and $\hat{P}_2$ reduces from $\breve{T}_2$, ultimately reaching $T_2$ at point $S$.}
\label{table1}
\end{table}

 \subsubsection{Step by Step Covering of the Lower Part of the Boundary} 

 This article is concerned with conditions that the lower part of the boundary posses its richest construction, in the sense that the points from $A$ to $S$ are traversed as depicted in Fig.~\ref{move-D} and the lower part of the boundary finally becomes tangent to the sum-rate front. This requires imposing some conditions, which are derived throughout the article, and are shown to be feasible (by adjusting the power budgets vs. channel gains $a$ and $b$). Other cases for the structure of the lower part that are not discussed here can be concluded by studying the scenarios  that each of the derived conditions fails to be satisfied. 

As we will see later, to realize the richest structure for the lower part of the boundary curve, we are interested in the conditions that $\mu_1$ is a decreasing function of $\hat{P}_2$ for given $\hat{P}_1$ (always true referring to~\ref{DD5p}),  and  an increasing function of $\hat{P}_1$ for given $\hat{P}_2$  (always true referring to~\ref{DD5p2}). Likewise, $\mu_2$ is a decreasing function of $\hat{P}_1$ for given $\hat{P}_2$ (always true referring to~\ref{DD5p}),  and  an increasing function of $\hat{P}_2$ for given $\hat{P}_1$  (always true referring to~\ref{DD5p2}). Exceptions to this structure happens when the lower part of the boundary reaches the sum-rate front prior to completing the structure shown in Fig.~\ref{move-D}. This can happen in two forms: (i) The lower part of the boundary, prior to completing the sequence of points shown in Fig.~\ref{move-D}, intersects with the sum-rate front at an angle other than $\pi/4$. (ii) The lower part of the boundary, prior to completing completing the sequence of points shown in Fig.~\ref{move-D}, becomes tangent to the sum-rate front, i.e., it reaches the sum-rate front at an angle of $\pi/4$.  

\begin{theorem}
Assuming: (i) $P_2>\breve{T}_2>T_2$,  and (ii) $P_1(1-b)<P_2(1-a)$, movement from point $C$ to point $S$ is as shown in Fig.~\ref{move-D}. 
\end{theorem}

First, it should be noted that, from Eq.~\ref{DD5p}, $\mu_1\leq 1$ is a decreasing function of $\hat{P}_2$ for fixed $\hat{P_1}$. 
Formation of the segment from point $C$ to point $D_1$ can be easily verified by manipulating Eq.~\ref{eq6} and  Eq.~\ref{eq9}. It follows that, from point $C$ to point $D_1$,  Eq.~\ref{eq9}, which is satisfied with equality, acts as the sole bottleneck in determining $\mu=\mu_2$. From point $D_1$ to point $D_2$, we have $\mu_1>1$, which means Eq.~\ref{eq9} continues to act as the sole bottleneck in determining $\mu=\mu_2$. 
From earlier discussions, in this range $\mu_1$ monotonically decreases (see Eq.~\ref{DD5p}) and $\mu_2$ monotonically increases (see Eq.~\ref{DD5p2}). Finally, at point $D_2$, $\mu_1$  decreases to one, meaning that beyond $D_2$, Eq.~\ref{eq6} has the potential to act as the bottleneck in determining $\mu$. In mathematical terms, solving equation~\ref{eq6} for $\mu_1$, it is  concluded that  
 \begin{equation}
\mu_1 \leq 1, ~~\mbox{if}~~\hat{P}_2\geq T_2.
\label{eq10}
\end{equation}
Similarly, solving equation~\ref{eq9} for $\mu_2$, it is concluded that
   \begin{equation}
\mu_2 \leq 1, ~~\mbox{if}~~\hat{P}_1\geq T_1.
\label{eq11}
\end{equation}
However, from point $D_2$ to point $D_3$, we have $\mu_1>\mu_2$, and consequently, $\mu_2$ and Eq.~\ref{eq9} continue to govern the bottleneck. From point $D_2$ to point $D_3$, $\hat{P}_1$ remains at the maximum value of $P_1$ and $\hat{P}_2$ continues to increase, resulting in further reduction in $\mu_1<1$ and further increase in $\mu_2<1$. It follows that, eventually, at some point $D_3$, the increasing $\mu_2<1$ reaches the decreasing $\mu_1<1$. At point $D_3$, we have $\hat{P}_2=\breve{T}_2>T_2$ and $\hat{P}_1=P_1$. The value of $\breve{T}_2$ can be computed by setting the value of $\mu_1$ from Eq.~\ref{eq6} equal to the value of $\mu_2$ from Eq.~\ref{eq9}, and replacing $\hat{P}_1=P_1$. This results in,
\begin{align} \label{Mu1}
\mu_1=\mu_2 & ~\Longleftrightarrow~ \frac{\left(b P_1 +1\right) \left(\hat{P}_2-b-ab \hat{P}_2+1\right)}{b \hat{P}_2 \left(P_1+a \hat{P}_2+1\right)}=
\frac{\left(a \hat{P}_2 +1\right) \left(P_1-a-ab P_1+1\right)}{a P_1 \left(\hat{P}_2+b P_1+1\right)}. 
\end{align}
Equation~\ref{Mu1} can be solved in terms of $\hat{P}_2$ with a  positive root equal to $\breve{T}_2$. Exact solution to this  equation is not discussed in this work.  Finally, from point $D_3$ to point $S$ (start of the sum-rate front approached from the lower part of the boundary), the value of $\mu_1$ from Eq.~\ref{eq6} and the  value of $\mu_2$ from Eq.~\ref{eq9} continue to be equal. In moving from point $D_3$ to point $S$, subsequent infinitesimal steps are formed by alternating between 
Eq.~\ref{eq6} and Eq.~\ref{eq9} as the determining factor (bottleneck) governing movements of $\delta$-layers of power. The first $\delta$-step beyond point $D_3$ should be accompanied by an increase in $\mu_2$. This goal can be achieved by either continuing to increase $\hat{P}_2$ (similar to $\delta$-steps prior to point $D_3$), or by decreasing  $\hat{P}_1$. Note that  in this range, $\mu_2$ is a decreasing function of  $\hat{P}_1=P_1$ (for a fixed $\hat{P}_2$) and an increasing function of $\hat{P}_2$ (for a fixed $\hat{P}_1$). Let us focus on the first $\delta$-step beyond point $D_3$.  Let us first assume the first $\delta$-step is taken, like steps before point $D_3$, by further  increasing $\hat{P}_2$. This would reduce $\mu_1$, resulting in $\mu_1\lnapprox\mu_2$. Consequently, $\mu_1$ would become the bottleneck at this new point, meaning that subsequent step should reduce $\mu_2$ and/or increase $\mu_1$. This can be achieved by either reducing $\hat{P}_2$, or by increasing $\hat{P}_1=P_1$. The first option of reducing $\hat{P}_2$ would result in returning to point $D_3$, and the second option of increasing $\hat{P}_1$ is not possible as $\hat{P}_1$ is already at its maximum possible value of $P_1$. In the language of computing optimum solution using derivatives, this situation entails solution with respect to $\hat{P}_1$ is located  on the boundary of feasible region of the underlying optimization problem. In summary, this means the the first $\delta$-step beyond point $D_3$ cannot be based on increasing $\hat{P}_2$. In conclusion, we have: (i) the first $\delta$-step beyond point $D_3$ shall be taken by reducing $\hat{P}_1$, causing an increase in $\mu_2$ and a reduction in $\mu_1$, resulting in $\mu_1$ becoming the new bottleneck. (ii) The $\delta$-step beyond this first step should be taken by reducing $\hat{P}_2$, resulting in an increase in $\mu_1$ and a reduction in $\mu_2$, causing $\mu_2$ to become the new bottleneck. In subsequent $\delta$-steps from point $D_3$ to point $S$, the above two steps of (i) and (ii) alternate, further reducing both $\hat{P}_1$ and $\hat{P}_2$. Finally, at point $S$, we have $\hat{P}_1=T_1$ and $\hat{P}_2=T_2$.  

In summary, in moving from point $D_1$ to point $S$, initially  $\hat{P}_1$ remains fixed at its maximum possible value of $P_1$ and $\hat{P}_2$ increases until $\hat{P}_2=\breve{T}_2$, and then, both $\hat{P}_1$ and $\hat{P}_2$ decrease until reaching the point $S$. As will be discussed later (see Eq.~\ref{sum1} and Eq.~\ref{sum2}), the case shown in Fig.~\ref{move-D} where point $S$ is realized on the lower part of the boundary occurs if $P_1(1-b)<P_2(1-a)$. 

Steps along the sum-rate front do not entail changing the amount of power allocated to private and public messages of the two users, it only proceeds by relying on different layered code-books for the public parts to achieve different points on the sum-rate fronts. Sum-rate front finally reaches point $\breve{S}$ in Fig.~\ref{move-D}, namely the first point on the sum-rate front when the upper part of the capacity boundary is traversed clockwise. Next, we discuss the structure of the sum-rate front in more details.   

Figure~\ref{sum-rate-simple} depicts structures of Gaussian code-books corresponding to the public messages at point $S$, and Fig.~\ref{private} depicts the structures of Gaussian code-books for corresponding private messages. It follows that the total rate contributed to $R_{sum}$ from private messages is equal to: $0.5\log_2(1/ab)$. The total rate contributed to $R_{sum}$ from public messages is governed by the smaller of the two sum-rate values corresponding to user 1 (sum-rate in MAC1) and user 2 (sum-rate in MAC2), as depicted in Fig.~\ref{sum-rate-simple}. The reason is that public messages in Fig.~\ref{sum-rate-simple} should be decoded at both receivers, and consequently, the smaller of the two sum-rate (intersections of the MAC channels in Fig.~\ref{sum-rate-simple}) will determine the contribution of the public messages to the overall sum-rate. It follows that, 

\begin{equation}
R_{sum}=\min\left[0.5\log_2\left(\frac{1}{ab}+P_1-T_1+a(P_2-T_2)\right), 0.5\log_2\left(\frac{1}{ab}+P_2-T_2+b(P_1-T_1)\right)\right]
\label{sum1}
\end{equation}
or, replacing  for $T_1=\frac{1-a}{ab}$ and $T_2=\frac{1-b}{ab}$ , 
\begin{eqnarray}
R_{sum}=0.5\log_2\left(\frac{1}{ab}+P_1-T_1+a(P_2-T_2)\right),~~\mbox{if}~~P_1(1-b)<P_2(1-a),  \\
R_{sum}=0.5\log_2\left(\frac{1}{ab}+P_2-T_2+b(P_1-T_1)\right),~~\mbox{if}~~P_1(1-b)>P_2(1-a).
\label{sum2}
\end{eqnarray}
For $P_1(1-b)<P_2(1-a)$, the sum-rate front coincides with the sum-rate front of the MAC channel formed at receiver 1, and  the lower part becomes tangent to the sum-rate front at point $S$. In this case, the lower part of the capacity region will be as shown in Fig.~\ref{move-D} (point $S$ falls on the boundary, while point $\hat{S}$ falls outside boundary) with code-books as depicted in Figs.~\ref{sum-rate-simple} and \ref{private}. 

For $P_1(1-b)>P_2(1-a)$, sum-rate front would be governed by receiver 2. In this case, the lower part of the boundary, prior to reaching to point $S$, would intersect with the sum-rate front at an angle other than $\pi/4$, and consequently,  point $S$ would fall outside the boundary. On the other hand, the upper part of the curve would become tangent to the sum-rate front at point $\hat{S}$. 

This entails either the lower part is tangent to the sum-rate front and the upper part intersects with the sum-rate front, or vice versa. It follows that, if $P_1(1-b)=P_2(1-a)$, then both the lower part and the upper part will be tangent to the sum-rate front. Note that for $P_1=T_1$ and $P_2=T_2$, the boundary coincides with the sum-rate front. It is easy to verify that the condition $(P_1,P_2)=(T_1,T_2)$, results in $P_1(1-b)=P_2(1-a)$, as expected. 

From the above discussions, it follows that (under the condition of Fig.~\ref{move-D}), the value of $\mu$ is a continuous function in the range $[\mu_A,1]$ in the lower part of the boundary, and a continuous function in the range  $[\mu_{\hat{S}},\mu_{\hat{A}}]$ in the upper part of the boundary, where ${\hat{A}}$ is the starting (corner) point on the upper part of the boundary (with maximum possible value for $R_2$ irrespective of the value of $R_1$).

  \begin{figure}[htp]
   \centering
   \includegraphics[width=0.45\textwidth]{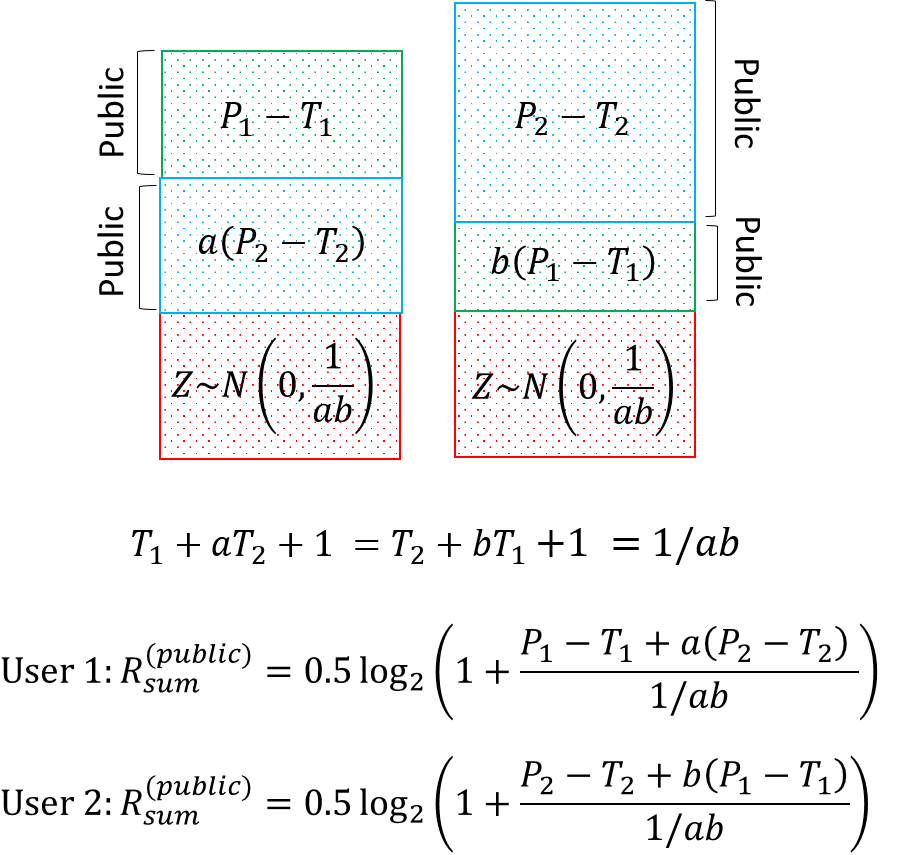}
   \caption{Structure of Gaussian code-books corresponding to public messages at point $S$.}
   \label{sum-rate-simple}
 \end{figure}

  \begin{figure}[htp]
   \centering
   \includegraphics[width=0.55\textwidth]{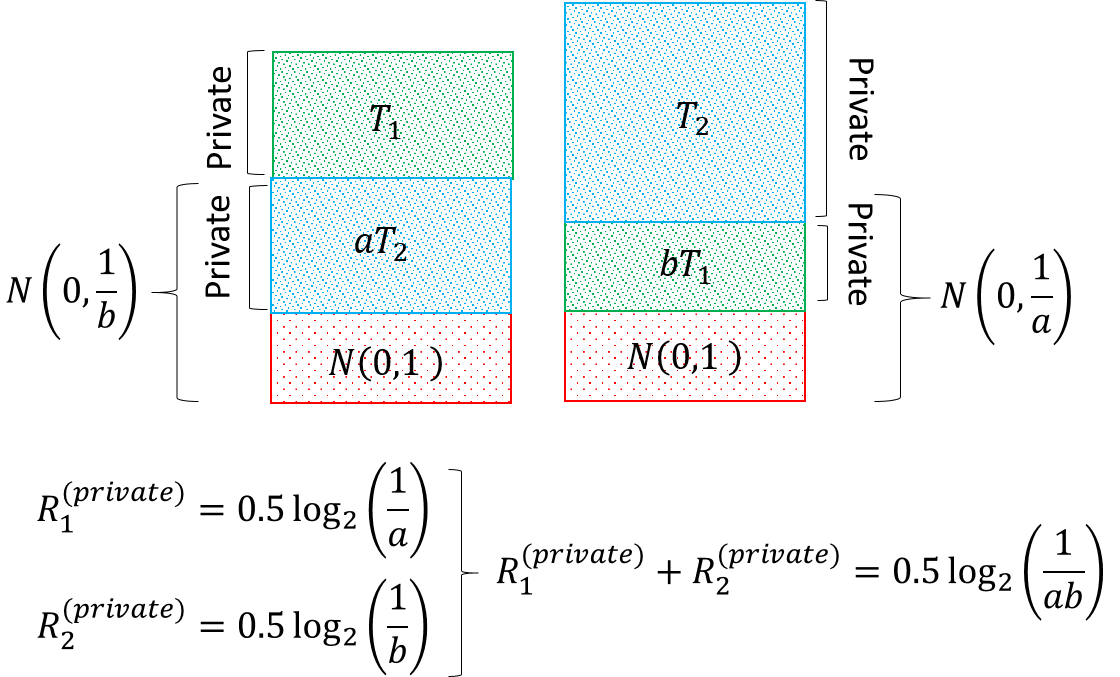}
   \caption{Structure of Gaussian code-books corresponding to private messages at point $S$.}
   \label{private}
 \end{figure}

 \subsubsection{Some Special Cases: $P_1=T_1$ and/or $P_2=T_2$}  

 In the following, two special cases of $P_1=T_1$ and/or $P_2=T_2$ will be discussed. 
Let us consider formation of the lower part starting from point $A$ in the case of $P_1=T_1$. At point $A$, message of user 1 is entirely private, as a result,
\begin{equation}
 P_1=T_1~\Longrightarrow~ \hat{P}_1=T_1.
\end{equation}
Referring to computations in Fig.~\ref{sum-rate-tradeoff2}, it is concluded that for $P_1=T_1$, regardless of how the power of of user 2 is divided between its public and private components, the sum-rate remains the same as that of point $A$. This means, for $P_1=T_1$, the lower part of the boundary coincides with the sum-rate front. Likewise, for $P_2=T_2$, the upper part of the boundary coincides with the sum-rate front (see Fig.~\ref{sum-rate-tradeoff}).  It follows that, as mentioned earlier, if $P_1=T_1$ and $P_2=T_2$, the sum-rate front forms the entire boundary. 
  
          \begin{figure}[htp]
   \centering
   \includegraphics[width=0.55\textwidth]{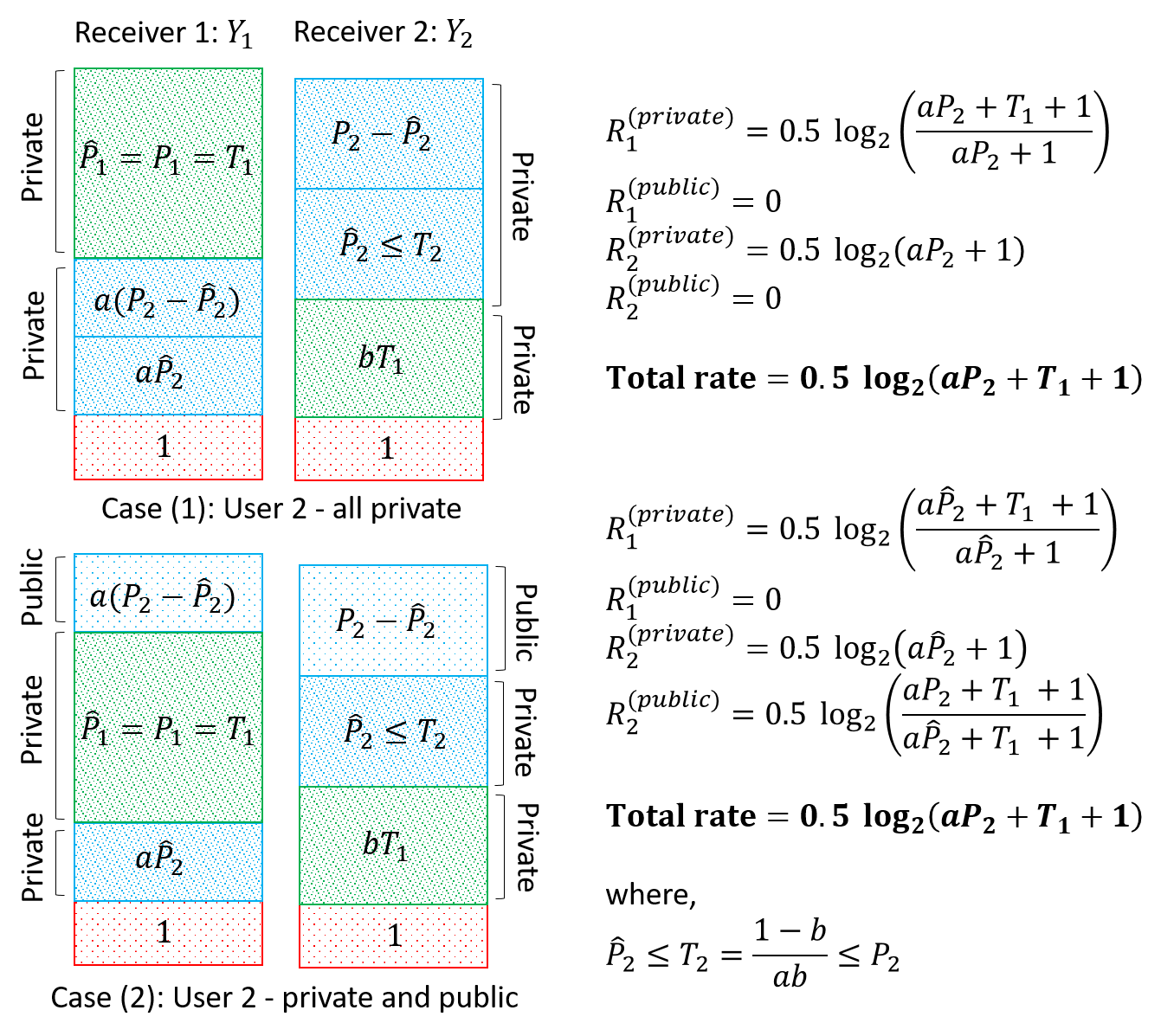}
   \caption{For $P_1=T_1$, moving along the lower part of the boundary (coinciding with sum-rate front)  by changing the ratio of power allocated to private and public messages for user 2.}
   \label{sum-rate-tradeoff2}
 \end{figure} 

          \begin{figure}[htp]
   \centering
   \includegraphics[width=0.55\textwidth]{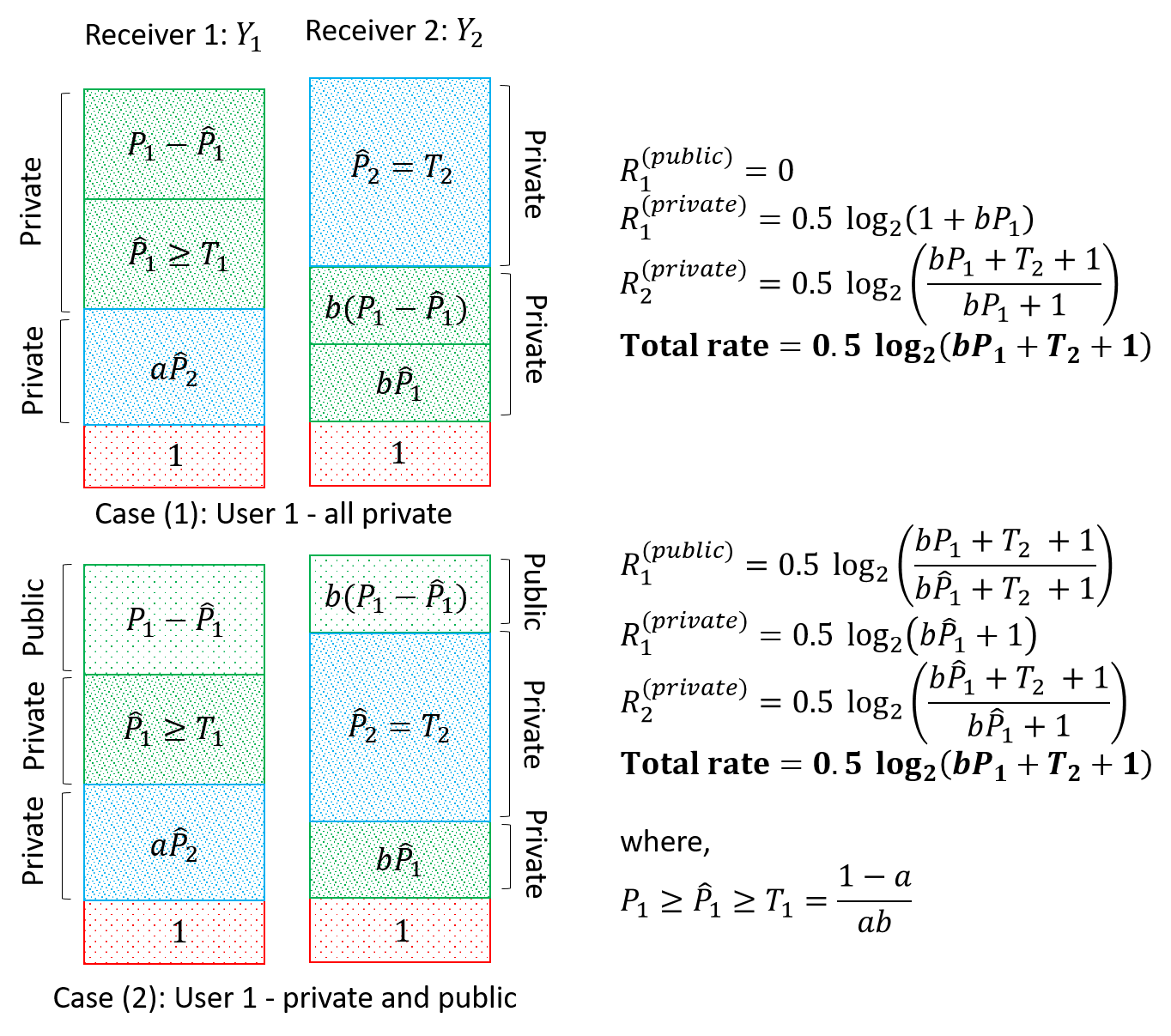}
   \caption{For $P_2=T_2$, moving along the upper part of the boundary (coinciding with sum-rate front)  by changing the ratio of power allocated to  private and public messages for user 1.}
   \label{sum-rate-tradeoff}
 \end{figure}
 
      \subsection{Intersection of Multiple Access Channels due to Public Messages}\label{sec-inter}

Consider scenarios that both user 1 and user 2 has a non-zero public message.  Considering private messages as noise, the public messages $U_1$ and $U_2$ form a MAC1 at $Y_1$ and a MAC2 at $Y_2$. The requirement for decoding of the two public messages at both receivers entails that $U_1$ and $U_2$ should be at the intersection of the two MACs formed at the two receivers. Let us rely on Fig.~\ref{MAC1-2d}\,(1) to define the parameters involved in determining the shape of the intersection of the two MACs.  
In general, the shape of the intersection is governed by  relative positions of rate values $R^{+}_1(1)$, $R^{-}_1(1)$, $R^{+}_1(2)$, $R^{-}_1(2)$ along the $R_1$-axis, and relative positions of rate values $R^{+}_2(2)$, $R^{-}_2(2)$, $R^{+}_2(1)$, $R^{-}_2(1)$ along the $R_2$-axis.   To compare these rate values,  we compare the  corresponding Signal-to-Noise Ratios (noting that all rate values are computed as the capacity of a relevant AWGN channel, formed as part of a corresponding MAC). We have,

          \begin{eqnarray}
           \mbox{SNR}_{ R_1^{+}(1)} & = & \frac{\rho P_1}{\sigma_1^2} \\
 \mbox{SNR}_{R_1^{-}(1)} & = & \frac{\rho P_1}{a \theta P_2+\sigma_1^2} \\
              \mbox{SNR}_{ R_1^{+}(2)} & = & \frac{b\rho P_1}{\sigma_2^2} \\
 \mbox{SNR}_{R_1^{-}(2)} & = & \frac{b\rho P_1}{\theta P_2+\sigma_2^2} \\
           \mbox{SNR}_{R_2^{+}(2)} & = & \frac{\theta P_2}{\sigma_2^2} \\
            \mbox{SNR}_{R_2^{-}(2)} & = & \frac{\theta P_2}{b\rho P_1+\sigma_2^2}  \\
              \mbox{SNR}_{R_2^{+}(1)} & = & \frac{a\theta P_2}{\sigma_1^2} \\
 \mbox{SNR}_{R_2^{-}(1)} & = & \frac{a\theta P_2}{\rho P_1+\sigma_1^2}  
           \end{eqnarray}
           where,
             \begin{align} \label{sigma1-sigma2}
 \sigma_1^2=(1-\rho)P_1+a (1-\theta) P_2+1\\ \label{sigma1-sigma2p}
  \sigma_2^2=b (1-\rho) P_1+(1-\theta)P_2+1. 
 \end{align}
  Note that, under the conditions that the lower part of the boundary is as depicted in Fig.~\ref{move-D}, at point $S$ we have (see Fig.~\ref{sum-rate-simple}),
\begin{equation}
\sigma_1^2=\sigma_2^2=\frac{1}{ab}.
\label{1ab}
\end{equation}

 Obviously,
   \begin{eqnarray} 
 R^{+}_1(1) & > & R^{-}_1(1) \\
  R^{+}_2(2)& > & R^{-}_2(2)\\
  R^{+}_1(2)& > & R^{-}_1(2)\\
  R^{+}_2(1)& > & R^{-}_2(1).
     \end{eqnarray} 
 To determine the shape of the intersection, we need to compute:
\begin{align}
\mbox{SGN}\!\left(R^{+}_1(1)-R^{+}_1(2)\right)  & ~~=  \\
\mbox{SGN}\!\left(\frac{\rho P_1}{\sigma_1^2}-\frac{b\rho P_1}{\sigma_2^2}\right) & ~~= \\
\mbox{SGN}\!\left(\sigma_2^2-b\sigma_1^2\right) & ~~= \\
\mbox{SGN}\!\left( b (1-\rho) P_1+(1-\theta)P_2+1 - b(1-\rho)P_1-ab (1-\theta) P_2-b\right) & ~~=  \\
\mbox{SGN}\!\left((1-ab)(1-\theta) P_2+1-b\right) & ~~=~~+1  
\label{eq62}
 \end{align} 
This means, 
 \begin{eqnarray} 
R^{+}_1(1) >R^{+}_1(2)~~~\mbox{always}.
\label{EqHKn1} 
\end{eqnarray} 

    \begin{eqnarray} \label{eq-SGN4}
\mbox{SGN}\!\left(R^{+}_1(2)-R^{-}_1(1)\right) & = \\ \label{eq-SGN0}
\mbox{SGN}\!\left(\frac{b\rho P_1}{\sigma_2^2}-\frac{\rho P_1}{a\theta P_2+\sigma_1^2}\right)  & =  \\ \label{eq-SGN1}
\mbox{SGN}\!\left(\frac{1}{(1-\rho) P_1+\frac{1}{b}(1-\theta)P_2+\frac{1}{b}}-
\frac{1}{a\theta P_2+(1-\rho) P_1+a(1-\theta) P_2+1} \right)& =  \\ \label{eq-SGN2}
\mbox{SGN}\!\left(a\theta P_2+(1-\rho) P_1+a(1-\theta) P_2+1-(1-\rho) P_1-\frac{1}{b}(1-\theta)P_2-\frac{1}{b}\right) & =\\ \label{eq-SGN3}
\mbox{SGN}\!\left(P_2-\frac{1-b}{\theta+ab-1}\right) \\ \nonumber 
    \end{eqnarray} 
where \ref{eq-SGN1} is obtained by replacing $\sigma_1^2=(1-\rho)P_1+a (1-\theta) P_2+1$ and 
 $ \sigma_2^2=b (1-\rho) P_1+(1-\theta)P_2+1$ from  \ref{sigma1-sigma2} and \ref{sigma1-sigma2p} in~\ref{eq-SGN0}, and \ref{eq-SGN2} follows \ref{eq-SGN1} since the  denominators of the two fractions  in~\ref{eq-SGN1} are positive.

It follows from~\ref{eq-SGN3} that, 
    \begin{eqnarray} \label{Novel}
    R^{+}_1(2)>R^{-}_1(1)~~~\mbox{if}~~~P_2>\frac{1-b}{\theta+ab-1}.
      \end{eqnarray}
Expressing condition in~\ref{Novel} in terms of 
$\theta$ being feasible, i.e., $\theta\leq1$, we conclude  
\begin{equation}
 \theta\leq 1 \Longleftrightarrow P_2\geq \frac{1-b}{ab}=T_2
\label{Novel2}
\end{equation}
which is always valid. 
    Finally, 
    \begin{eqnarray} \label{eq-SGN5}
    \mbox{SGN}\!\left(R^{-}_1(1)-R^{-}_1(2)\right)=
        \mbox{SGN}\!\left(
        \frac{\rho P_1}{a\theta P_2+\sigma_1^2}-
        \frac{b\rho P_1}{\theta P_2+\sigma_2^2}\right) & = & \\ \label{eq-SGN6}
\mbox{SGN}\!\left(\frac{1}{a\theta P_2+(1-\rho) P_1+a(1-\theta) P_2+1} -
\frac{b}{\theta P_2+b(1-\rho) P_1+(1-\theta) P_2+1}\right) & = & \\ \label{eq-SGN7}
\mbox{SGN}\!\left(\frac{1}{(1-\rho) P_1+a P_2+1} -
\frac{1}{(1-\rho) P_1+\frac{1}{b} P_2+\frac{1}{b}}\right) & = & +1 \\ \nonumber 
            \end{eqnarray} 
where \ref{eq-SGN6} is obtained by direct replacement from \ref{sigma1-sigma2} and \ref{sigma1-sigma2p}, and  \ref{eq-SGN7} is concluded since $\frac{1}{b}>1>a$. 
  This means, from~\ref{eq-SGN7},  
 \begin{eqnarray}  \label{Eq71}
R^{-}_1(1)>R^{-}_1(2) ~~~\mbox{always}. 
\end{eqnarray} 
Similar to $R_1$-axis,  it follows that for $R_2$-axis, 
\begin{align} \label{eq64}
R^{+}_2(2) >R^{+}_2(1)~&~\mbox{always} \\ \label{eq64p}
R^{+}_2(1)>R^{-}_2(2)~&~\mbox{if}~~P_1>\frac{1-a}{\rho+ab-1} \\ \label{eq64z}
R^{-}_2(2)>R^{-}_2(1) ~&~\mbox{always}. 
\end{align}
Note that the conclusions in \ref{Eq71} and \ref{eq64z} are consistent with the results of Theorem~\ref{theorem3-TH}. 
Expressing condition in~\ref{eq64p} in terms of 
$\rho$ being feasible, i.e., $\rho\leq1$, we conclude  
\begin{equation}
 \rho\leq 1 \Longleftrightarrow P_1\geq \frac{1-a}{ab}=T_1
\label{Novel3}
\end{equation}
which is always valid. Noting above discussions, depending on 
 \begin{eqnarray} 
 R^{+}_1(2)\stackrel{?}{>}R^{-}_1(1) \\
R^{+}_2(1)\stackrel{?}{>}R^{-}_2(2) 
 \end{eqnarray} 
 there are four cases for the 
 intersection of the MAC channels, as shown in Fig.~\ref{MAC1-2d}. The conditions governing each case are summarized next. 

  \begin{eqnarray}  \label{eq115}
   \mbox{Case 1 in Fig.~\ref{MAC1-2d}:}  &  R^{-}_1(1)\geq R^{+}_1(2)~~\mbox{and}~~R^{-}_2(2)\geq R^{+}_2(1) \\  \label{eq115p}
      \mbox{Case 2 in Fig.~\ref{MAC1-2d}:}  &  R^{-}_1(1)\leq R^{+}_1(2)~~\mbox{and}~~ R^{-}_2(2)\geq R^{+}_2(1) \\   \label{eq115z}
         \mbox{Case 3 in Fig.~\ref{MAC1-2d}:}  & R^{-}_1(1)\geq R^{+}_1(2)~~\mbox{and}~~ R^{-}_2(2)\leq R^{+}_2(1) \\  \label{eq115t}
         \mbox{Case 4 in Fig.~\ref{MAC1-2d}:}  &  R^{-}_1(1)\leq R^{+}_1(2)~~\mbox{and}~~R^{-}_2(2)\leq R^{+}_2(1). 
      \end{eqnarray}

        \begin{figure}[htp]
   \centering
   \includegraphics[width=0.7\textwidth]{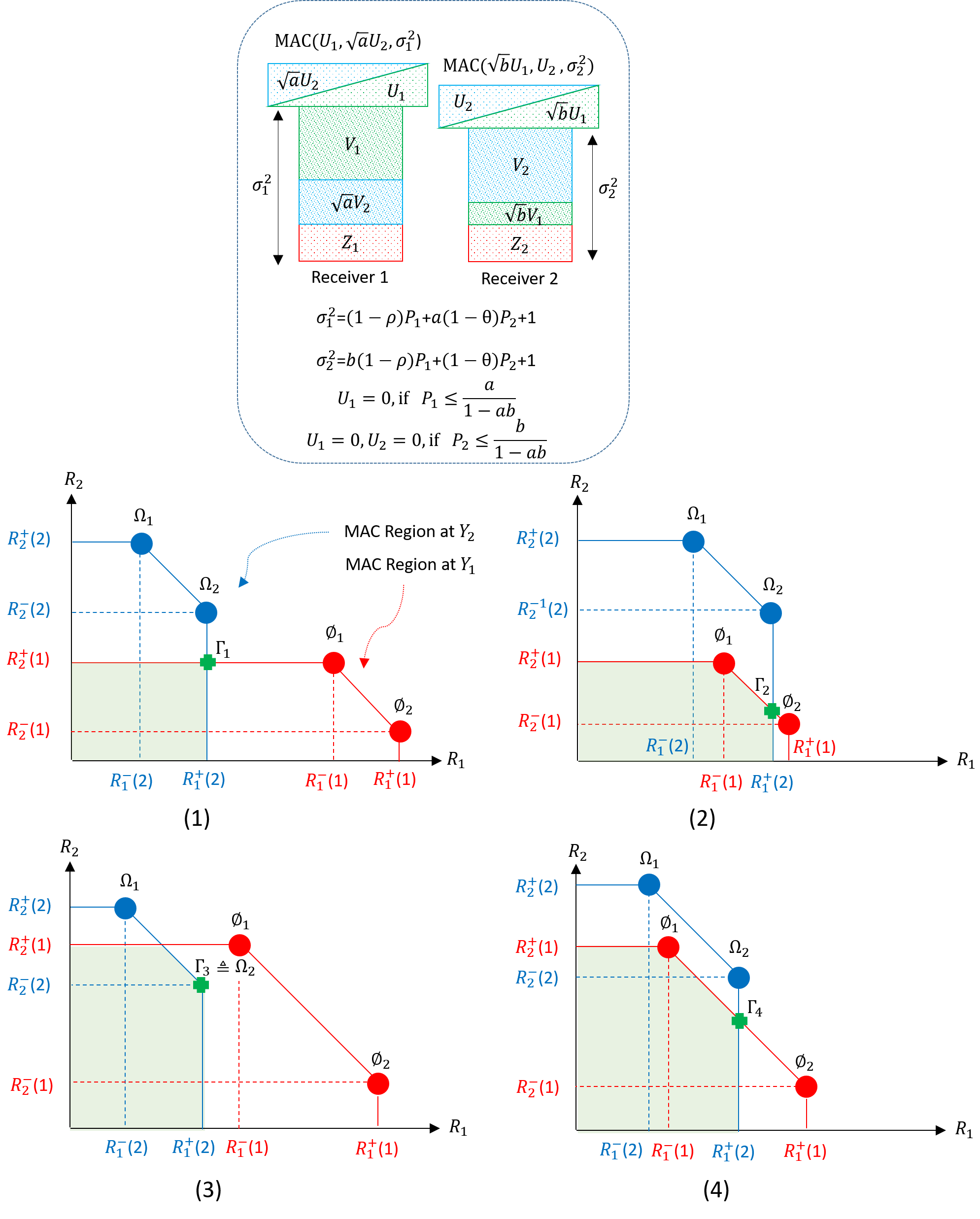}
   \caption{Four possible configurations for the structure of  MAC channels corresponding to public messages. Note that in Fig.~\ref{MAC1-2d}\,(1) point $\Gamma_1$ is the optimum solution regardless of the value of $\mu$, and in Fig.~\ref{MAC1-2d}\,(2)\,(3)\,(4) points $\Gamma_2$, $\Gamma_3$, $\Gamma_4$ are optimum, respectively, because $\mu\leq 1$ (this article focuses on the lower part of the boundary, i.e., $\mu\leq 1$). Points $\Gamma_2$ and $\Gamma_4$ in Fig.~\ref{MAC1-2d}\,(2)\,(4) necessitate joint decoding of public messages. Note that, if joint decoding is needed, it will be for both MACs, as it should realize a non-corner point on each of the two MAC regions located at their intersection, see Fig.~\ref{MAC1-2d}\,(2)\,(4). Fig.~\ref{MAC1-2d}\,(3) necessitates successive decoding (corner point of MAC at $Y_2$) and order of decoding is to maximize $R_1$ since $\mu\leq 1$.}
   \label{MAC1-2d}
 \end{figure}

Figure~\ref{MAC1-2d} presents four possible shapes for the intersection of the two MACs. For values of $\mu<1$, one is interested to realize the points specified as $\Gamma_1$ to $\Gamma_4$ in Fig.~\ref{MAC1-2d}.
In all four cases, public messages are successively decoded at receiver 2. For cases  in Fig.~\ref{MAC1-2d}(1) and Fig.~\ref{MAC1-2d}(3), public messages are successively decoded at the receiver 1, while for cases  in Fig.~\ref{MAC1-2d}(2) and Fig.~\ref{MAC1-2d}(4), public messages are jointly  decoded at receiver 1.   Relying on known properties of MAC,  to realize points $\Gamma_3$ and $\Gamma_4$ on the sum-rate front of receiver 1 in Fig.~\ref{MAC1-2d}-(2) and  Fig.~\ref{MAC1-2d}-(4),  joint decoding can be replaced by successive decoding if  a larger number of public messages (a higher number of layers) is used. 

By some rearrangement of terms, conditions in \ref{Novel} and \ref{eq64p} can be expressed as,
\begin{align} \label{modcond1}
  R^{+}_1(2) \geq  R^{-}_1(1) & ~\Longleftrightarrow~~\theta  \geq  \frac{1-b}{P_2}-ab+1 \\  \label{modcond2}
 R^{+}_2(1) \geq  R^{-}_2(2) & ~\Longleftrightarrow~~\rho  \geq  \frac{1-a}{P_1}-ab+1
 \end{align} 
respectively. Referring to Fig.~\ref{move-D}, MAC1 and MAC2 are formed from point $D_3$ to point $S$. Note that the right hand sides of \ref{modcond1} (and/or \ref{modcond2}) will be equal to one for $P_2=T_2$ (and/or $P_1=T_1$). Under the assumptions made in this article, that $P_1>T_1$ and $P_2>T_2$, the right hand sides of \ref{modcond1} and \ref{modcond2} will be less than one. 

 At point $S$, we have $\hat{P}_1=T_1=\frac{1-a}{ab}$ and 
$\hat{P}_2=T_2=\frac{1-b}{ab}$, resulting in 
\begin{align} \label{mod0}
\rho_S~=~1-\frac{T_1}{P_1}~=~1-\frac{1-a}{abP_1}~=~\frac{abP_1-1+a}{abP_1} \\ \label{mod0p}
\theta_S~=~1-\frac{T_2}{P_2}~=~1-\frac{1-b}{abP_2}~=~\frac{abP_2-1+b}{abP_2}.
\end{align}
The conditions in \ref{modcond1} and \ref{modcond2}, are respectively, expressed at point $S$ as
\begin{align} \label{mod1}
\rho_S~=~\frac{abP_1-1+a}{abP_1}~\stackrel{?}{>}~ \frac{1-a}{P_1}-ab+1 ~~\Longleftrightarrow~~ P_1~>~\frac{2(1-a)}{1+2ab}, \\ \label{mod1p}
\theta_S~=~\frac{abP_2-1+b}{abP_2}~\stackrel{?}{>}~ \frac{1-b}{P_2}-ab+1 ~~\Longleftrightarrow~~ P_2~>~\frac{2(1-b)}{1+2ab}. 
\end{align}
Noting $P_1>T_1=\frac{1-a}{ab}$ and  $P_2>T_2=\frac{1-b}{ab}$, it follows that in all cases we have, 
\begin{align} \label{mod2}
P_1~>~\frac{1-a}{ab}& ~>~ \frac{2(1-a)}{1+2ab} \\ \label{mod2p}
P_2~>~\frac{1-b}{ab}& ~>~ \frac{2(1-b)}{1+2ab}.
\end{align}
This means conditions in \ref{modcond1} and \ref{modcond2} are always satisfied at point $S$.  
Two remarks are worth emphasizing.

\noindent{\bf Remark 1:} Referring to Fig.~\ref{MAC1-2d}, it is observed at all optimum points, i.e.,  $\Gamma_1$, 
$\Gamma_2$, $\Gamma_3$ and $\Gamma_4$, we have,
\begin{align}
R_{U_1}& =R^{+}_1(2) \\
& =I(U_1;Y_2|U_2) \\
& =\frac{1}{2}\log_2\left(1+\frac{b\rho P_1}{(1-\theta)P_2+b(1-\rho)P_1+1}\right) \\
& =\frac{1}{2}\log_2\left(\frac{(1-\theta)P_2+bP_1+1}{(1-\theta)P_2+b(1-\rho)P_1+1}\right).~~~~~~~~~~~~~~~~\square
\label{eq91n}
\end{align}

\noindent{\bf Remark 2:} The conditions for Case 1 in Fig.~\ref{MAC1-2d}\,(1), are 
\begin{align} \label{modcond1pz}
  R^{+}_1(2) \leq  R^{-}_1(1) & ~\Longleftrightarrow~~\theta  \leq  \frac{1-b}{P_2}-ab+1 \\  \label{modcond2pz}
 R^{+}_2(1) \leq  R^{-}_2(2) & ~\Longleftrightarrow~~\rho  \leq  \frac{1-a}{P_1}-ab+1.
 \end{align} 
 Let us assume $\theta$ is increased by reducing $P_{V_2}$ and increasing $P_{U_2}$. Noting that for conditions shown in Fig.~\ref{MAC1-2d}\,(1), $U_1$ and $U_2$ are sequentially decoded, it follows that, as far as $R^{-}_1(1)$ is concerned, the signal power is $P_{U_1}$ and the noise power is $P_{V_1}+aP_{U_2}+aP_{V_2}=P_{V_1}+aP_2$.  This means, by increasing $\theta$, $R^{-}_1(1)$ will not be affected. On the other hand, by increasing $\theta$ such that,
\begin{equation}
\theta = \frac{1-b}{P_2}-ab+1~\Longrightarrow~R^{+}_1(2) = R^{-}_1(1).
\label{MoveR}
\end{equation}
This means point $\Gamma_1$ overlaps with point $\phi_1$, increasing $R_{ws}$. Figure~\ref{Move-R} depicts this case.  If $\theta$ is further increased, resulting in $R^{+}_1(2) > R^{-}_1(1)$, then the case shown in Fig.~\ref{MAC1-2d}\,(2) will be formed. Note that although this effect is explained in terms of increasing $\theta$, in some cases, $\rho$ is increased, causing the point $\Gamma_1$ to overlap with point $\Omega_2$.
{\em The final conclusion is that in all cases, the optimum point is located on the sum-rate front of} MAC1 {\em and/or on the sum-rate front of} MAC2. $\square$ 

        \begin{figure}[htp]
   \centering
   \includegraphics[width=0.4\textwidth]{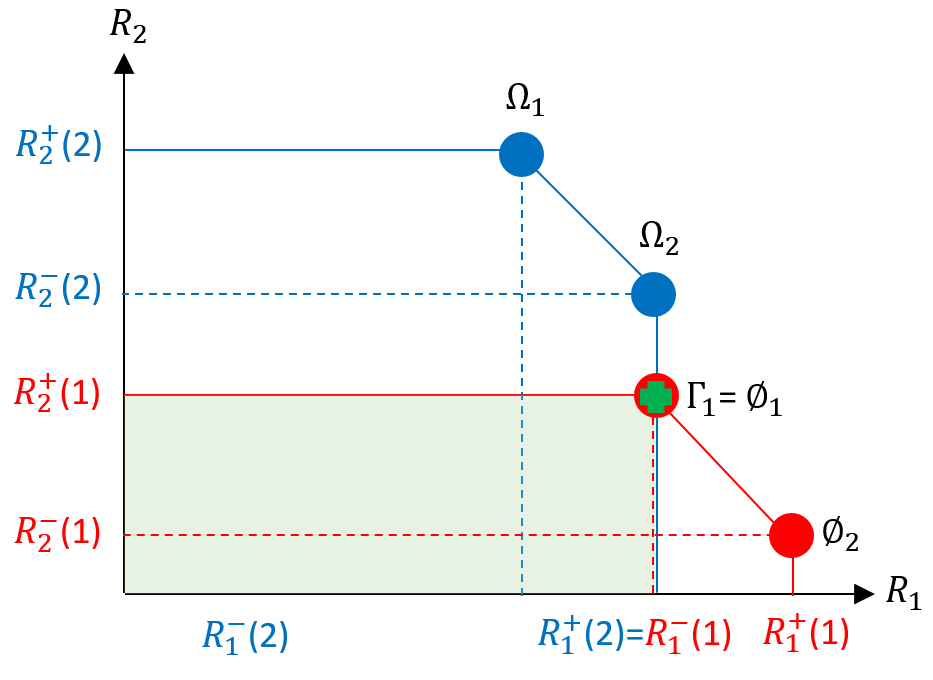}
   \caption{Configuration in Fig.~\ref{MAC1-2d}\,(1) where the optimum power allocation has shifted the point $\Gamma_1$ to overlap with point $\phi_1$.}
   \label{Move-R}
 \end{figure}

In computing the nested optimality condition for private messages, in Section~\ref{optimality-condition},  the change in $R_{ws}$ due to a composite $\delta$-step is computed as sum of two partial components. Each component accounts for the change in $R_{ws}$ due to the effect of the $\delta$-layer of one of the two users alone. This property was justified by dividing the boundary into segments with continuous slope, and applying the proof to each segment separately where the end point of a segment coincides with the starting point of the following segment (and thereby the proof could continue from segment to segment).  In the following, relying on capacity arguments, this issue is studied from a different angle. It is established that, considering the effect of two $\delta$-layers (that form a composite step) simultaneously, and forming the intersection of the underlying 
MAC regions with infinitesimal powers, results in the same outcome as relying on superposition of partial derivatives under the assumption of partial derivatives being continuous.

\subsubsection{Equivalency of a Composite $\delta$-step to Two Consecutive Simple $\delta$-steps: Continuity of Slope}\label{optimality-condition2}

Figure~\ref{MAC1-2e} depicts  two public $\delta$-layers, one for each of the two users,  used in deriving the conditions for nested optimality of private messages.  Consider a point on the curve where private messages are of powers $\breve{\hat{P}}_1$, $\breve{\hat{P}}_2$. As one moves along the lower part of the boundary,
$\breve{\hat{P}}_1= P_1 \searrow T_1$  and $\breve{\hat{P}}_2=0\nearrow \breve{T}_2 \searrow T_2$.  A composite $\delta$-step, defined when power allocations of both users are changed, occur in the range of moving from point $D_3$ to point $S$ (see Fig.~\ref{move-D}). In this range, powers of private messages are reduced, and in turn, powers of public messages are increased.  Note that the public layers formed prior to this composite $\delta$-step see the same noise plus interference levels before and after the step. Consequently,  public layers that were formed before taking the step can be decoded without affecting their  contribution to $R_{ws}$. As a result, the actual change in $R_{ws}$ will be due to terms computed next.  

To compute such terms, let us consider Fig.~\ref{MAC1-2e}. As mentioned earlier, public messages formed prior to the composite $\delta$-step are decoded and removed. $\hat{P}_1$ and $\hat{P}_2$ are the remaining powers of private messages after deducting $\delta P_1$ and $\delta P_2$, where $\delta P_1$ and $\delta P_2$ are used to send the new  infinitesimal public messages. We have:

         \begin{figure}[htp]
   \centering
   \includegraphics[width=0.45\textwidth]{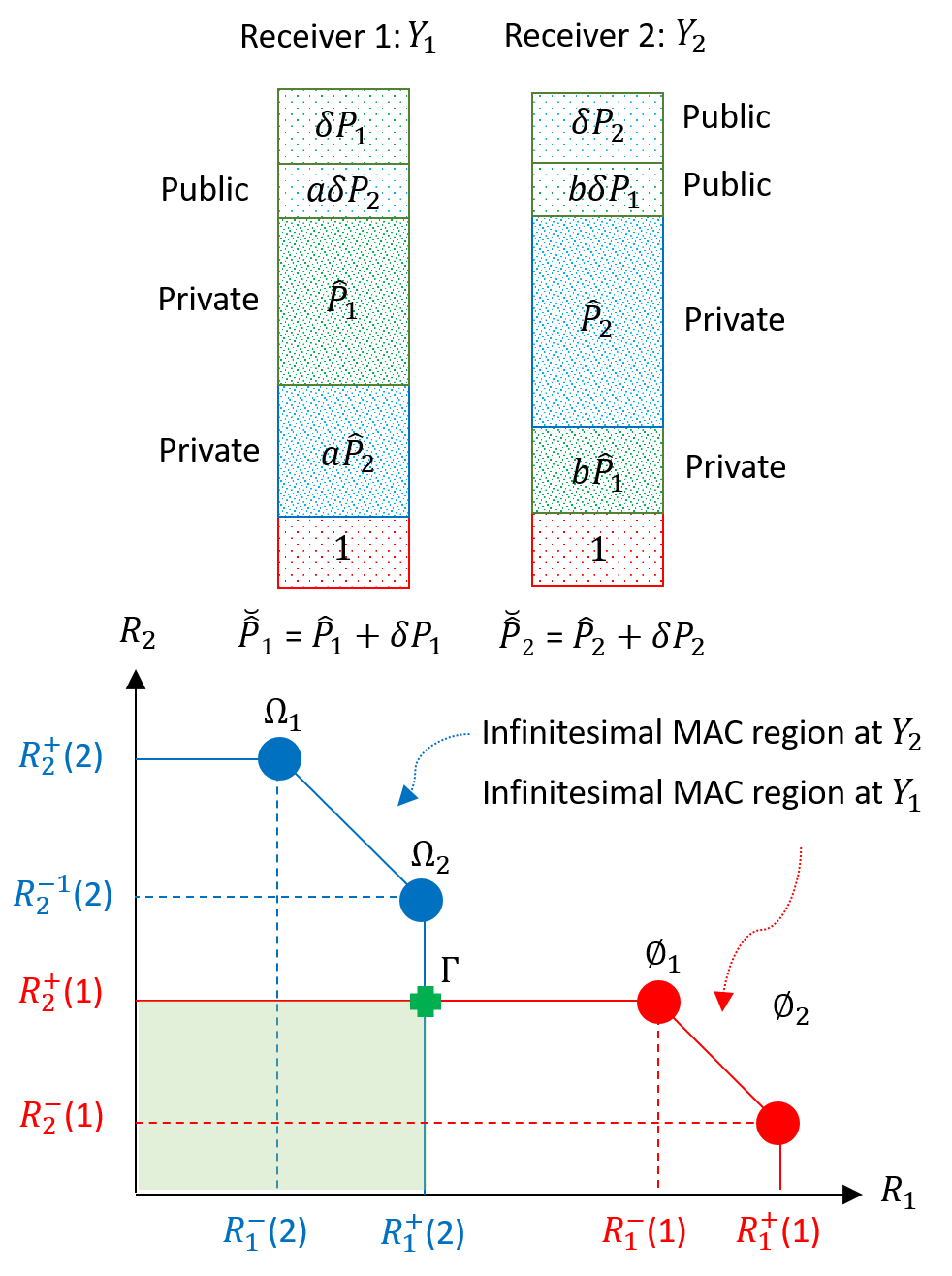}
      \caption{Intersection of infinitesimal MAC channels used in deriving the conditions for nested optimality of private messages.  $\Gamma$ is the optimum point at the intersection of the two MAC capacity regions. The partial rate due to $\delta$-layer of user 1, i.e., $\delta P_1$, can be computed without including the $\delta$-layer of user 2, i.e., $\delta P_2$, and vice versa. In other words, partial rates due to the two $\delta$-layers, i.e., $\delta P_1$ and $\delta P_2$, are added. This allows expressing a composite 
$\delta$-step in terms of two consecutive simple $\delta$-steps. }
   \label{MAC1-2e}
 \end{figure}
     
            \begin{eqnarray}  \label{EE89}
   R^+_1(2) & = & 0.5\log_2\left(1+\frac{b\delta P_1}{b\hat{P}_1+\hat{P}_2+1}\right) \simeq 0.5\left(\frac{b\delta P_1}{b\hat{P}_1+\hat{P}_2+1}\right)
   =0.5\left(\frac{\delta P_1}{\hat{P}_1+\frac{1}{b}\hat{P}_2+\frac{1}{b}}\right) \\ \label{eq101}
      R^-_1(1) & = & 0.5\log_2\left(1+\frac{\delta P_1}{\hat{P}_1+a\hat{P}_2+1+a\delta P_2}\right) \\
      & = & 0.5\log_2\left(\hat{P}_1+a\hat{P}_2+1+a\delta P_2+\delta P_1\right)-0.5\log_2\left(\hat{P}_1+a\hat{P}_2+1+a\delta P_2\right) \\ \label{eq104}
      & \simeq & 0.5 \left(\frac{a\delta P_2+\delta P_1}{\hat{P}_1+a\hat{P}_2+1}\right)-0.5\left(\frac{a\delta P_2}{\hat{P}_1+a\hat{P}_2+1} \right) =
      0.5 \left(\frac{\delta P_1}{\hat{P}_1+a\hat{P}_2+1} \right). 
       \end{eqnarray} 
       Comparing \ref{EE89} with \ref{eq104}, we obtain,
             \begin{eqnarray} 
             R_1^-(1)>R_1^+(2).
              \label{eq106}
             \end{eqnarray}           
        Referring to Fig.~\ref{MAC1-2e}, we have:
            \begin{eqnarray} \label{EE94}
   R^+_2(1) & = & 0.5\log_2\left(1+\frac{a\delta P_2}{a\hat{P}_2+\hat{P}_1+1}\right) \simeq 0.5\left(\frac{a\delta P_2}{a\hat{P}_2+\hat{P}_1+1}\right)=0.5\left(\frac{\delta P_2}{\hat{P}_2+\frac{1}{a}\hat{P}_1+\frac{1}{a}}\right) \\ \label{eq101p}
      R^-_2(2) & = & 0.5\log_2\left(1+\frac{\delta P_2}{\hat{P}_2+b\hat{P}_1+1+b\delta P_1}\right) \\
      & = & 0.5\log_2\left(\hat{P}_2+b\hat{P}_1+1+b\delta P_1+\delta P_2\right)-0.5\log_2\left(\hat{P}_2+b\hat{P}_1+1+b\delta P_1\right) \\ \label{EE97}
      & \simeq & 0.5 \left(\frac{b\delta P_1+\delta P_2}{\hat{P}_2+b\hat{P}_1+1}\right)-0.5\left(\frac{b\delta P_1}{\hat{P}_2+b\hat{P}_1+1} \right) =
      0.5 \left(\frac{\delta P_2}{\hat{P}_2+b\hat{P}_1+1} \right). 
       \end{eqnarray} 
    Comparing \ref{EE94} with \ref{EE97}, we obtain,
        
             \begin{eqnarray} 
             R_2^-(2)>R_2^+(1).
             \label{eq112}
             \end{eqnarray} 
     From \ref{eq106} and \ref{eq112}, we conclude the intersection of the two MACs (governing the rate of public $\delta$-layers) is as shown in 
Fig.~\ref{MAC1-2e}, and the rate of public layers, corresponding to Point 
     $\Gamma$ in Fig.~\ref{MAC1-2e}, is equal to:
     
                       \begin{eqnarray} 
                       (R_1^{public},R_2^{public})=(R_1^+(2), R_2^+(1)).
                        \label{eq113}
                       \end{eqnarray} 
From~\ref{eq104}, $R_1^{public}$  is determined solely by $\delta P_1$, and from~\ref{EE97},  $R_2^{public}$  is determined solely by $\delta P_2$. This means, the conditions for having a stationary optimum solution for user 1 (see~\ref{eq6}) and for user 2 (see~\ref{eq9}) are each determined by a single (simple) $\delta$-step involving a change in power allocation of user 1, or power allocation of user 2, but not both at the same step.  
In other words, the total change in the rate, which is the sum of \ref{eq104} and \ref{EE97}, is expressed as a coefficient $C_1$ times $\delta P_1$ plus a coefficient $C_2$ times $\delta P_2$. 
This allows expressing a composite 
$\delta$-step as two consecutive simple $\delta$-steps, which in essence reflects the fact the slope of the boundary in the range of interest is continuous.

\section{Optimality of Gaussian Distribution} \label{sec3}

This section proves the optimality of Gaussian distribution for achieving the boundary of each constituent region in the upper concave envelope of the capacity region for a 2-users weak GIC. First, some definitions. 

{\bf Upper Concave Enveloper and Constituent Region:} By optimally dividing the power budgets $P_1$ and $P_2$ among several instances of 2-users weak GIC, and time sharing among their associated optimum solutions (refereed as constituent regions, hereafter), one can  achieve the so-called upper concave envelope of constituent regions.  
Upper concave envelope (potentially) replaces (improves) some parts of the overall capacity region that fall around the intersection of two constituent regions by a time sharing line tangent to the two intersecting constituent regions. This means, by time sharing between two points, one on each of the intersecting constituent regions, the boundary is enlarged. This article is concerned with a single constituent region at a time.   $\square$

{\bf Strategy:} This term is used to refer to power, time/bandwidth allocation and encoding/decoding methods associated with a constituent region.  Optimization of strategies and computation of upper concave envelope is beyond the scope of this article. This article provides the tools  to optimize the constituent region once the allocated resources (time/bandwidth/power) are known. To focus on a single constituent region, simply called capacity region hereafter when there is no chance of confusion, the power budgets are fixed at $P_1$ and $P_2$ and available time/bandwidth is fully utilized by both users. Each boundary point (for given $P_1$ and $P_2$) has its associated (optimized) encoding/decoding methods. Note that for any constituent region, constraints on power budgets $P_1$ and $P_2$ will be  satisfied with equality. The reason is that if there is unused power for any of the two users, it can be simply added to the public power budget of that user, which would result in increasing $R_{ws}$.   $\square$

 {\bf State:}  This term is used to refer to part of the {\em strategy} that is under the control of a transmitter. Resource allocation (power/time/bandwidth), encoding rate and encoding method in a {\em strategy} are part of the associated {\em state}, while decoding method, which is under the control of the receiver, is the other piece of the {\em strategy}.   $\square$

In the following, for the sake of simplicity, first, in Section~\ref{sec3.2}, the proof for optimality of Gaussian is  presented relying on scalar inputs. This is then generalized to vector inputs in Section~\ref{sec3.3}. 

\subsection{Optimality of Gaussian Distribution for Scalar Inputs} \label{sec3.2}
 
Let us first consider $Y_1$. 
Let us assume the boundary is divided into a number of segments, where the slope of tangent to the boundary is continuous over each segment. 
 For one such segment, let us consider a sequence of probability distributions corresponding to $Y_1$ in subsequent $\delta$-steps along the boundary, starting from a Gaussian distribution at the corner point $A$ maximizing $R_1$. These probability distributions are denoted as $p_0(y_1)$, $p_1(y_1)$, $p_2(y_1)$, ...., $p_L(y_1)$, where $p_0(y_1)$ is the probability distribution at the corner point $A$, which is Gaussian,  and $p_1(y_1)$ is the probability distribution at the end point on the first $\delta$-step after the corner point $A$, and $p_L(y_1)$ is the probability distribution at the end point of the $L$'th $\delta$-step (which concludes the particular segment of the boundary). It is assumed that, $L \approx O(1/\delta)$ is an integer.  It is obvious that the exact value of $\delta$ can be changed from segment to segment to make sure all segments contain an integer number of $\delta$-steps.  As an alternative, the final $\delta$-step can be scaled down such that the $L$'th steps falls on the end point of the segment under consideration (see Fig.~\ref{sub-delta}).

    \begin{figure}[htp]
   \centering
   \includegraphics[width=0.35\textwidth]{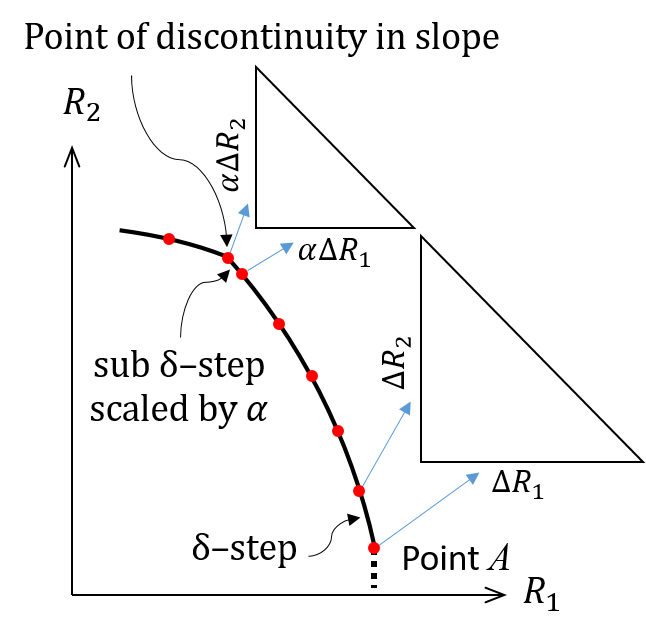}
   \caption{Moving  counterclockwise along the boundary and handling discontinuity in $\mu$ by shrinking the last $\delta$-step by $\alpha$ to end in the point of discontinuity.}
   \label{sub-delta}
 \end{figure}

For simplicity of notation, in arguments/derivations that would be applicable to both $Y_1$ and $Y_2$, the subscript $1$ (or $2$) will be dropped from $Y_1$ (or $Y_2$). Relying on the notations used in the proof of Theorem 7.4.1. on page 335 of~\cite{Gallagher}, we use $A$ to show the variance of the output under consideration ($Y_1$ or $Y_2$), and $\phi_A(y)$ to show a Gaussian distribution of variance $A$. Note that $p_0(y)=\phi_A(y)$.
 Relying on these notations, let us consider the case that the Kullback–Leibler divergence between a distribution $p_1(y)$ of variance $A$ and a Gaussian distribution of variance $A$, expressed as 
\begin{equation}
\phi_A(y)=\frac{1}{\sqrt{2\pi A}}\exp\left(-\frac{y^2}{2A}\right).
\label{eq92}
\end{equation}
satisfies, 
\begin{equation}
D[p_1(y)||\phi_A(y)]\equiv \int_{-\infty}^{+\infty} p_1(y) \log_2\left( \frac{p_1(y)}{\phi_A(y)}\right) dy =-\epsilon_1,~\mbox{where}~\epsilon_1>0,~\mbox{and}~\epsilon_1 \ll \delta \approx O(\delta^2).
\label{eq93}
\end{equation}
The reason for $\epsilon_1 \ll \delta\approx O(\delta^2)$  will become clear later\footnote{The assumption $\epsilon_1 \ll \delta$ and other assumptions concerning Kullback–Leibler divergence, e.g., conditions in \ref{eq95} and \ref{eq97} will be proved later in Section~\ref{sec313}.}.
Condition $\epsilon_i \ll \delta$ is expressed as $\epsilon_i \approx O(\delta^2)$, hereafter.
Following Theorem 7.4.1. on page 335 of~\cite{Gallagher}, one can conclude
\begin{equation}
 \int_{-\infty}^{+\infty} p_1(y) \log_2\left( \frac{1}{\phi_A(y)}\right)dy=\frac{1}{2}\log_2 (2\pi e A).
\label{eq94}
\end{equation}
 Combining expression~\ref{eq93} and~\ref{eq94}, it follows that  
\begin{equation}
D[p_1(y)||\phi_A(y)]=-\epsilon_1 ~\Longrightarrow~ h_{p_1(y)}(Y)\equiv \int_{-\infty}^{+\infty} p_1(y) \log_2 \frac{1}{p_1(y)}dy=\frac{1}{2}\log_2 (2\pi e A)-\epsilon_1
 \label{eq95}
\end{equation}
where $h_{p_1(y)}(Y)$ is the entropy of random variable $Y$ with probability density $p_1(y)$.  
From expression~\ref{eq95}, it is concluded that  $\log_2(1/p_1(y))$ is equal to:
 \begin{equation}
\log_2 \frac{1}{p_1(y)}=\left(\frac{1}{2}\log_2(2\pi e A)\times \frac{y^2}{A}\right)-\epsilon_1.
 \label{eq96}
\end{equation}
Using \ref{eq96}, we obtain
\begin{equation}
 \int_{-\infty}^{+\infty} f(y) \log_2 \frac{1}{p_1(y)}dy=\frac{1}{2}\log_2 (2\pi e A)-\epsilon_1
 \label{eq95O}
\end{equation}
where $f(y)$ is a valid probability density function.

Next let us consider a second probability distribution, $p_2(y)$ of variance $A$, satisfying 
 \begin{equation}
D[p_2(y)||p_1(y)]\equiv \int_{-\infty}^{+\infty} p_2(y) \log_2\left( \frac{p_2(y)}{p_1(y)}\right)=-\epsilon_2,~\mbox{where}~\epsilon_2>0,~\mbox{and}~\epsilon_2\ll \delta \approx O(\delta^2)
\label{eq97}
\end{equation}
where $p_1(y)$ satisfies \ref{eq96}. After some simple substations from \ref{eq96} and \ref{eq95O}, it follows that, 
 \begin{equation}
h_{p_2}(y)(Y)\equiv \int_{-\infty}^{+\infty} p_2(y) \log_2 \frac{1}{p_2(y)}dy=\frac{1}{2}\log_2 (2\pi e A)-\epsilon_1-\epsilon_2. 
\label{eq98}
\end{equation}
From~\ref{eq98}, we conclude,  $\log_2(1/p_2(y))$ is equal to:
 \begin{equation}
\log_2 \frac{1}{p_2(y)}=\left(\frac{1}{2}\log_2(2\pi e A)\times \frac{y^2}{A}\right)dy-\epsilon_1-\epsilon_2.
 \label{eq96p}
\end{equation} 
This procedure can continue from one $\delta$-step to the next, establishing that, if
 \begin{equation}
D[p_i(y)||p_{i-1}(y)]=-\epsilon_i~\Longrightarrow~
h_{p_i(y)}(Y)=\frac{1}{2}\log_2 (2\pi e A)-\sum_{\ell=1}^i \epsilon_\ell.
\end{equation} 

 Finally, noting there are $L\approx O(1/\delta)$ steps to reach to the end of the boundary segment under consideration, we obtain, 
 \begin{equation}
h_{p_{L}(y)}(Y)-h_{\phi_A(y)}(Y)=-\sum_{i=1}^{L}{\epsilon_i} \approx LO(\delta^2)\approx O(\delta)\rightarrow 0. 
\label{TotL}
\end{equation} 
On the other hand, since Gaussian distribution has the maximum entropy among all distributions with the same variance (Theorem 7.4.1. on page 335 of~\cite{Gallagher}),  it is concluded  the probability density of $Y_1$ remains Gaussian throughout the segment (recall that segment means part of the boundary with a continuous slope). As the  probability density for the end point of the boundary segment under consideration is Gaussian, and the end point of one segment is the start point of the subsequent segment (boundary is continuous), a similar proof can be applied to the next segment (and its subsequent segments) completing the boundary. 

A more general statement concerning the relationship between entropy values of random variables with an infinitely small Kullback–Leibler divergence of $\hat{\epsilon}\approx 0$ is expressed next. Notations 
$\hat{\epsilon}$ and $O(\hat{\epsilon})$ are used to refer to infinitesimal quantities, regardless of their signs and their exact values.  

\begin{theorem} \label{Th7}
Given probability distributions $p(y)$ and $q(y)$ of variance $A$, where 
 \begin{equation}
h_{p(y)}= \frac{1}{2}\log_2(2\pi e A)+O(\hat{\epsilon})~~~\Longleftrightarrow~~~\log_2(p(y)) = \left(\frac{1}{2}\log_2(2\pi e A)\times \frac{y^2}{A}\right)+O(\hat{\epsilon})
\end{equation}
then,
 \begin{equation}
\mbox{if}~D(p||q)= O(\hat{\epsilon})~\Longrightarrow~\log_2(q(y))= \left(\frac{1}{2}\log_2(2\pi e A)\times \frac{y^2}{A}\right)+O(\hat{\epsilon})
\end{equation}
and,
 \begin{equation}
\mbox{if}~D(q||p)= O(\hat{\epsilon})~\Longrightarrow~\log_2(q(y))= \left(\frac{1}{2}\log_2(2\pi e A)\times \frac{y^2}{A}\right)+O(\hat{\epsilon}).
\end{equation}
\end{theorem}
\begin{proof}
The proof follows by considering the proof of Theorem~7.4.1 on page 335 of~\cite{Gallagher}.
\end{proof}
To generalize Theorem~\ref{Th7}, we need the following definition. 

{\bf Maximum Entropy Vector}: Consider a vector ${\bf v}$ of dimension $\mathsf{D}$ with an auto-correlation matrix of trace $A\times\mathsf{D}$.  It is well known that vector ${\bf v}$ will have maximum entropy, and thereby it is called a $\mathsf{D}$-dimensional maximum entropy vector, if it is composed of i.i.d. Gaussian components each of variance $A$, i.e., ${\bf v}$ is Gaussian with an auto-correlation matrix equal to  
$A\times I$ where $I$ is a $\mathsf{D}\times\mathsf{D}$ identity matrix.    $\square$

Relying on the definition of Maximum Entropy Vector, a generalization of Theorem~\ref{Th7} follows. 
\begin{theorem} \label{Th7p}
Consider a $\mathsf{D}$-dimensional vector ${\bf y}$ with an auto-correlation matrix of a trace $A\times\mathsf{D}$. 
Consider probability distributions $p({\bf y})$ and $q({\bf y})$, where 
 \begin{equation}
h_{p({\bf y})}= \frac{\mathsf{D}}{2}\log_2(2\pi e A)+O(\hat{\epsilon})~~~\Longleftrightarrow~~~\log_2(p({\bf y})) = \left(\frac{1}{2}\log_2(2\pi e A)\times \frac{||{\bf y}||^2}{A}\right)+O(\hat{\epsilon})
\end{equation}
then,
 \begin{equation}
\mbox{if}~D(p||q)= O(\hat{\epsilon})~\Longrightarrow~\log_2(q({\bf y}))= \left(\frac{1}{2}\log_2(2\pi e A)\times \frac{||{\bf y}||^2}{A}\right)+O(\hat{\epsilon})
\end{equation}
and,
 \begin{equation}
\mbox{if}~D(q||p)= O(\hat{\epsilon})~\Longrightarrow~\log_2(q({\bf y}))= \left(\frac{1}{2}\log_2(2\pi e A)\times \frac{||{\bf y}||^2}{A}\right)+O(\hat{\epsilon}).
\end{equation}
\end{theorem}

\begin{proof}
The proof follows by considering the proof of Theorem~7.4.1 on page 335 of~\cite{Gallagher} applied to each component of the vectors under consideration, as well as the fact that a Gaussian vector with i.i.d. components is the only possible option with an entropy the same as that of a maximum entropy vector.   
\end{proof}
 
\subsubsection{Sketch of the proof concerning optimality of Gaussian inputs}

The proof is based on dividing the capacity region, i.e., the boundary of each constituent region under consideration, into segments over which (partial) derivative(s) of $R_{ws}$ in terms of $\rho$ and/or $\theta$ is/are continuous. Regardless of how many such segments exist, the total number of $\delta$ steps to cover all of them is equal to $L=O(1/\delta)$, and this is the only point used in the proof (see~\ref{TotL}). 
 Focusing on one such segment, $\delta$-steps are expressed in terms of infinitesimal changes in $\rho$ and/or in $\theta$. A $\delta$-step is called {\em simple} if it entails changing only one of the two parameters, i.e., $\rho$ or $\theta$, and it is called {\em composite} if it entails changing both $\rho$ and $\theta$.  As partial derivatives are assumed to be continuous\footnote{Recall that the assumption behind dividing the boundary into segments is that over each segment the partial derivatives of $R_{ws}$ with respect to $\rho$ and/or $\theta$ are continuous.}, it follows that a composite step can be replaced by two consecutive simple steps. The proof is based on showing that if the distribution at the starting point of a simple $\delta$-step is Gaussian, it will  be Gaussian at its end point as well. More precisely, it will be shown that the Kullback–Leibler divergence between distributions of $Y_1$ (as well as $Y_2$) at the start and end points of a simple $\delta$-step is $O(\delta^2)$. This condition entails that if the distribution of $Y_1$ (and/or $Y_2$) at the start point on a segment is Gaussian, it will be Gaussian at its end point as well.  As the number of simple steps along each segment of the boundary is $O(1/\delta)$, it follows that the Kullback–Leibler divergence between the distributions of $Y_1$ (as well as $Y_2$) at the start and end point on any boundary segment will be $O(\delta)\rightarrow 0$. 
On the other hand, since the distributions of $Y_1$ (as well as $Y_2$) at the starting point on each segment is almost Gaussian\footnote{In the sense that the difference between entropy of the distribution and the entropy of a Gaussian distribution of the same variance is infinitely small.}, the corresponding distribution at the end point will be almost Gaussian as well. 
The proof over each  $\delta$-step is based on the following points: 

\begin{enumerate}

\item Consider a sequence of consecutive simple $\delta$-steps, each involving an infinitesimal change in $\rho$ or $\theta$ (but not changing both $\rho$ and $\theta$ in a single step). 

\item Superimpose a binary channel on top of the underlying GIC with inputs corresponding to selection of a transmitter's {\em states} at the start and end points of the corresponding $\delta$-step.

\item Adjust the rate embedded in the superimposed binary channel such that the corresponding binary input is detectable at both receivers, i.e., at $Y_1$ and at $Y_2$. 

\item Deploy a successive decoding at $Y_1$ and $Y_2$ to first detect the input  of the superimposed binary channel, and then, knowing the corresponding transmitter's states,  decode the main code-books for the GIC. As a result, the rate of the binary channel will be additional to the rate corresponding to time-sharing between the start and end points on the $\delta$-step, with a time sharing coefficient which is equal to the optimum (capacity achieving) input probability (probability of binary channel input symbols) for the superimposed binary channel. 

\item The extra rate due to the binary channel should be 
$\ll O(\delta)\approx O(\delta^2)$, otherwise, it would violate the optimality of the infinitesimal segment of the boundary between the start and end points of the $\delta$-step.  Note that as $\delta \rightarrow 0$, the approximation (in terms of $\delta$) to the boundary bounded within a $\delta$-step is a line (first order, linear approximation) which coincides with the time-sharing line  between the start and the end points on the corresponding $\delta$-step. A separation (between the actual boundary and the time sharing line) in the order of $O(\delta)$ would contradict that the time-sharing line is  the first order approximation  to the actual boundary. 

\item Relying on capacity arguments presented in~\cite{Gallagher} (applied to the superimposed, binary input, continuous output channel), the separation between boundary and time-sharing line is shown to be related to two equal Kullback–Leibler divergence values: (i) the Kullback–Leibler divergence  between ``an intermediate point on the boundary, $M$'' with ``the starting point on the $\delta$-step'', and (ii) the Kullback–Leibler divergence  between $M$ with ``the end point on the $\delta$-step".  Location of point $M$ (in terms of its Kullback–Leibler divergence  to the "starting point" and to the "end point" on the $\delta$-step) is shown to be related to the capacity of the superimposed binary channel. 
Let us assume, without loss of generality, that  the simple $\delta$-step involves user 1, i.e, changing $\rho$, and consequently, $Y_1$ is considered.
For the maximum deviation of the actual boundary from the time-sharing line to be in the order of $O(\delta^2)$, the Kullback–Leibler divergence between the distributions of $Y_1$ at the start and end points on the simple $\delta$-step should be $O(\delta^2)$. By limiting the two Kullback–Leibler divergence values, which are shown to be equal, it is proved that if $Y_1$ at the starting point of the $\delta$-step is Gaussian, then it will be Gaussian at the corresponding end point as well. Note that although the argument here has focused on $Y_1$, a similar argument is applicable to simple $\delta$-steps involving $\theta$, concluding similar results for the distribution of $Y_2$.

\item Noting that $Y_1=X_1+aX_2+Z_1$ and $Y_2=X_2+bX_1+Z_2$, where $Z_1~N(0,1)$ and $Z_2~N(0,1)$, it is concluded that for $Y_1$ and $Y_2$ to be Gaussian,  the distributions of $X_1$ and $X_2$ has to be Gaussian as well. Note that the powers of $Y_1$ and $Y_2$ at all points along the boundary are equal to $P_1+\sigma^2$ and $P_2+\sigma^2$ , where $\sigma^2$, power of the additive Gaussian  noise,  is normalized to one. 

\item Noting $X_1$ and $X_2$ are Gaussian, it follows that the code-books for the private and public messages should be Gaussian as well, and the code-books for $X_1$ (as well as $X_2$) are obtained by superimposing (adding) Gaussian sub-layers for private and public parts. These points will be formally proved in Appendix~\ref{APP2}.
\end{enumerate}

\subsubsection{A closer look at $\delta$-steps along the boundary}

 To proceed with the formal proof, let us consider Figs. \ref{segments-proof} and~\ref{segments-proof2}, showing  a  simple step from state-tuple $(\mathcal{S}_1, \mathcal{S}_{2})$ to state-tuple $(\mathcal{S'}_1, \mathcal{S}_{2})$, followed by a second  simple step from state-tuple $(\mathcal{S'}_1, \mathcal{S}_{2})$ to state-tuple $(\mathcal{S'}_1, \mathcal{S'}_{2})$. Together, the two simple steps form a composite step from $(\mathcal{S}_1, \mathcal{S}_{2})$ to $(\mathcal{S'}_1, \mathcal{S'}_{2})$. 
An (stationary) optimum  solution with respect to $\theta$ corresponds to,
\begin{eqnarray}
\frac{\partial R_{ws}}{\partial \theta}=\frac{\partial R_1}{\partial \theta}+\mu 
\frac{\partial R_2}{\partial \theta}=0 \Longrightarrow \mu=- \frac{\displaystyle\frac{\partial R_1}{\partial \theta}}{\displaystyle\frac{\partial R_2}{\partial \theta}}\simeq  -\frac{\Delta^{\theta} R_1}{\Delta^{\theta} R_2}
\label{partial1}
\end{eqnarray}
where $\Delta^{\theta} R_1$ and $\Delta^{\theta} R_2$ show the changes in $R_1$ and $R_2$ due to the change in $\theta$.  
Likewise, a (stationary) optimum solution with respect to $\rho$ corresponds to,
\begin{eqnarray}
\frac{\partial R_{ws}}{\partial \rho}=\frac{\partial R_1}{\partial \rho}+\mu 
\frac{\partial R_2}{\partial \rho}=0 \Longrightarrow \mu=- \frac{\displaystyle\frac{\partial R_1}{\partial \rho}}{\displaystyle\frac{\partial R_2}{\partial \rho}}\simeq  -\frac{\Delta^\rho R_1}{\Delta^\rho R_2}
\label{partial2}
\end{eqnarray}
where $\Delta^{\rho} R_1$ and $\Delta^{\rho} R_2$ show the changes in $R_1$ and $R_2$ due to the change in $\rho$.  Recalling that $\theta P_1$ and $\rho P_2$ are the fractions of power allocated to the private messages of user 1 and user 2, respectively, it follows that,  $R_1$ is a monotonically increasing function of $\theta$ and a monotonically decreasing function of $\rho$. Likewise, $R_2$ is a monotonically decreasing function of $\theta$ and a monotonically increasing function of $\rho$.

  \begin{figure}[htp]
   \centering
   \includegraphics[width=0.75\textwidth]{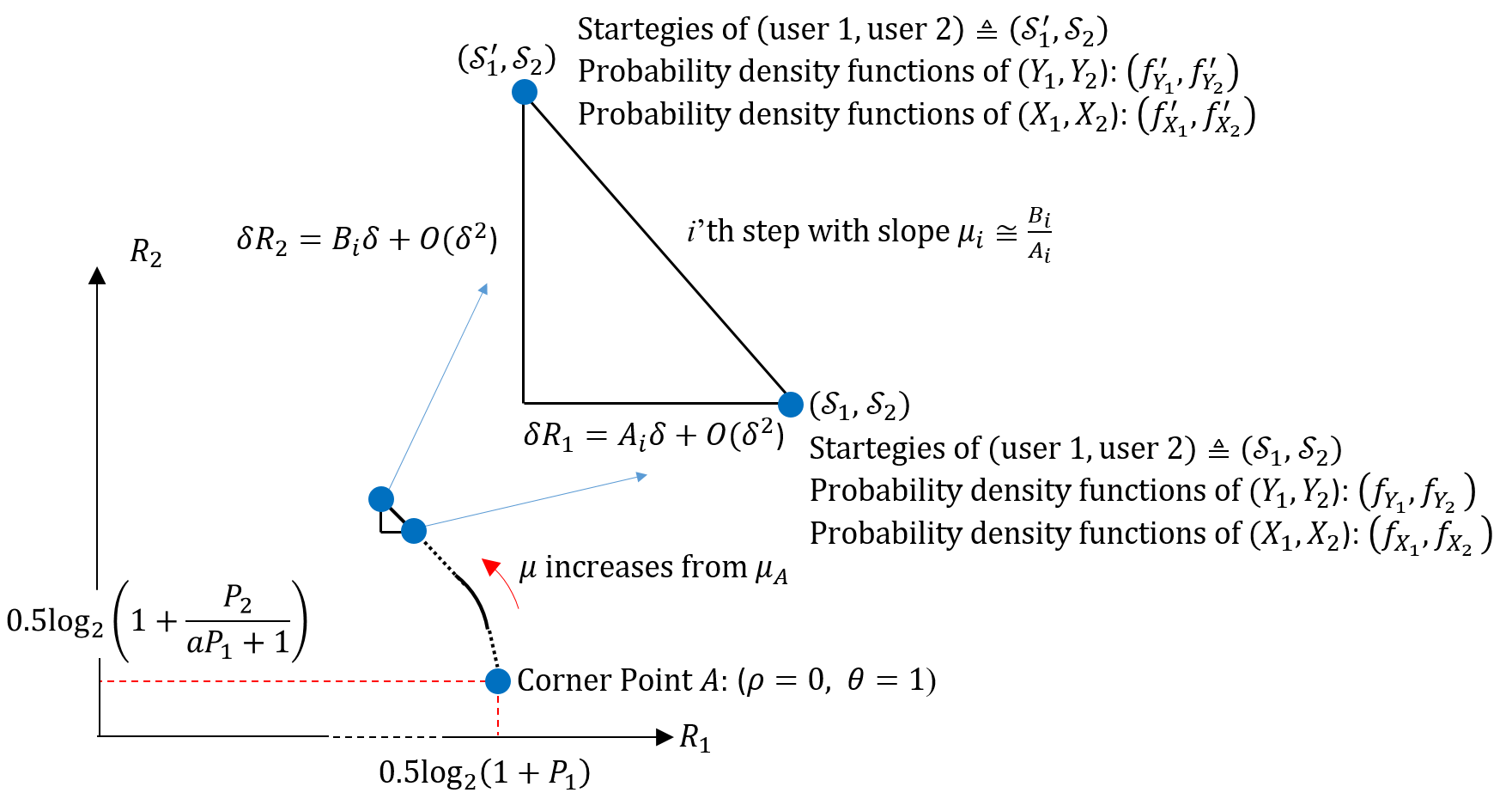}
   \caption{Strategies corresponding to the start and end points on a  simple $\delta$-step. $A_i$ and $B_i$ denote the values of partial derivatives at the point under study on the boundary.}
   \label{segments-proof}
 \end{figure}

   \begin{figure}[htp]
   \centering
   \includegraphics[width=0.6\textwidth]{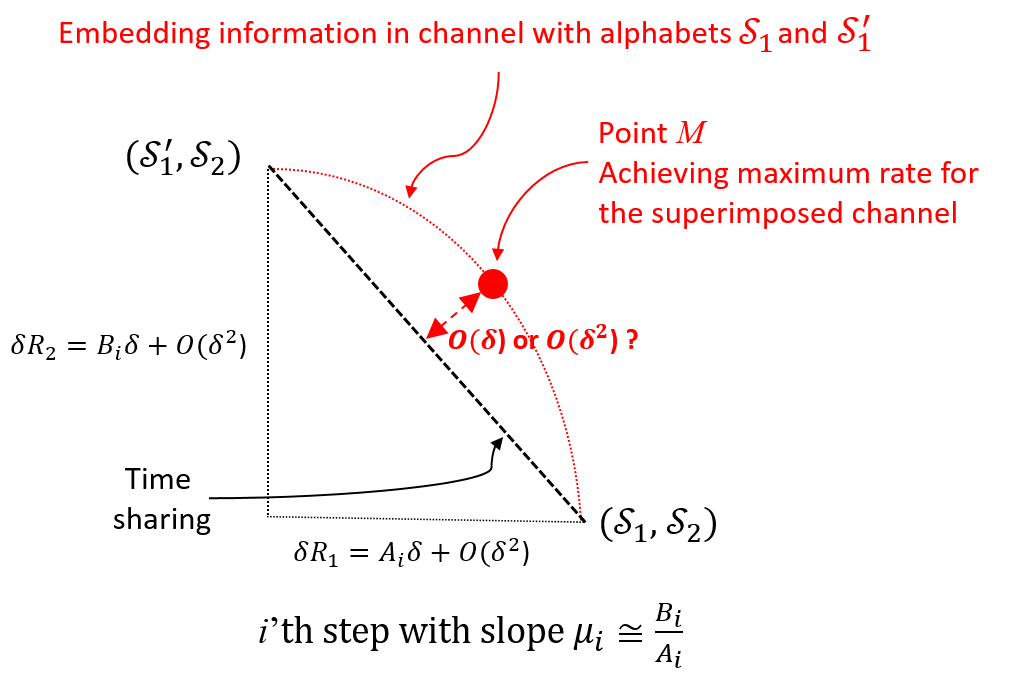}
   \caption{Straight line (time sharing line) connecting the start and end points on an infinitesimal segment (a simple $\delta$-step). $A_i$ and $B_i$ denote the values of partial derivatives at the point under study on the boundary.}
   \label{segments-proof2}
 \end{figure}

   \begin{figure}[htbp]
   \centering
   \includegraphics[width=0.37\textwidth]{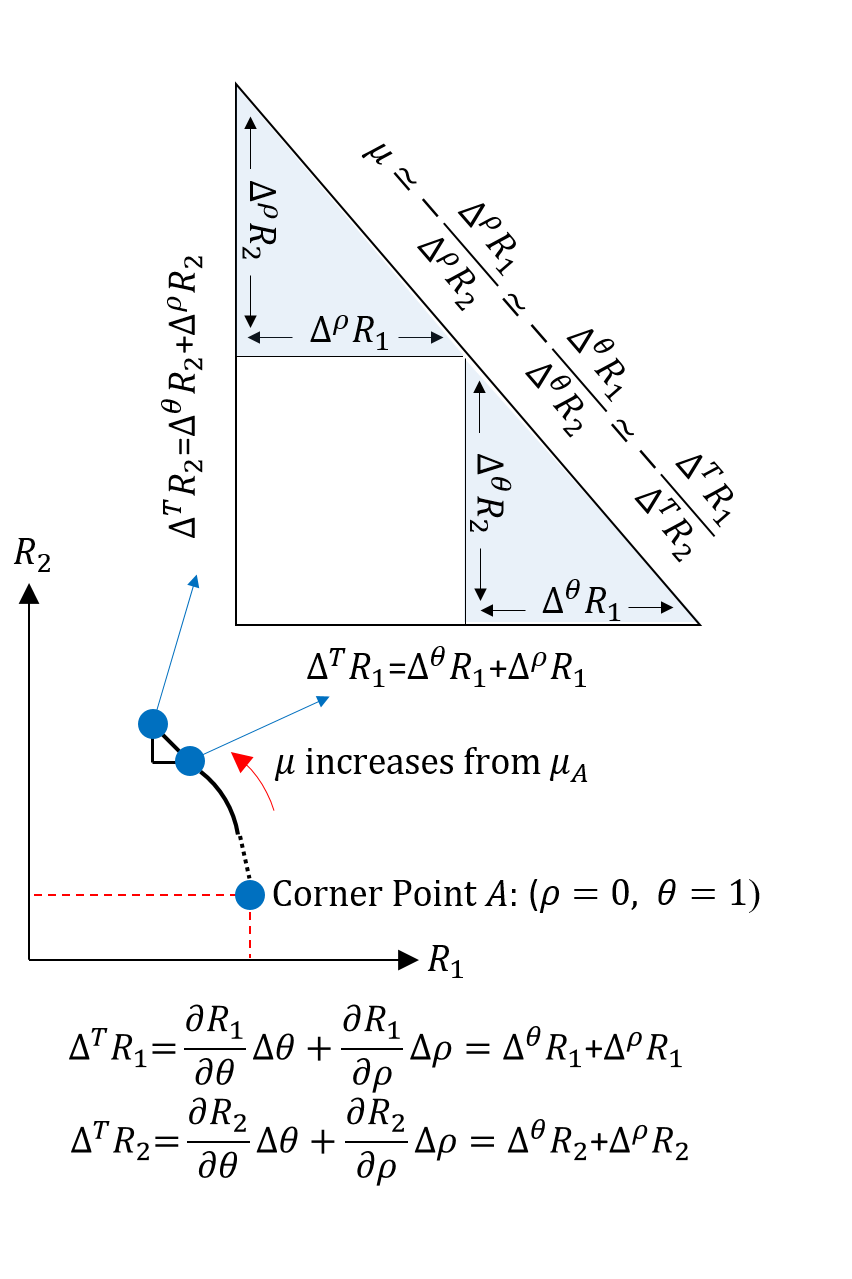}
   \caption{Moving  counterclockwise along a segment of the boundary where partial derivatives, and consequently $\mu$, are continuous. Two consecutive simple steps are constructing  a composite step.}
   \label{delta-delta}
 \end{figure}

   \begin{figure}[htp]
   \centering
   \includegraphics[width=0.6\textwidth]{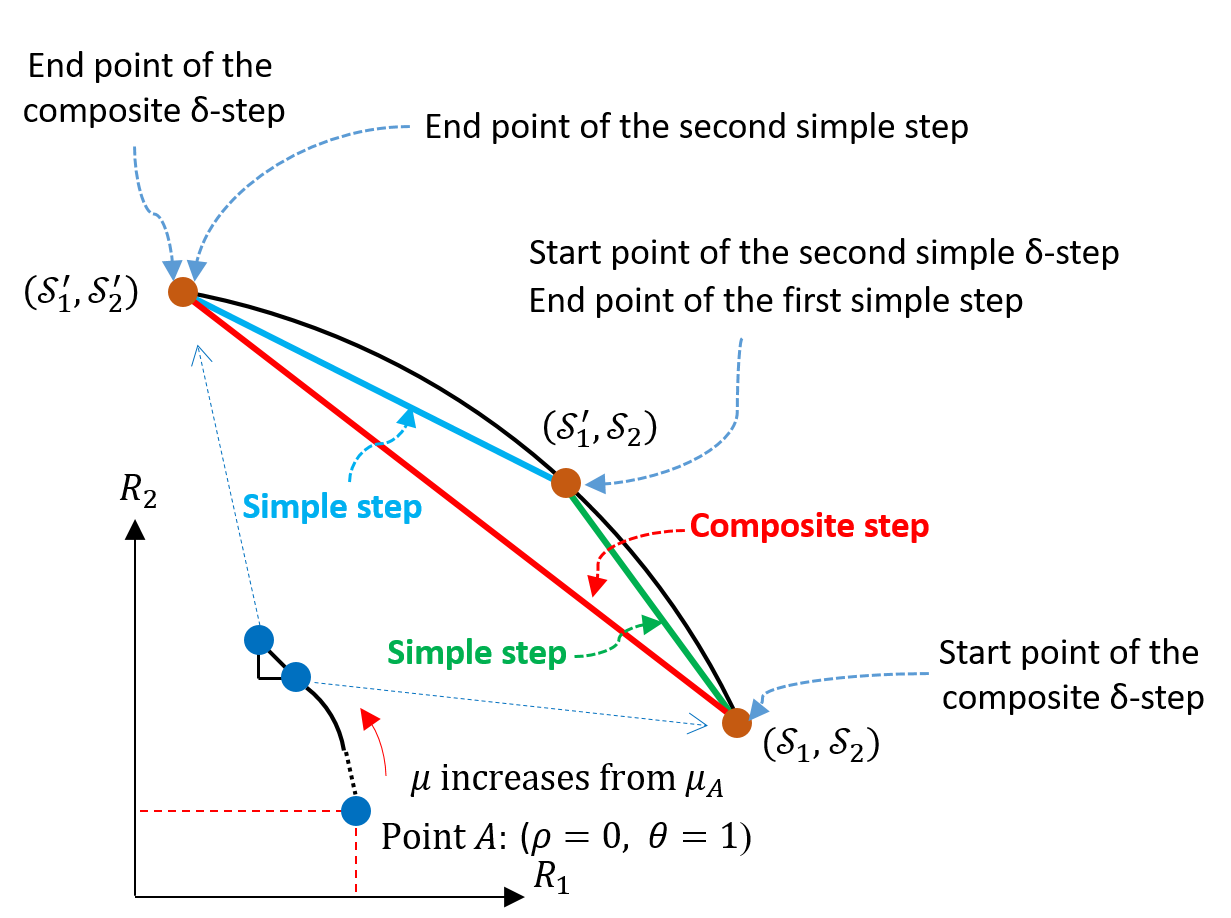}
   \caption{Consecutive simple steps constructing  a composite step.}
   \label{simple-vs-composite-step}
 \end{figure}

Using \ref{partial1} and \ref{partial2},  we obtain (see Fig.\ref{delta-delta}),
\begin{eqnarray}
\mu=\frac{\Delta^{\theta} R_1}{\Delta^{\theta} R_2}
\overset{~(a)~}{=}
\frac{\Delta^\rho R_1}{\Delta^\rho R_2}\overset{~(b)~}{=}
\frac{\Delta^\theta R_1+\Delta^\rho R_1}{\Delta^\theta R_2+\Delta^\rho R_2}
\label{partial3}
\end{eqnarray}
where \ref{partial3}\,(a) is obtained by combining \ref{partial1} and  \ref{partial2}, and \ref{partial3}\,(b) is obtained by simple manipulation of the two sides of \ref{partial3}\,(a). In summary, a composite step moving along $\theta$ and $\rho$  can be decomposed into two simple steps, a  first simple step  along $\theta$, followed by a second simple step along $\rho$, or vice versa (see Fig.~\ref{delta-delta}).  Figure~\ref{simple-vs-composite-step} demonstrates the difference between a composite step with its associated approximating line, as compared to having two simple steps each with its own approximating line. Moving a $\delta$-step along the boundary has the following characteristics:
\begin{enumerate}
\item The start and end point of each $\delta$-step are on the boundary.
\item The line connecting the start and end points of a $\delta$-step is the best linear approximation to the actual segment of the boundary in the sense that the error between approximating line and the actual boundary is $\ll\delta$, denoted as $O(\delta^2)$. In other words, the approximating line (time sharing between the start and end points on the $\delta$-step) results in a rate that is within  $O(\delta^2)$ of the actual boundary points within the corresponding $\delta$-step. 
\end{enumerate}

\subsubsection{Details of the proof concerning optimality of Gaussian inputs} \label{sec313}

{\bf Proof:} 
Let us consider a composite $\delta$-step composed of two simple $\delta$-steps as shown in Fig.~\ref{simple-vs-composite-step}. The first simple $\delta$-step, involving $\rho$ and $\mathcal{S}_1$, starts at point $(\mathcal{S}_1, \mathcal{S}_{2})$ and ends at point $(\mathcal{S'}_1, \mathcal{S}_{2})$. 
Fig.~\ref{segments-proof-binary} shows a general superimposed binary channel with inputs $(\mathcal{S}_1, \mathcal{S}_{2})$  and $(\mathcal{S}^{'}_1, \mathcal{S}^{'}_{2})$.  Let us consider a simple $\delta$-step involving $\rho$, i.e., user 1, only. The corresponding superimposed binary channel is shown in Fig.~\ref{segments-proof-binary}\,(2), where user 1 can transmit additional information by selecting one of the two symbols $(\mathcal{S}_1, \mathcal{S}^{'}_1)$ while user 2 selects $\mathcal{S}_{2}$ in all transmissions. Figure~\ref{segments-proof-binary}\,(3) depicts a simplified form for the channel in~\ref{segments-proof-binary}\,(2). 

      \begin{figure}[htp]
   \centering
   \includegraphics[width=0.6\textwidth]{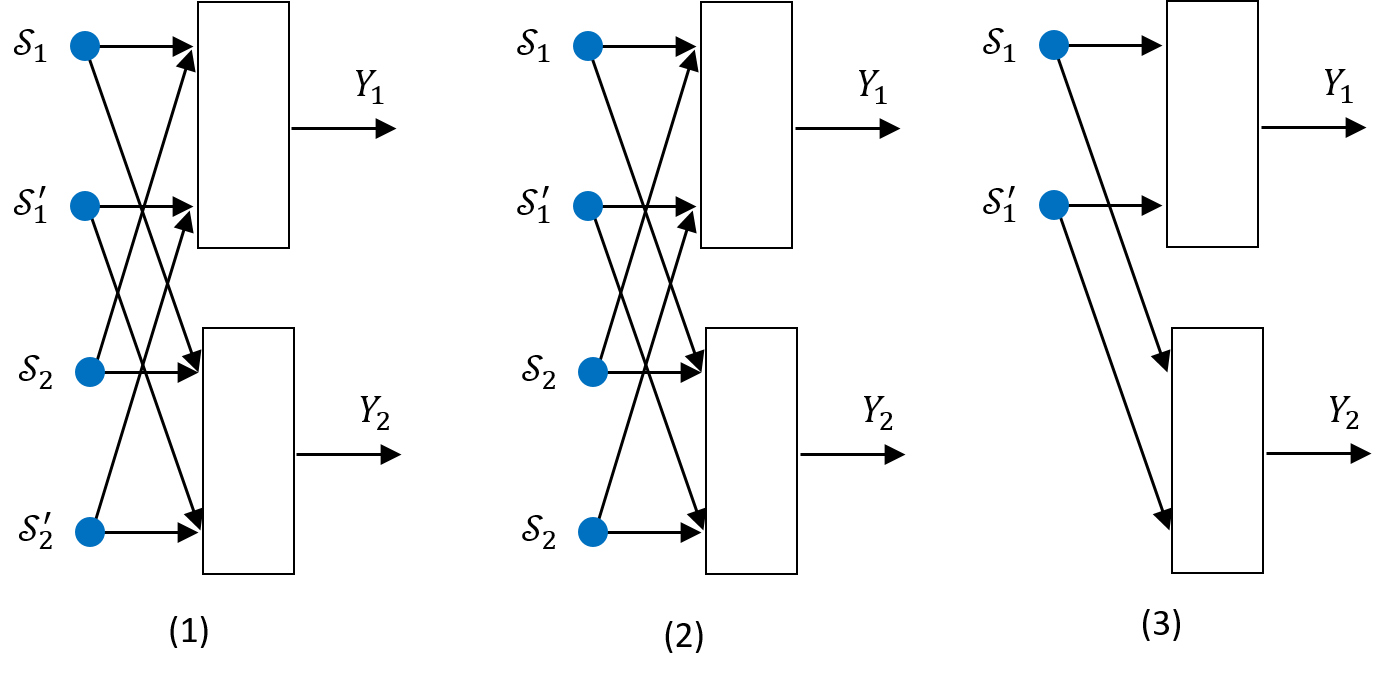}
   \caption{Superimposed binary input continuous output channel  used to embed information in the selection of strategies corresponding to the start and end points on an infinitesimal segment, i.e., $\delta$-step: Fig.~\ref{segments-proof-binary}\,(1) represents a composite $\delta$-step, Fig.~\ref{segments-proof-binary}\,(2) represents a simple $\delta$-step for which only ${\cal S}_1$ has changed, and finally, Fig.~\ref{segments-proof-binary}\,(3) is a simplified representation of Fig.~\ref{segments-proof-binary}\,(2) where ${\cal S}_2$ is omitted.}  
   \label{segments-proof-binary}
 \end{figure}

Let us assume user 1 embeds some information in the binary channel with input symbols $\mathcal{S}_1$ and $\mathcal{S}_{1}^{'}$. To do so without compromising the rate of the underlying 2-users GIC, it is required that the binary symbol  selected by user 1 is detectable at both receivers, so the embedded rate should  be the smaller of the rates detectable at $Y_1$ and $Y_2$. To continue, let us assume the two inputs for the binary channel in Fig.~\ref{segments-proof-binary}\,(3) are used with probabilities $\psi$ and $1-\psi$, where $\psi$ is optimized to maximize the smaller of the two rates, $C\left[(\mathcal{S}_1, \mathcal{S}_{1}^{'})  \rightarrow Y_1 \right]$ and 
$C\left[(\mathcal{S}_1, \mathcal{S}_{1}^{'})  \rightarrow Y_2 \right]$, where
\begin{equation} 
\mbox{Rate~of}~(\mathcal{S}_1, \mathcal{S}_{1}^{'}) \rightarrow Y_1 \triangleq C\left[(\mathcal{S}_1, \mathcal{S}_{1}^{'})  \rightarrow Y_1 \right]
\label{Eq-proof2}
\end{equation}
\begin{equation} 
\mbox{Rate~of}~(\mathcal{S}_1, \mathcal{S}_{1}^{'}) \rightarrow Y_2 \triangleq C\left[(\mathcal{S}_1, \mathcal{S}_{1}^{'})  \rightarrow Y_2 \right].
\label{Eq-proof3}
\end{equation} 
In this case, the receivers will first detect the symbol corresponding to the superimposed binary channel. As $\mathcal{S}_{2}$ has remained the same in all transmissions, and $\mathcal{S'}_{1}$ is detected error-free at both receivers, it is concluded that receivers can detect portions of the transmission for which the coding scheme associated with the starting point of the associated $\delta$-step is in effect, and likewise,  can detect portions of the transmission for which the coding scheme associated with the end point of the associated $\delta$-step is in effect. This allows to perform the encoding/decoding that are optimized for the starting point, as well as for the end point, of the $\delta$-step under consideration. 
With some misuse of notation,  $R_{ws}(\mathcal{S}_1, \mathcal{S}_{2})$ is used to emphasize the dependency of $R_{ws}$ on states $\mathcal{S}_1$ and  $\mathcal{S}_{2}$, respectively. 
The rate after decoding the alphabet of the superimposed binary channel will be: 
\begin{equation}
\psi R_{ws}(\mathcal{S}_1, \mathcal{S}_{2})+(1-\psi) R_{ws}(\mathcal{S}^{'}_1, \mathcal{S}_2),~\psi\in[0,1].
\label{Eq-proof3p}
\end{equation}
The rate of GIC in~\ref{Eq-proof3p} is equivalent to rate achieved by time-sharing between the start and end points on the first  simple $\delta$-step with time-sharing coefficient $\psi$, i.e., the probability maximizing the smaller of the two rates in the superimposed binary channel. The corresponding line coincides with the first order approximation to the boundary curve mentioned earlier.  
This rate in~\ref{Eq-proof3p} is summed up with the additional data embedded in the superimposed binary channel. As the solution without the additional rate due to the superimposed binary channel is assumed  to be optimum, the gap between the total rate and the time sharing rate in~\ref{Eq-proof3p} should be $\ll \delta$,  denoted as $O(\delta^2)$. This means the additional rate due to the superimposed binary channel should be $O(\delta^2)$ or less.     

The point with maximized superimposed rate is shown as point $M$ in Fig.~\ref{segments-proof2}. Without loss of generality, let us assume binary channel $(\mathcal{S}_1, \mathcal{S}_{1}^{'}) \rightarrow Y_1$ has a lower rate. This means user 1 can embed $C\left[(\mathcal{S}_1, \mathcal{S'}_1)\rightarrow Y_1\right]$ bits in $X_1$ which would increase $R_1$
 without decreasing   $R_2$. 
Using Theorem 4.5.1 on page 91 of \cite{Gallagher}, the capacity of the superimposed binary channel is equal to the maximum of the mutual information between each input, averaged over the output, which can be directly expressed in terms of the  Kullback–Leibler divergence. 
To be consistent with the notations used in~\cite{Gallagher}, let us use the notation $s$ to refer to the  input symbols for the binary channel in Fig.~\ref{segments-proof-binary}.  The optimality conditions in expression 4.5.3 on page 91 of \cite{Gallagher}  (rewritten for a channel with binary input and continuous output) are as follows: 
\begin{equation}
C\left[(\mathcal{S}_1, \mathcal{S}_{1}^{'})  \rightarrow Y_1 \right]=\max_\psi I(s=\mathcal{S}_1;Y_1)=\max_\psi I(s=\mathcal{S}_{1}^{'};Y_1)
\end{equation}
or,
\begin{equation}
C\left[(\mathcal{S}_1, \mathcal{S}_{1}^{'})  \rightarrow Y_1 \right]=\max_\psi \int_{Y_1}p_{Y_1|\mathcal{S}_1}(y_1)\log_2\frac{p_{Y_1|\mathcal{S}_1}(y_1)}{p_{Y_1}(y_1)} d y_1=
\max_\psi \int_{Y_1}p_{Y_1|\mathcal{S}_{1}^{'}}(y_1)\log_2\frac{p_{Y_1|\mathcal{S}_{1}^{'}}(y_1)}{p_{Y_1}(y)}d y_1
\label{KL1}
\end{equation}
where maximization is over the probabilities of binary input symbols $\psi=P(\mathcal{S}_1)$ and $1-\psi=P(\mathcal{S}_{1}^{'})$.   
From~\ref{KL1}, and definition of Kullback–Leibler divergence, it is concluded,
\begin{equation}
C\left[(\mathcal{S}_1, \mathcal{S}_{1}^{'})  \rightarrow Y_1 \right]=
\max_\psi D_{K\!L}(p_{Y_1|\mathcal{S}_1}(y_1)||p_{Y_1}(y_1)) = 
\max_\psi D_{K\!L}(p_{Y_1|\mathcal{S}_{1}^{'}}(y_1)||p_{Y_1}(y_1))
\label{KL2}
\end{equation}
where $p_{Y_1|\mathcal{S}_1}(y_1)$ and $p_{Y_1|\mathcal{S}_{1}^{'}}(y_1)$ are the probability distribution of $Y_1$ at the start and end points on the infinitesimal $\delta$-step in Fig.~\ref{segments-proof2}, and 
\begin{equation}
 p_{Y_1}(y_1)=\psi p_{Y_1|\mathcal{S}_1}(y_1) +(1-\psi)p_{Y_1|\mathcal{S}_{1}^{'}}(y_1)
\label{KL3}
\end{equation}
is the probability density of $Y_1$ at point $M$ in Fig.~\ref{segments-proof2}. 
Since the infinitesimal segment is part of the boundary, the maximum separation between the time sharing line (straight line in Fig.~\ref{segments-proof2}) and the furthest point on the actual boundary curvature should be $O(\delta^2)$. This means, in Fig.~\ref{segments-proof2}, the separation between point $M$ and time sharing line is  $O(\delta^2)$. 

 As mentioned earlier,  point $M$ in Fig.~\ref{segments-proof2} maximizes the smaller of the two rates given in~\ref{Eq-proof2} and \ref{Eq-proof3} (see Fig.~\ref{segments-proof-binary}). Using notation $p_M(y_1)$ to refer to the distribution of $Y_1$ at point $M$ in Fig.~\ref{segments-proof2}, and noting expression~\ref{KL2}, we have
\begin{equation}
D_{K\!L}[p_M(y_1)||p_{Y_1|\mathcal{S}_1}(y_1)]=D_{K\!L}[p_M(y_1)||p_{Y_1|\mathcal{S'}_1}(y_1)]\ll \delta \approx O(\delta^2).
\label{KL-main}
\end{equation}
Noting that power will be entirely used, let us use the notation $A$ to refer to the variance of distribution of $Y_1$ at points along the infinitesimal $\delta$-step in Fig.~\ref{segments-proof2}. As mentioned earlier, notation $\hat{\epsilon}$ is used to refer to any infinitesimal value regardless of its sign. Since $p_{Y_1|\mathcal{S}_1}(y_1)$ is Gaussian by assumption ($Y_1$ at the starting point is  Gaussian), from Theorem~\ref{Th7}, it follows that $p_M(y_1)$ is Gaussian as well.  Noting from \ref{KL-main}, we have 
$D_{K\!L}[p_M(y_1)||p_{Y_1|\mathcal{S'}_1}(y_1)]\approx \hat{\epsilon}$. Finally, noting Theorem~\ref{Th7} and also noting that $p_M(y_1)$ is Gaussian, it is concluded that 
$p_{Y_1|\mathcal{S'}_1}(y_1)$, i.e., the distribution at the end point of the $\delta$-step is Gaussian as well, i.e., 
\begin{equation}
\log_2(p_{Y_1|\mathcal{S}_{1}^{'}}(y_1))\approx \left(\frac{1}{2}\log_2(2\pi e A)\times \frac{y_1^2}{A}\right)+\hat{\epsilon}.
\label{eq116p}
\end{equation}
Expression~\ref{eq116p} entails that the distribution at the starting point of the next  $\delta$-step is Gaussian. Relying on a similar reasoning as in the previous $\delta$-step, one can continue until reaching the end point on the segment under consideration. As the end point on one segment is the starting point on the next segment, the argument can continue until all segments are covered. On the other hand, $Y_1=X_1+\sqrt{a}X_2+Z_1$ and $Y_1=X_2+\sqrt{b}X_1+Z_2$, where $Z_1$ and $Z_2$ are Gaussian. Cramér’s decomposition theorem states that if the sum of several independent random variables in Gaussian, then each of those random variable has to be Gaussian as well. Here, as $X_1$, $X_2$ and $Z_1$ are independent and $Z_1$ is Gaussian, a Gaussian $Y_1$ entails $X_1$ and $X_2$ should be Gaussian. A similar conclusion is derived about $X_1$ and $X_2$ being Gaussian, noting that $Y_2$ and $Z_2$ are Gaussian.  $\square$

\subsubsection{Layering of Public/Private Messages and Decoding Order} The upper bounding region (corresponding to the intersection of MACs) discussed  in Appendix~\ref{APP2} is optimized  by layering based on $X_1=U_1+V_1$, $X_2=U_2+V_2$ (needed in order to achieve the capacity region of the underlying MACs whose intersection forms the upper bounding region). 
Section \ref{sec2} for gradual construction of the lower part of the boundary  verifies that the procedure is indeed realizing the intersection of MAC1 and MAC2. Optimum power allocation is then deployed to enlarge the overlapping regions in order to maximize 
$R_{ws}=R_1+\mu R_2$. Realizing the bounding region requires  that the private and public code-books are Gaussian, independent of each other and added (to form each transmitted signal), resulting in superposition coding. It also entails that public messages should be decoded prior to private messages. $\square$

\subsubsection{Optimality of Gaussian Distribution for Vector Inputs} \label{sec3.3}

A question may arise if capacity achieving inputs in 2-users weak GIC could be vectors with non-i.i.d. components, i.e., ${\bf x}_1=\{X_{11},X_{12},...,X_{1\mathsf{D}}\}$ and ${\bf x}_2=\{X_{21},X_{22},...,X_{2\mathsf{D}}\}$, resulting in outputs 
${\bf y}_1={\bf x}_1+{\bf z}_1$ and ${\bf y}_2={\bf x}_2+{\bf z}_2$, where $\mathsf{D}$ is the dimension of input/output vectors,. In this case, each component of ${\bf x}_1$, and
 likewise each component of ${\bf x}_2$, rely on certain allocation of total power budget. Let use the notation ${\bf p}_1=\{P_{11},P_{12},...,P_{1\mathsf{D}}\}$ 
and ${\bf p}_2=\{P_{21},P_{22},...,P_{2\mathsf{D}}\}$ to specify the power allocated to different components of input vectors, ${\bf x}_1$ and ${\bf x}_2$. In this case, point $A$, which maximizes $R_1$, is composed of private messages of powers 
$\{P_{11},P_{12},...,P_{1\mathsf{D}}\}$ for user 1, and public messages of powers 
$\{P_{21},P_{22},...,P_{2\mathsf{D}}\}$ for user 2. It is well know that the capacity of a point-to-point AWGN channel, if the input is viewed as a vector, is achieved using a vector with i.i.d. Gaussian components. As a result, to maximize $R_1$ at point $A$, code-books across different coordinates will be independent of each other, and will have an equal allocation of power, i.e., 
\begin{equation}
P_{11}=P_{12}...=P_{1\mathsf{D}}=\frac{\sum_{d=1}^\mathsf{D} P_{1d}}{\mathsf{D}}.
\label{eq128}
\end{equation}
As a result, the rate $R_2$ is maximized for 
\begin{equation}
P_{21}=P_{22}...=P_{2\mathsf{D}}=\frac{\sum_{d=1}^\mathsf{D} P_{2d}}{\mathsf{D}}.
\label{eq129}
\end{equation}
Expression \ref{eq129} is concluded because user 2 sees a set of parallel Gaussian channels with equal noise powers, i.e.,  
\begin{equation}
\sigma^2+\frac{\sum_{d=1}^\mathsf{D} P_{1d}}{\mathsf{D}},
\label{eq129p}
\end{equation}
where $\sigma^2=1$ is the power of original AWGN.  
The final conclusion is that both ${\bf x}_1$ and ${\bf x}_2$ will be maximum entropy vectors (Gaussian with i.i.d. components). 

For the case of using vector inputs, similar to the case of using scalar inputs, boundary is divided into segments with continuous slope. Moving along the boundary is achieved by taking simple steps. The difference with the case of scalar inputs is that, for vector inputs, each strategy involves power allocation and encoding/decoding of vectors ${\bf x}_1$ and ${\bf x}_2$. Now consider two strategies that are in effect at the starting point and at the end point on a simple step. Similar to the case of scalar inputs, one can embed information in the selection of the two strategies in a manner that the selected strategy is decodable at both receivers. Following a line of reasoning similar to the case with scalar inputs, it follows that the  capacity of superimposed binary channel  is $O(\delta^2)$. Relating the capacities of superimposed binary channels to Kullback–Leibler divergence, it follows that the distribution of ${\bf y}_1$ at the starting point on the simple step will be the same as the distribution of ${\bf y}_1$ at its end point.  A similar conclusion is valid for the  distribution of ${\bf y}_2$, which will be the same  at the starting point and at the end point on the simple step. On the other hand, the distributions of ${\bf y}_1$ and ${\bf y}_2$ at the starting point $A$ (starting point on the simple step) are maximum entropy vectors (Gaussian with i.i.d. components). Noting Theorem~\ref{Th7p},  it is concluded that ${\bf y}_1$ and ${\bf y}_1$ at the end point of the first simple $\delta$-step (simple $\delta$-step starting from point $A$) will be maximum entropy vectors as well (Gaussian with i.i.d. components).  Continuing from the end point of the first $\delta$-step, it follows that  ${\bf y}_1$ and ${\bf y}_2$ will be  maximum entropy vectors throughout the segment of the boundary under consideration. On the other hand, as the AWGN terms are Gaussian vectors of maximum entropy, for ${\bf y}_1$ and ${\bf y}_2$ to be maximum entropy, ${\bf x}_1$ and ${\bf x}_2$ should be maximum entropy as well. 
This means, the solution for given values of power allocated to ${\bf x}_1$ and ${\bf x}_2$ will be i.i.d. repetition of a single letter Gaussian $X_1$ and a single letter Gaussian $X_2$. In other words, the random code-book obtained by forming the Gaussian vectors with i.i.d. components, and i.i.d. repetition of such vectors in time, could be also obtained by  generating a simple code-book with i.i.d. components in time.  

The above lines of arguments will be valid for any strategy, i.e., allocation of power budgets, time and bandwidth to the two users. This means, the optimum solution for a given strategy is composed of random code-books with i.i.d. Gaussian components.   By optimally dividing the power budgets, as well as time and bandwidth, among such i.i.d. single letter solutions, and time sharing among them (to compute the upper concave envelope), one may obtain composite code-books composed of Gaussian vectors with independent, but not necessarily identically distributed components.  Derivation of all possible strategies, and computing the resulting upper concave envelope, is beyond the scope of this article. 

In situations that the upper concave envelope dictates using more than one strategy, components of  $X_1$ (as well as $X_2$) will be independent Gaussian, but with different variances. In other words, $X_1$ and $X_2$ are Gaussian vectors with diagonal correlation matrices, but with unequal diagonal elements. Indeed, in such a  case, diagonal elements can be divided into groups, where the elements within each group are equal, but are not equal to the elements within a different group. The number of groups will be equal to the number of strategies contributing to the formation of the upper concave envelope. 
It is well known that in such a case, changing the basis for $X_1$ and $X_2$ transforms $X_1$ and $X_2$ into Gaussian vectors with non-diagonal correlation matrices, meaning the components become dependent on each other. 
 
\section{Relation to Han Kobayashi (HK) Rate  Region} \label{sec4}

 It is shown that the HK achievable rate region, upon removal of some redundant  constraints, include the capacity region of the 2-users weak GIC. Note that optimizing the HK region includes all possible decoding orders, including those pursued in capacity achieving procedures. Relying on this point, selection of constraints to be removed is guided by encoding/decoding  procedures deployed in the upper bounding region.  

The solution obtained in this manner expresses $R_{ws}$ in terms of the linear combination of some mutual information terms.  Then, power can be optimally allocated to maximize the resulting $R_{ws}$.

Expanded Han-Kobayashi constraints\footnote{See expressions 3.2 to 3.15 on page 51 of~\cite{CIC4}, with the changes ($\mbox{current~article} \leftrightarrow [2]$):
$U_1 \leftrightarrow W_1$, 
$U_2 \leftrightarrow W_2$,
$V_1 \leftrightarrow U_1$, 
$V_2 \leftrightarrow U_2$,
$R_{U_1}\leftrightarrow T_1$, $R_{U_2}\leftrightarrow T_2$, 
$R_{V_1}\leftrightarrow S_1$, $R_{V_2}\leftrightarrow S_2$.
Note that, time sharing parameter $Q$ in ~\cite{CIC4} represents a {\em strategy} in this article.} and the associated optimization problem  can be expressed as~\cite{CIC4},

\begin{align} 
\label{HK1p}
\mbox{Maximize:}~~~~~ & R_1+\mu R_2  \\
\mbox{Subject to:}~~~ &   \nonumber \\ \label{HK1}
 R_{U_1}    & ~~{\le}~~   I(U_1;Y_1|U_2,V_1)     \\ \label{HK2}
 R_{U_1}  & ~~{\le}~~   I(U_1;Y_2|U_2,V_2)    \\ \label{HK3}
 R_{U_2}   & ~~{\le}~~    I(U_2;Y_1|U_1,V_1)    \\ \label{HK4}
 R_{U_2}   & ~~{\le}~~   I(U_2;Y_2|U_1,V_2)    \\  \label{HK5}
 R_{V_1}  & ~~{\le}~~  I(V_1;Y_1|U_1,U_2)   \\ \label{HK6}
 R_{V_2}  & ~~{\le}~~   I(V_2;Y_2|U_1,U_2)    \\ \label{HK7}  
 R_{U_1}+R_{U_2}   & ~~{\le}~~ I(U_1,U_2;Y_1|V_1)    \\  \label{HK8}
 R_{U_1}+R_{U_2}   & ~~{\le}~~   I(U_1,U_2;Y_2|V_2)  \\  \label{HK9}
R_{U_1}+R_{V_1}  & ~~{\le}~~  I(U_1,V_1;Y_1|U_2)    \\ \label{HK10}
 R_{U_2}+R_{V_2}   & ~~{\le}~~   I(U_2,V_2;Y_2|U_1)   \\ \label{HK11}
R_{U_2}+R_{V_1}   & ~~{\le}~~  I(U_2,V_1;Y_1|U_1)   \\ \label{HK12}  
R_{U_1}+R_{V_2}    & ~~{\le}~~   I(U_1,V_2;Y_2|U_2)   \\  \label{HK13}
 R_{U_1}+R_{U_2}+ R_{V_1}  & ~~{\le}~~   I(U_1,U_2,V_1;Y_1)    \\ \label{HK14}
 R_{U_1}+R_{U_2} + R_{V_2}    & ~~{\le}~~   I(U_1,U_2,V_2;Y_2)   \\ \label{HK15}
 E(X_1^2)& ~~=~~   P_1  \\  \label{HK16}
 E(X_2^2)& ~~=~~   P_2.
\end{align}
Note that in achieving the upper bounding region, 
\ref{HK5} and \ref{HK6} are satisfied with equality.  
We modify the HK constraints by substituting $R_{V_1}=I(V_1;Y_1|U_1,U_2)$ and 
$R_{V_2}=I(V_2;Y_2|U_1,U_2)$. This results in, 
\begin{align} 
\label{sHK1p}
\mbox{Maximize:}~~~~~ & R_1+\mu R_2  \\
\mbox{Subject to:}~~~ &   \nonumber \\ \label{sHK1}
 R_{U_1}    & ~~{\le}~~   I(U_1;Y_1|U_2,V_1)     \\ \label{sHK2}
 R_{U_1}  & ~~{\le}~~   I(U_1;Y_2|U_2,V_2)    \\ \label{sHK3}
 R_{U_2}   & ~~{\le}~~    I(U_2;Y_1|U_1,V_1)    \\ \label{sHK4}
 R_{U_2}   & ~~{\le}~~   I(U_2;Y_2|U_1,V_2)    \\  \label{sHK5}
 R_{V_1}  & ~~{=}~~  I(V_1;Y_1|U_1,U_2)   \\ \label{sHK6}
 R_{V_2}  & ~~{=}~~   I(V_2;Y_2|U_1,U_2)    \\ \label{sHK7}  
 R_{U_1}+R_{U_2}   & ~~{\le}~~ I(U_1,U_2;Y_1|V_1)    \\  \label{sHK8}
 R_{U_1}+R_{U_2}   & ~~{\le}~~   I(U_1,U_2;Y_2|V_2)  \\  \label{sHK9}
R_{U_1}  & ~~{\le}~~  I(U_1;Y_1|U_2)    \\ \label{sHK10}
R_{U_2}  & ~~{\le}~~   I(U_2;Y_2|U_1)   \\ \label{sHK11}
R_{U_2}  & ~~{\le}~~  I(U_2;Y_1|U_1)  \\ \label{sHK12}  
R_{U_1}    & ~~{\le}~~   I(U_1;Y_2|U_2)    \\  \label{sHK13}
 R_{U_1}+R_{U_2} & ~~{\le}~~   I(U_1,U_2;Y_1)    \\ \label{sHK14}
 R_{U_1}+R_{U_2}   & ~~{\le}~~   I(U_1,U_2;Y_2)    \\ \label{sHK15}
 E(X_1^2)& ~~=~~   P_1  \\  \label{sHK16}
 E(X_2^2)& ~~=~~   P_2.
\end{align}
It follows that \ref{sHK1}, \ref{sHK2}, \ref{sHK3}, \ref{sHK4}, \ref{sHK7} and \ref{sHK8} are redundant. Note that for each of these, a similar constraint exists among the rest of the set, but with less interfering term(s) (an interference term is removed). Removing these redundant constraints, we obtain,
\begin{align} 
\label{zHK1p}
\mbox{Maximize:}~~~~~ & R_1+\mu R_2  \\
\mbox{Subject to:}~~~ &   \nonumber \\ \label{zHK1}
 R_{V_1}  & ~~{=}~~  I(V_1;Y_1|U_1,U_2)   \\ \label{zHK2}
 R_{V_2}  & ~~{=}~~   I(V_2;Y_2|U_1,U_2)    \\ \label{zHK3}  
R_{U_1}  & ~~{\le}~~  I(U_1;Y_1|U_2)    \\ \label{zHK4}
R_{U_2}  & ~~{\le}~~   I(U_2;Y_2|U_1)   \\ \label{zHK5}
R_{U_2}  & ~~{\le}~~  I(U_2;Y_1|U_1)  \\ \label{zHK6}  
R_{U_1}    & ~~{\le}~~   I(U_1;Y_2|U_2)    \\  \label{zHK7}
 R_{U_1}+R_{U_2} & ~~{\le}~~   I(U_1,U_2;Y_1)    \\ \label{zHK8}
 R_{U_1}+R_{U_2}   & ~~{\le}~~   I(U_1,U_2;Y_2)    \\ \label{zHK9}
 E(X_1^2)& ~~=~~   P_1  \\  \label{zHK10}
 E(X_2^2)& ~~=~~   P_2.
\end{align}
The constraints in \ref{zHK1} to \ref{zHK8} clearly capture the intersection of MAC1 and MAC2 in decoding $(U_1,U2,V_1)$ at $Y_1$ while considering $V_2$ as interference, and  decoding $(U_1,U2,V_2)$ at $Y_2$ while considering $V_1$ as interference. 

Noting that from \ref{eq91n}, $R_{U_1}=R^{+}_1(2)=I(U_1;Y_2|U_2)$ and from Fig.~\ref{MAC1-2d} (see also \ref{EqHKn1}), $R_1^+(1)>R^{+}_1(2)$, we conclude \ref{zHK6} is satisfied with equality and \ref{zHK3} is redundant. Also, from Fig.~\ref{MAC1-2d} (see also \ref{eq64}), $R_2^+(2)>R^{+}_2(1)$, we conclude  \ref{zHK4} is redundant.  Substituting in \ref{zHK1} to \ref{zHK9}, we obtain,
\begin{align} 
\label{zHK1pn}
\mbox{Maximize:}~~~~~ & R_1+\mu R_2  \\
\mbox{Subject to:}~~~ &   \nonumber \\ \label{zHK1n}
 R_{V_1}  & ~~{=}~~  I(V_1;Y_1|U_1,U_2)   \\ \label{zHK2n}
 R_{V_2}  & ~~{=}~~   I(V_2;Y_2|U_1,U_2)    \\ \label{zHK3n}
R_{U_1}    & ~~{=}~~   I(U_1;Y_2|U_2)    \\  \label{zHK7n}  
R_{U_2}  & ~~{\le}~~  I(U_2;Y_1|U_1)  \\ \label{zHK6n}  
 R_{U_1}+R_{U_2} & ~~{\le}~~   I(U_1,U_2;Y_1)    \\ \label{zHK8n}
 R_{U_1}+R_{U_2}   & ~~{\le}~~   I(U_1,U_2;Y_2)    \\ \label{zHK9n}
 E(X_1^2)& ~~=~~   P_1  \\  \label{zHK10n}
 E(X_2^2)& ~~=~~   P_2.
\end{align}
It is fairly straightforward to further simplify \ref{zHK1n} to \ref{zHK8n} for each of the four cases in Fig.~\ref{MAC1-2d}. Accordingly, one can find the inequalities  that are satisfied with equality in each case, which are then solved (as a set of equations) to compute $R_1+\mu R_2$.  Then, the power is optimally allocated in order to maximize the resulting $R_{ws}$.

{\bf Remark 3:} According to basics of linear programming, it is expected that four inequality constraints to emerge among the constraints~\ref{HK1} to \ref{HK14} that will be satisfied with equality, determining (as a system of linear equations) the four unknown rate values $(R_{U_1},R_{V_1}, R_{U_2}, R_{V_2})$. This means, one more inequality among \ref{zHK7n}, \ref{zHK6n} and \ref{zHK8n} should be satiated with equality. However, optimum power allocation can cause more than 4 inequality constraints to be satisfied with equality. A intuitive explanation is as follows. 
The solution to a linear programming problem, viewing from the angle of the dual problem, can be expressed as that of minimizing a linear combination of slack variables. This means, in the optimum solution of linear program in \ref{HK1p} to \ref{HK14}, the more of the  constraints are satisfied with equality,  the more of the slack variables will become zero.  This reduces the linear combination to be minimized, which is expected to happen when the power is optimally allocated. This causes some of the removed inequalities to be satisfied with equality (when the power is optimally allocated).  $\square$

{\bf Remark 4:} Although this article has focused on weak 2-users GIC, the proof for optimality of Gaussian is independent of the values of cross gains, and thereby is universally applicable to  strong, mixed and Z interference channels, as well as to GIC with more than two users. In addition, the procedure explained for the construction of boundary can be used to compute the capacity region for arbitrary cross gain values, by re-deriving various conditions that have been established based on $a<1$ and $b<1$.  $\square$
\newpage 

\begin{center}
{\bf \LARGE Appendix}
\end{center}
\begin{appendix}

\section{Intersection of MACs as an Outer Bounding Region on the Capacity of 2-users GIC} \label{APP2}

First, for the sake of self-sufficiency, a review of relevant parts from the main text will be provided. 

{\bf Review:} A primary question is, if, in GIC, public and private messages should be jointly encoded into  Gaussian $X_1$/$X_2$ code-books, or successive encoding using $X_1=U_1+V_1$ and $X_2=U_2+V_2$ can be used as an alternative.  Joint encoding is used to provide a trade-off between rates of different messages mapped to a Gaussian code-book. Joint encoding/decoding cannot be replaced with multi-layer encoding/successive decoding if there is a restriction on the number of code layers to be used. In other words, if the allowed number of code layers does not permit realizing the desired trade-off using power allocation among layers. It is also known that joint encoding/decoding can be replaced with multi-layer encoding/successive decoding by increasing the number of code-layers. In 2-users GIC, in optimizing $R_{ws}=(R_{U_1}+R_{V_1})+\mu (R_{U_2}+R_{V_2})$, there is neither a trade-off between
$R_{U_1}$, $R_{V_1}$, nor between $R_{U_2}$, $R_{V_2}$, so joint encoding/decoding is not required. In addition, as discussed in this Appendix, a joint encoding would not allow recovering $U_1$ (or $U_2$) at $Y_1$ (or $Y_2$), while considering the private message $V_2$ (or $V_1$) as noise. Such a joint encoding would be equivalent to the trivial case of adding up Gaussian code-words corresponding to public and private messages at each transmitter. The minimum required number of layers in $X_1$ as well as in $X_2$ is equal to two, because there are two types of messages for each sender, public and private.  It may be required to further divide the public messages into sub-layers to realize specific points on the sum-rate front. $\square$

As seen in Fig.~\ref{GICfig}, 2-users GIC includes two MACs, but in analyzing their corresponding rate tuples, only a subset of rate values are included in each case. The rate values that are excluded are referred to as ``do-not-care" hereafter.  In particular, the rate of private message of user 1 (user 2) is do-not-care in MAC2 (MAC1), respectively. The rate-tuple after excluding the do-not-care entries represent the projection  on the remaining sub-space. Reversing projection operation, referred to as ``lifting", is performed by including the missing rate value. Notations $\overline{M\!AC_1}$ and $\overline{M\!AC_2}$  are used to refer to the actual MAC regions, whereas notations MAC1 and MAC2 refer to the corresponding projections. Layering of the form $X_1=U_1+V_1$ and $X_2=U_2+V_2$, with $U_1$, $V_1$, $U_2$, $V_2$ independent Gaussian will be used in conjunction with $\overline{M\!AC_1}$ and $\overline{M\!AC_2}$. 

Note that in $\overline{M\!AC_1}$ composite message  $(U_1,U_2,V_1,V_2)$ should be decodable at $Y_1$ relying on an encoding/decoding procedure designed at $X_1$ and $X_2$ for achieving the MAC capacity at $Y_1$, which has a sum-rate of $0.5\log_2(P_1+aP_2+\sigma^2/\sigma^2)$. Similarly,  in $\overline{M\!AC_2}$ composite message  $(U_1,U_2,V_1,V_2)$ should be decodable at $Y_2$ relying on an encoding/decoding procedure designed at $X_1$ and $X_2$ for achieving the MAC capacity at $Y_2$, which has a sum-rate of $0.5\log_2(bP_1+P_2+\sigma^2/\sigma^2)$. Let us assume the  same power allocation (dividing of power between $U_1$, $V_1$ and between $U_2$, $V_2$) is used in $\overline{M\!AC_1}$ and $\overline{M\!AC_2}$.  Such a  power allocation can scan the entire region of $\overline{M\!AC_1}$ (or $\overline{M\!AC_2}$), but for each point on $\overline{M\!AC_1}$ (or $\overline{M\!AC_2}$), there will be a fixed set of corresponding points in $\overline{M\!AC_2}$ (or $\overline{M\!AC_1}$). 

 Let us consider a rate tuple $(R^G_{U_1},R^{G}_{U_2},R^{G}_{V_1},R^{G}_{V_2})$ belonging to the capacity region of the 2-users GIC, which is projected on point $\mathbf{P}_1$ in sup-space 
$(R_{U_1},R_{U_2},R_{V_1})$ and 
on point $\mathbf{P}_2$ in sup-space $(R_{U_1},R_{U_2},R_{V_2})$. The intersection of a MAC region with any of its (rate) axes  corresponds to the rate in a point-to-point channel without interference, and thereby has the maximum possible value. This means $\mathbf{P}_1$ and $\mathbf{P}_2$ cannot fall outside the  region MAC1 or MAC2 simply because of having one of its rate components exceeding the maximum rate of MAC1 or MAC2 along $R_1$ axis and/or along the $R_2$ axis. In conclusion, if 
$\mathbf{P}_1$ (or $\mathbf{P}_2$) falls outside the region MAC1 (or MAC2), it should surpass the partial sum-rate corresponding to MAC1 (or to MAC2).  Figure~\ref{2Dproj} depicts a simple example of a two-dimensional projection, where the shaded triangular area is the only possible location for a point, e.g., point $\mathbf{P}$, falling outside the intersection.  Let us assume $\mathbf{P}_1$ has surpassed the sum-rate in MAC1. It is easy to see that the partial sum-rate in MAC1 is minimized if $V_2$ is the last layer to be decoded at $\overline{M\!AC_1}$.
The same encoding/decoding procedures used in GIC can be applied to achieve the projected point $\mathbf{P}_1$ in MAC1. Upon decoding, $U_1$, $V_1$ can be removed because $X_1$ is independent of $X_2$. Upon removing $U_1$, $V_1$, a Gaussian code-vector remains that uniquely points to the decoded $U_2$. In view of joint typicality, this means if $U_2$ and $V_2$ are jointly encoded, the joint encoding is a trivial case obtained by adding up independent Gaussian code-vectors of rates $R_{U_2}$, $R_{V_2}$, corresponding to layering $X_2=U_2+V_2$.
Now let us remove decoded components, i.e., $U_1$, $U_2$ and $V_1$ by subtracting their corresponding code-vectors from $Y_1$. What remains would be a Gaussian vector with i.i.d. components of power $P_{V_2}+\sigma^2$.
This provides an opportunity to generate a random Gaussian code-book of power $P_{V_2}$ to send a rate of $0.5\log_2(1+P_{V_2}/\sigma^2)$ for do-not-care entry $V_2$ in $\overline{M\!AC_1}$.  Since the rate contributed by lifting is maximum for given $P_{V_2}$ and point $\mathbf{P}_1$ surpasses the partial sum-rate of 
$R_{U_1}+R_{U_2}+R_{V_1}$ in MAC1, then the lifted point should fall outside the sum-rate of $\overline{M\!AC_1}$ which is not possible. 

The conclusion is hat GIC rate region is bounded within MAC1 and MAC2, in other words, is bounded within their intersection. On the other hand, HK-region is known to be achievable in GIC and is shown in Section~\ref{sec4} to coincide with the intersection of MACs explained above.  This proves the intersection of MACs is the optimum capacity region for GIC. 
$\square$

      \begin{figure}[htp]
   \centering
   \includegraphics[width=0.35\textwidth]{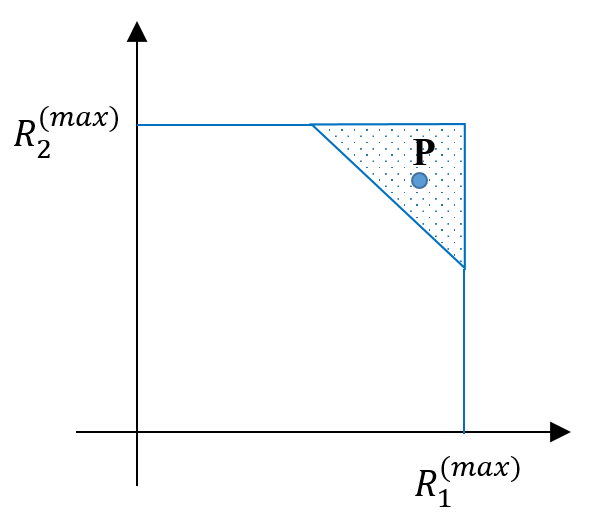}
   \caption{A simple example of a two-dimensional projection, where the shaded triangular area is the only possible location for a point falling outside the admissible region.}  
   \label{2Dproj}
 \end{figure}

\section{Superposition Coding and Successive Decoding (SC-SD) in a Point-to-Point AWGN Channel}  \label{sec1.4A}

Consider a Gaussian code-book/layer of power $\breve{P}$,  subject to an additive Gaussian noise of power $\breve{\sigma}^2$, achieving a rate (capacity) of 
\begin{equation}
\breve{R}=0.5\log_2\left(1+\frac{\breve{P}}{\breve{\sigma}^2}\right).
\label{EE2}
\end{equation} 
It is well known that, through Superposition Coding and Successive Decoding (SC-SD), the same total rate can be achieved by superimposing (adding)  two separate (continuous) Gaussian layers of powers $\gamma \breve{P}$ and $(1-\gamma)\breve{P}$, $0\leq \gamma\leq 1$, achieving rates  
\begin{equation}
\breve{R}_1=0.5\log_2\left(1+\frac{\gamma \breve{P}}{(1-\gamma) \breve{P}+\breve{\sigma}^2}\right)
\label{EE3}
\end{equation}  and 
\begin{equation}
\breve{R}_2=0.5\log_2\left(1+\frac{(1-\gamma) \breve{P}}{\breve{\sigma}^2}\right),
\label{EE4}
\end{equation} 
respectively, where $\breve{R}=\breve{R}_1+\breve{R}_2$ (see Fig. \ref{nested-awgn}). 

  \begin{figure}[htp]
   \centering
   \includegraphics[width=0.33\textwidth]{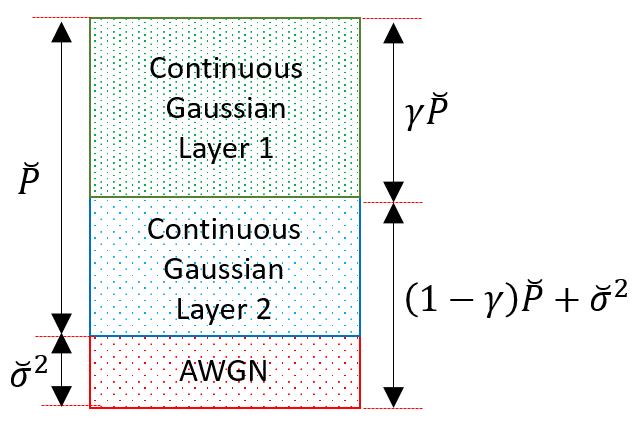}
   \caption{Superposition Coding and Successive Decoding (SC-SD) in a point-to-point AWGN channel.}
   \label{nested-awgn2}
 \end{figure}
   
The technique of SC-SD can be applied in a recursive manner and achieve the capacity of a point-to-point AWGN channel using any desired number of separate (continuous) Gaussian layers. In the limit of using layers with infinitesimal power of $\delta \breve{P}$, there will be $L=\breve{P}/\delta$ layers, where the rate of the $l$'th layer, $l=1,\cdots, L$ (indexed from top to bottom according to the order of decoding) is equal to 
\begin{equation}
\breve{R}_l=0.5\log_2\left (1+\frac{\delta \breve{P}}{(L-l)\delta \breve{P}+\breve{\sigma}^2} \right).
\label{cont}
\end{equation}
In summary, if the rates of a continuum of infinitesimal Gaussian layers, indexed from top to bottom according to the order of decoding by $l=1,\cdots, L$, satisfies expression \ref{cont}, we refer to the layer as being {\em continuous}.   Note that two successive continuous layers of powers $\gamma \breve{P}$ (with a rate given in Eq.~\ref{EE3}) and $(1-\gamma)\breve{P}$ (with a rate given in Eq.~\ref{EE4}), $0\leq \gamma\leq 1$, can be merged into a single continuous layer of power $\breve{P}$ (with a rate given in Eq.~\ref{EE2}).

\end{appendix}

\end{document}